\documentclass[a4paper,11pt,leqno]{article}
\makeatletter
\usepackage{amsfonts,delarray,amssymb,color,amsmath,amsthm,a4,a4wide}
\usepackage[mathscr]{euscript}
\usepackage[ansinew]{inputenc}

\newtheorem{theo}{Theorem}
\newtheorem{pro}{Proposition}[section]
\newtheorem{lem}[pro]{Lemma}
\newtheorem{coro}[pro]{Corollary}
\newtheorem{remark}[pro]{Remark}
\newtheorem{defi}[pro]{Definition}

\def\Xint#1{\mathchoice
   {\XXint\displaystyle\textstyle{#1}}%
   {\XXint\textstyle\scriptstyle{#1}}%
   {\XXint\scriptstyle\scriptscriptstyle{#1}}%
   {\XXint\scriptscriptstyle\scriptscriptstyle{#1}}%
   \!\int}
\def\XXint#1#2#3{{\setbox0=\hbox{$#1{#2#3}{\int}$}
     \vcenter{\hbox{$#2#3$}}\kern-.5\wd0}}
\def\dashint{\Xint-}

\newcommand{\ed}[1]{{\color{black}{#1}}}

\DeclareMathOperator{\supp}{Supp}

\DeclareMathOperator{\dist}{dist}
\DeclareMathOperator{\diam}{diam}

\def\lb{\langle}
\def\rb{\rangle}
\def\tP{{\widetilde P}}

\def\({\left(}
\def\){\right)}
\def\1{\mathbf{1}}

\def\a{{\alpha}}
\def\bm{{m_0}}

\def\bnun{{\bar \nu_n}}
\def\bgn{{\bar g_n}}
\def\bjn{{\bar \j_n}}

\def\D{\displaystyle}

\def\B{{\mathcal{B}}}

\def\curl{{\rm curl\,}}

\def\diam{\mathrm{diam}\ }
\def\div{\mathrm{div} \ }

\def\dt0{{{\frac{d}{dt}}_{|t=0}}}

\def\D{\displaystyle}

\def\E{\Sigma}

\def\ep{\varepsilon}

\def\f{{\mathbf f}}
\def\F{{F_n}}

\def\gF{{\mathbf F}}

\def\hal{\frac{1}{2}}

\def\I{{I}}

\def\indic{\mathbf{1}}
\def\indic{\mathbf{1}}

\def\j{E}

\def\jnint{{\j_n^{\rm int}}}

\def\K{K_n^\beta}
\def\cK{{\mathcal K}}

\def\LK{{\mathcal L_K}}
\def\LKk{{\mathcal L_{K,i}}}
\def\LKo{{\mathcal L_{K,0}}}
\def\Lj{{\Gamma(\j)}}
\def\Lji{{\Gamma(J_i)}}
\def\lep{{|\mathrm{log }\ \ep|}}

\def\loc{{\text{\rm loc}}}
\def\Lp{{L^p_\loc(\mr^2,\mr^2)}}

\def\l|{\left|}

\def\wb{{\overline{m}}}
\def\bw{{\underline{m}}}

\def\mc{\mathbb{C}}

\def\mn{\mathbb{N}}
\def\mn{\mathbb{N}}
\def\mr{\mathbb{R}}
\def\mo{{\mu_0}}

\def\mz{\mathbb{Z}}
\def\nab{\nabla}

\def\np{\nab^{\perp}}

\def\om{\Omega}

\def\P{{\mathscr P}}

\def\p{\partial}
\def\Q{{\mathbb{P}_n^\beta}}
\def\wQ{{\widetilde{\mathbb{P}}_n^\beta}}
\def\radon{\mathcal{M}}

\def\ro{\rho}
\def\r|{\right|}
\def\sd{\bigtriangleup}
\def\sm{\setminus}

\def\supp{\text{Supp}}

\def\Rb{{\overline R}}

\def\tq{{\frac34}}

\def\T{{\mathbb{T}}}

\def\UR{{\mathbf U_R}}

\def\vp{\varphi}

\def\vu{{\vec u}}

\def\vv{{\vec v}}

\def\w{{w_n}}
\def\W{\mathbb{W}}

\def\x{{\zeta}}

\def\z{\zeta}
\def\Z{Z_n^\beta}

 \numberwithin{equation}{section}
\title{2D Coulomb Gases and the Renormalized Energy}
\author{Etienne Sandier and Sylvia Serfaty}
\date{Second version, March 2013}
\begin{document}

\maketitle

\begin{abstract}
We  study the statistical mechanics of classical two-dimensional  ``Coulomb gases" with general potential and arbitrary $\beta$, the inverse of the temperature. Such ensembles also correspond to  random matrix models in some particular cases. The formal limit case $\beta=\infty$ corresponds to ``weighted Fekete sets" and also falls within our analysis.

 It is known that in such a system points should  be asymptotically distributed according to a macroscopic ``equilibrium measure," and that a large deviations principle holds for this, as proven by Ben Arous and Zeitouni  \cite{bz}.

 By a  suitable splitting of the Hamiltonian, we connect the problem to the ``renormalized energy" $W$, a Coulombian interaction for points in the plane   introduced in \cite{ss1}, which is expected to  be a good way of measuring the disorder of an infinite configuration of points in the plane.  By so doing, we are able to  examine the situation at the microscopic scale, and obtain several new results: a next order asymptotic expansion of the partition function, estimates on the probability  of fluctuation from the equilibrium measure  at microscale, and a large deviations type result, which states that configurations above a certain threshhold of  $W$ have exponentially small probability. When $\beta\to \infty$, the estimate becomes sharp,  showing that the system has to  ``crystallize" to a minimizer of $W$. In the case of weighted Fekete sets, this corresponds to saying that these sets should microscopically look almost everywhere like minimizers of $W$, which are conjectured to be ``Abrikosov" triangular lattices.

\end{abstract}

\noindent
{\bf keywords: } Coulomb gas, one-component plasma, random matrices, Ginibre ensemble, Fekete sets, Abrikosov lattice, triangular lattice, renormalized energy, large deviations, crystallization.
\section{Introduction}

We are interested in studying the probability law
\begin{equation}
\label{loi}
d\Q(x_1, \dots, x_n)=   \frac{1}{\Z} e^{-\frac{\beta}{2} \w (x_1, \dots, x_n)}dx_1\, \dots dx_n\end{equation} where $\Z$ is the associated  partition function (the normalizing factor such that $\Q$ is a probability measure) and
\begin{equation}
\label{wn}
\w (x_1, \dots, x_n)= -  \sum_{i \neq j} \log |x_i-x_j| +n  \sum_{i=1}^n V(x_i),\end{equation} is the Hamiltonian.
Here   the $x_i$'s belong to  $\mr^2$ (identified with  the complex plane $\mc$), $\beta>0$ is a parameter corresponding to (the inverse of) the temperature
and $V$ is a potential  satisfying some growth and regularity assumptions, which we will detail below.

The probability law $\Q$ is the Gibbs measure   of what is  called either a classical  ``two-dimensional Coulomb system"  or ``Coulomb gas"  or   ``two-dimensional one-component plasma", or sometimes ``Gaussian $\beta$-ensemble" or Dyson gas.   It was first pointed out by Wigner
\cite{wigner} and later exploited by  Dyson \cite{dyson},
 that Coulomb gases are naturally related to random matrices. This is somehow due to the
 fact that  $\exp( \sum_{i\neq j} \log |x_i-x_j|)$ is the square of the Vandermonde
 determinant $\prod_{i<j} |x_i-x_j|$ and thus the law $\Q$,  in the particular case
  when $V(x)=|x|^2$  and $\beta=2$   corresponds   to the law  of eigenvalues for   the {\it  Ginibre ensemble} (as shown in \cite{ginibre}, see also \cite{mehta}, Chap. 15), which is the set of matrices with independent normal (complex) Gaussian entries.
 For the general background and references to the literature, we refer to the book by Forrester \cite{forrester}.
These Coulomb gases  and the Ginibre ensemble have been much  studied, particularly from the random matrix point of view, but also from the statistical mechanics point of view, particularly relevant references in the physics literature are \cite{aj,sm,jlm}.

  The Gibbs measure $\Q$ can also be studied for $x_i$'s belonging to the real line (which we call the one-dimensional case).
  In the context of statistical mechanics (general $\beta$), this corresponds to ``log gases," and
in the context of random matrices ($\beta=1,2,4$), to Hermitian or symmetric random matrices (whose eigenvalues are always real). Even when $\beta \notin\{1,2,4\}$ there is a corresponding random matrix model \cite{de}, although it is more complicated. This one-dimensional case has been significantly more studied than the two-dimensional situation, and more can be achieved as well, in particular local statistics and ``universal" spacings of eigenvalues have been established \cite{vv,bey1,bey2}. This was only very recently extended to the two-dimensional case \cite{byy}.
\\
We also  apply our methods and extend them  to the one-dimensional setting, this is the object of  our companion paper \cite{ss2}.
   Let us finally also point out that studying $\Q$ with the $x_i$'s restricted to the unit circle and  with  $\beta=1,2,4$  also has a random matrix interpretation:  it corresponds to the so-called circular ensembles, e.g. in the $\beta=2$ case,  eigenvalues of the unitary matrices distributed according to the Haar measure, also a well-studied model. We also plan on examining this case in the future.

The current research on the random matrix aspect in the complex case focuses on studying the more general case of random matrices with entries that are not necessarily Gaussian and showing  that the average behavior is the same as for the Ginibre ensemble, see  \cite{bai,tv,tv2}.
We are instead  limited to exact Vandermonde factors but we emphasize that our results are
valid for all $\beta$, hence they are not limited to random matrices or the determinantal case and thus for the proof we cannot rely on any explicit random matrix model. Our results  have some universality feature in the sense that they are valid for a large class of potentials  $V$ (only a few growth and regularity \medskip assumptions are made).

The function $\w$ can also be  studied for its own sake: it can seen as  the interaction energy between similarly charged particles confined by the potential $V$. The case where $V(x)$ is  quadratic  arises for instance as the  interaction energy for superconducting vortices in the  Ginzburg-Landau theory, in the regime where their number is fixed, bounded (see \cite{livre}, Chapter 11).
In the case (not treated here) without potential $V$ but where the points are constrained to be on a manifold (such as the sphere) or a compact set, the minimizers of $\w$, or maximizers of $\prod_{i<j} |x_i-x_j|$,  are known as Fekete points or Fekete sets, cf. the book of Saff and Totik \cite{safftotik} for general reference. These are interesting in their own right -- they arise mainly in polynomial interpolation -- and the literature on the question of their distribution in various situations is vast.
 When instead $V$ is a general smooth enough function  (the situation we treat here), the  minimizers of $\w$ are   called  {\it weighted Fekete points} or weighted Fekete sets, and are also of interest, cf. again  \cite{safftotik}.

 We will pursue the analysis of these weighted Fekete sets, which  can be seen as the formal limit  $\beta\to +\infty$  of \eqref{loi}, in parallel with the analysis of \eqref{loi} for general $\beta$, and obtain new results in both cases.

\medskip

In the case of the Ginibre ensemble, i.e. when $V(x)=|x|^2$ and $\beta=1$, it is known that the ``spectral measure"  $\nu_n:= \frac{1}{n} \sum_{i=1}^n \delta_{x_i} $ converges  to the uniform  measure on the unit disc $\frac{1}{\pi} \indic_{B_1}\, dx$. More precisely $\Q$, seen as a probability on the space of probability measures on $\mc$ (the spectral measures) converges to a Dirac mass
 at $\mo=\frac1\pi \indic_{B_1}dx$.   This is the celebrated ``circular law", attributed in this case to Ginibre,  Mehta, an unpublished paper of Silverstein in 1984, and then Girko  \cite{girko1}.
 The large deviations from this  law  was established by Ben Arous and Zeitouni \cite{bz} (see Theorem \ref{thebz} below): they showed that a large deviations  principle holds with speed $n^{2}$ and rate function \begin{equation}
  \label{imubz}
  I(\mu)= \int_{\mr^2\times \mr^2} - \log |x-y|\, d\mu(x) \, d\mu(y) + \int_{\mr^2} |x|^2\, d\mu(x)\end{equation} whose unique minimizer among probabilities is of course the ``circle law" distribution $\frac{1}{\pi}\indic_{B_1}\, dx$.

For the case of a general $V$ and a general $\beta$,   the same large deviations principle holds   with  the rate function, analogue of \eqref{imubz}, being
\begin{equation}\label{Ib}
\I (\mu)= \int_{\mr^2\times \mr^2} - \log |x-y|\, d\mu(x) \, d\mu(y) + \int_{\mr^2} V(x)\, d\mu(x). \end{equation}
This can be  readapted from the proof of \cite{bz}, otherwise it is proven in possibly higher complex dimensions in \cite{berman}.
Again  the spectral measure $\nu_n= \frac{1}{n} \sum_{i=1}^n \delta_{x_i} $ converges to the minimizer among probability measures of $\I$, called the  {\it equilibrium measure}, which we will denote $\mo$.
 In the case of  weighted Fekete sets, the analogue to the circular law has been known to be  true for a much longer time: it was proved by  Fekete,  Polya and Sz\"ego that, for minimizers of
\eqref{wn},  $  \frac{1}{n} \sum_{i=1}^n \delta_{x_i} $ converges to the same equilibrium measure minimizing $I$ (then also referred to as the electrostatic interaction energy),
whose description  goes back to Gauss, and was carried out with modern mathematical rigor by Frostman \cite{frostman}.

 Fekete sets on the sphere are probably the most studied among Fekete sets (cf. \cite{sk}).
By stereographic projection of the sphere onto the plane, they can be treated as  \eqref{wn} but with a weakly confining potential  $V(x)= \log (1+|x|^2)$, which barely fails to be in the class treated in this paper. The corresponding   large deviations result (for the case with temperature) for such potentials was proven in  \cite{hardy}.

\smallskip

We are interested in examining the ``next order" behavior, or that of fluctuations around the limiting distribution $\mo$. Let us mention that such  questions have already been addressed, often with the point of view of deriving explicit scaling limits (e.g. \cite{bors,ginibre}) or laws for certain statistics  of fluctuations. One can see  for example  \cite{ahm,ahm2} where the authors essentially prove that the law of  the linear statistics of the fluctuations is a Gaussian with specific variance and mean, or also \cite{rider} for related results.  These results are however valid only for the ``determinantal case" $\beta=2$.  Our approach and results are in some sense orthogonal, and again valid for all $\beta$ and general $V$'s.

Recalling  $\I$ and $\mo$ are  found through the  large deviations at speed $n^{2}$, we look into the speed $n$ and, while  we do not prove a complete  large deviations principle  at this speed, we show  there is still a sort of rate function for which  a ``threshhold phenomenon"  holds. This analysis is based on an expansion, through a crucial but simple ``splitting formula" (which we present in Section \ref{subsecsplit} below)  of $\w(x_1,\dots,x_n)$, as  equal to $n^2 \I(\mo)-\frac{n}{2}\log n$ plus a term of order $n$,  whose prefactor tends as $n\to +\infty$ to the ``renormalized energy" $W$, a Coulombian interaction of points in the plane with a uniform neutralizing background, that we introduced in \cite{ss1} and whose definition we will recall below in Section \ref{sec1.2}. To be more precise  the limit term is the average value of $W$ on the set of blow-up limits of the configuration of points $x_1,\dots,x_n$ at the scale $1/\sqrt n$.   It is this average that partially plays the role of a rate function at speed $n$. For a precise statement, see Theorem \ref{th4}.

Another way of saying this is in the language of $\Gamma$-convergence (for a definition we refer to \cite{braides,dalmaso}, suffice it to know that this is the right notion of convergence to ensure that minimizers of $\w$ converge --  via their  empirical measures -- to minimizers of $\I$, i.e. to $\mo$):
it is not very difficult  to show (for a short proof, see \cite[Prop. 11.1]{livre}  -- the proof there is for $V$ quadratic but works with no change for general $V$) that $\frac{\w}{n^2}$ $\Gamma$-converges as $n \to \infty$ to $\I$, defined in \eqref{Ib}.
Here we examine the next order in the ``$\Gamma$-expansion" of $\w$, i.e. we study  the  $\Gamma$-convergence of $\frac{1}{n}(\w- n^2 \I(\mo)+\frac{n}{2}\log n)$, and  show that the $\Gamma$-limit is (the average of)  $W$. Consequently,  after  blow-ups at scale $\sqrt{n}$,  minimizers of $\w$ (i.e. weighted  Fekete sets) should  minimize the (average of the) renormalized energy $W$. For a precise statement, see Theorem \ref{th2}.

Before yet giving a precise definition, let us mention that we introduced the renormalized energy $W$ in \cite{ss1} for the study of the interaction of vortices in the context of the Ginzburg-Landau energy of superconductivity (for general reference on the topic, cf. \cite{livre}). Configurations that minimize the Ginzburg-Landau energy with applied magnetic field, exhibit in certain regimes ``point vortices" that are  densely packed (there are $n\gg 1 $ of them)  and are expected to  arrange themselves in perfect triangular lattices (i.e. with 60$^\circ$ angles), named {\it Abrikosov lattices} after the physicist who predicted them  \cite{abri}. The  Abrikosov lattices are indeed observed in experiments on superconductors\footnote{For photos one can see {\tt http://www.fys.uio.no/super/vortex/}}. In \cite{ss1} we made this partly rigorous by showing that minimizers of  the Ginzburg-Landau energy have vortices that minimize the renormalized energy $W$  after blow-up at the scale $\sqrt{n}$. The conjecture made in \cite{ss1},  also supported by some mathematical evidence (see Section \ref{sec1.2}),  is then that the minimal value  of $W$ is achieved by the triangular lattice ; if proven true this would completely justify why vortices form these patterns. Combining this conjecture   with the above conclusion that weighted  Fekete sets should (after blow-up) minimize $W$, we thus obtain the conjecture that they also  should locally form Abrikosov (triangular) lattices.

Does the same hold with finite temperature $\beta$, i.e. for Coulomb gases? Let us phrase the question more precisely. The law $\Q$ induces a probability measure on the family of blow-ups of $(x_1,\dots,x_n)$ around a given origin point in $\E:= \supp(\mu_0)$ --- the parameter of the family --- at the scale $\sqrt{n}$, a blow-up scale after which the resulting  points are typically separated by order $1$ distances. In the limit $n\to \infty$  this yields a probability measure on the set of configurations of points in the plane and we may ask if, almost surely, the  blow-up configurations minimize $W$.  Our results indicate  that this is not the case,  however we are able to prove that  there is a treshhold phenomenon, in the sense  that except with exponentially small probability, the average of $W$ is below  a certain constant, itself converging to the minimum of $W$ as $\beta\to \infty$, which indicates {\it crystallisation}, i.e. if the above conjecture is true, we should see Abrikosov lattices as $\beta\to\infty$.

To our knowledge, this is the first time Coulomb gases or Fekete sets are  rigorously connected to  triangular lattices, in agreement with  predictions in the physics literature (see  \cite{aj} and references therein).

\medskip

A  corollary of our  way of expanding $\w$  is that we obtain a next order estimate of the partition function $\Z$, a result we can already state:

\begin{theo} \label{valz}Let $V$ satisfy assumptions \eqref{assumpV1} -- \eqref{assumpV4} below. There exist functions $f_1, f_2$ depending only on $V$, such that for any $\beta_0>0$ and any $\beta \ge\beta_0$, and for $n$ larger than some $n_0$ depending on $\beta_0$, we have
\begin{equation}\label{1.5}
n\beta f_1(\beta) \le \log \Z -\(- \frac{\beta}{2} n^2 \I(\mo) + \frac{\beta n}{4} \log n\)\le  n \beta f_2(\beta),\end{equation}with $f_1,f_2$ bounded in $[\beta_0,+\infty)$ and such that
\begin{equation}
\lim_{\beta\to\infty} f_1(\beta)=\lim_{\beta\to\infty}f_2(\beta)= - \frac{\a}{2}\end{equation}
where $\a$ is some constant related to $W$, and explicited in \eqref{defa} below.
\end{theo}
This improves on the known results, which only gave the expansion $\log \Z \sim \frac{ \beta}{2} n^2 I(\mo)$. It also seems to contradict the result of the calculations of \cite{zw}.
Let us recall that an exact value for $\Z$ is only known for the Ginibre ensemble case of $\beta =2$ and $V(x)=|x|^2$: it is
$Z_n^2= n^{-\hal n(n+1)} \pi^n \prod_{k=1}^n k!$ (see \cite[Chap. 15]{mehta}). Known asymptotics allow to deduce (cf. \cite[eq. (4.184)]{forrester}):
\begin{equation}\label{vallogz}
\log Z_n^2 =- \frac{3n^2}{4}+ \frac{n}{2} \log n  + n ( -1 +    \hal \log 2 + \frac{3}{2}\log \pi )    + O(\log n)
\quad \text{as } \ n\to \infty,\end{equation} where we note the value of $I(\mu_0)$ is indeed $\frac34$ for  this potential.
On the other hand,  no exact formula exists for general potentials \footnote{an exception is the result of \cite{dgil} for a quadrupole potential},  nor for quadratic potentials if $\beta \neq 2$. This is in contrast with the one-dimensional situation for which, at least in the case of quadratic $V$, $\Z$ has an explicit expression for every $\beta$, given by the famous Selberg integral formulas (see e.g. \cite{agz}).

In statistical mechanics language, the existence of an exact asymptotic expansion up to order $n$ for $\log \Z$
is essentially the existence of a thermodynamic limit.
This is established in \cite{ln} for a three-dimensional Coulomb system,  and in a nonrigorous way  in \cite{sm} in two dimensions. The existence of the thermodynamic limit here remains to be completed by  getting upper and lower bounds which match up to $o(n)$ in \eqref{1.5}.

\medskip

As suggested by the strong analogy between  Coulomb  gases and interacting  vortices in the Ginzburg-Landau model, we will draw heavily on methods we  introduced in \cite{ss1}, such as the splitting, blow-up, the use of the ergodic theorem, and the properties of the renormalized energy.

\medskip

The rest of the introduction is organized as follows: first we give some more notation, give the assumptions we need to make on $V$  and state the splitting formula, then we present the definition of the renormalized energy and the main results from \cite{ss1} that we will use, and finally we  state our main results and comment on them.

\subsection{The equilibrium measure and the splitting formula}\label{subsecsplit}
We need to introduce some notation, and for this we need to describe the equilibrium measure $\mo$ minimizing \eqref{Ib}  among probability measures.

This description, which is now classical  in potential theory (see  \cite[Chap. 1]{safftotik}) says  that, provided $\lim_{|x|\to +\infty} \frac{V(x)}{2} - \log|x| = +\infty$ and $\log V$ is lower semicontinuous,   this equilibrium measure exists, is unique, and  is characterized by the fact that
there exists a constant $c$ such that, quasi-everywhere,
\begin{equation}\label{optcondmo}
U^\mo + \frac{V}{2}\ge c \quad \text{ and}  \quad U^\mo + \frac{V}{2}= c\quad \text{on  $\supp(\mo)$},\end{equation}
where  for any measure $\mu$, $U^\mu$ denotes the potential generated by  $\mu$, defined by
\begin{equation}\label{umu}
U^\mu(x)= - 2\pi \Delta^{-1}\mu:= - \int_{\mr^2} \log |x-y|\, d\mu(y).\end{equation}
Here and in all the paper, we denote by $\Delta^{-1}$ the operator of convolution by $\frac{1}{2\pi}\log |\cdot|$. It is such that $\Delta \circ \Delta^{-1}=Id$, where $\Delta $ is the usual Laplacian. We denote the support of $\mo$ by $\E$.

Another way to characterize  $U^\mo$ is as the solution of the following {\it obstacle problem}\footnote{The obstacle problem is a free-boundary problem and a much-studied   classical  problem in the calculus of variations,  for general reference see \cite{friedman,ks}.} : It is a superharmonic function bounded below by $c-V/2$ and harmonic outside the so-called {\em coincidence set}
\begin{equation}\label{coincidence}\E=\{U^\mo = c-V/2\}.\end{equation}
This implies in particular that $U^\mo$ is  $C^{1,1}$ if $V$ is \smallskip (see \cite{caff}).

It is now a good time to state the assumptions on $V$ that we assume are satisfied in the sequel.
\begin{equation}\label{assumpV1}  \lim_{|x|\to +\infty} \frac{V(x)}{2} - \log|x| = +\infty, \end{equation}
\begin{equation}\label{assumpV2} \text{ $V$ is $C^3$ and there exists $\bw,\wb>0$  s.t. $\bw\le \frac{\Delta V}{4\pi} \le \wb$,}\end{equation}
\begin{equation}\label{assumpV3} \text{$V$ is such that $\p \E$ is $C^{1}$,}    \end{equation}
\ed{\begin{equation}
\label{assumpV4}
\text{there exists} \ \beta_1>0\ \text{such that } \ \int_{\mr^2 \backslash B_1}
e^{-\beta_1 ( V/2(x)- \log |x|) } \, dx<+\infty.
\end{equation}}

The assumption \eqref{assumpV1} on the growth of $V$ is what is needed to apply the results from \cite{safftotik} and to guarantee that \eqref{Ib} has a minimizer. The other conditions  are technical and certainly not optimal, they are meant to ensure that $\mu_0$ and its support are regular enough and that $\mo$ never degenerates, which we will need for example when making explicit constructions.
Indeed, assumptions \eqref{assumpV2}--\eqref{assumpV3}, together with \eqref{optcondmo}, \eqref{umu}  and the regularity of $V$, ensure that
\begin{equation}\label{112}   d\mo= \bm(x)\,dx,\quad\text{where}\quad \bm(x) = \frac{\Delta V(x)}{  4\pi }\indic_{\E}(x).\end{equation} hence for the $\bw,\wb>0$  of \eqref{assumpV2} we have
\begin{equation}\label{bornmo}   \bw\le \bm\le \wb.\end{equation}
\ed{Assumption \eqref{assumpV4} is a supplementary assumption on the growth of $V$ at infinity, needed for the case with temperature, to show that the partition function is well-defined. It only requires slightly more than \eqref{assumpV1}, for example 
assuming $\lim_{|x|\to \infty} \frac{V}{2}(x) - (1+\ep) \log |x|= + \infty$ for some $\ep>0$ suffices.}

Next, we set  $\z= U^\mo + \frac{V}{2} - c$ where $c$ is the constant in \eqref{optcondmo} and \eqref{coincidence}. This function satisfies
\begin{equation}\label{eqz}
\left\{\begin{array}{rlll}
\Delta \z &=& \frac{1}{2}\Delta V\ \indic_{\mr^2 \sm \E} &
\\ \z &=& 0 & \text{quasi-everywhere in $\E$}\\
\z&>& 0 & \text{quasi-everywhere in $\mr^2\sm \E$}\end{array}\right.\end{equation}
From our  assumption \eqref{assumpV2}, \ed{ it follows from \cite[Lemma 5]{caff} that there exists $\kappa>0$ such that for every $x\in \mr^2$,
\begin{equation}\label{lemz}
\z(x) \ge \kappa\, \dist(x, \E)^2, \end{equation}
 and such a  rate is  in fact optimal \cite[Lemma 2]{caff}.}
All the quantities introduced so far: $\mu_0$, $\E$, $\zeta$, depend only on the data of $V$.

\medskip

The function $\zeta$ arises in the splitting formula for $\w$ which we now present.
As mentioned above, expanding  the probability density to the next order goes along with blowing-up the point configuration  by a factor  $\sqrt{n}$. We then denote the blown-up quantities by primes. For example the {\it blown-up coordinates} $x_i'= \sqrt{n} x_i$, $\bm'(x') = \bm(x)$, $d\mu_0'(x')= \bm'(x')dx'$ etc \ldots.

The splitting formula, proven in Section \ref{secproofsplit}, is the observation that, for any $x_1, \dots, x_n\in \mr^2$, we have
\begin{equation}\label{idwnbu}
\w(x_1, \dots, x_n)= n^2 \I (\mo)- \frac{n}{2} \log n + \frac{1}{\pi} W(\nab H'_n, \indic_{\mr^2} )  +2 n \sum_{i=1}^n  \z (x_i),\end{equation}
where
\begin{equation}\label{Hp}
H_n' := - 2\pi \Delta^{-1} \(\sum_{i=1}^n \delta_{x_i'} - \ed{ \mu_0'} \), \end{equation}
and where, in agreement with  formula \eqref{WR} below (the existence of the limit as $\eta\to 0$ will also be discussed there),
\begin{equation}
W(\nab H_n', \indic_{\mr^2}):=\lim_{\eta\to 0} \( \frac12\int_{\mr^2\sm\cup_{i=1}^n B(x_i', \eta)} |\nab H'_n|^2
+ \pi n \log \eta\).
\end{equation}
The function $H_n'$ physically corresponds to the electrostatic potential generated by the positive point charges $\sum_i \delta_{x_i'}$ and the diffuse negative charge $-\mu_0'$. Its opposite gradient, that we will denote by $E_n$ physically corresponds to the electric field generated by the charges (hence the notation).

Letting, for a measure $\nu$
\begin{equation}\label{defF}
\F(\nu):=
 \begin{cases}
\D \frac{1}{n}\( \frac{1}{\pi} W(\nab H_n', \indic_{\mr^2}) + 2n\int  \z \, d\nu\) & \ \text{if } \ \nu \ \text{is of the form}  \sum_{i=1}^n \delta_{x_i}\\
 + \infty & \ \text{otherwise,}
 \end{cases}
\end{equation}the relation \eqref{idwnbu} can be rewritten
\begin{equation}\label{lienfw}
\w(x_1, \dots, x_n)=  n^2 \I (\mo)- \frac{n}{2} \log n + n \F\(\sum_{i=1}^n \delta_{x_i}\).\end{equation}
This allows to separate orders  as announced since we will see that  $\F(\sum_{i=1}^n  \delta_{x_i} )$ is typically of order $1$.

We may next cancel out   leading order terms and  rewrite the probability law \eqref{loi} as
\begin{equation}\label{loi2}
d\Q(x_1, \dots, x_n) =\frac{1}{ \K} e^{- n\frac{\beta}{2} \F(\sum_i\delta_{x_i} )} \, dx_1 \dots dx_n\end{equation}
where
\begin{equation}\label{defK}
\K:= \Z e^{\frac{ \beta}{2}( n^2 \I (\mo)- \frac{n}{2} \log n)      }.
\end{equation}
As we will see below   $\log \K$ is of order $n\beta$,  which leads to Theorem \ref{valz}.

We will also denote
\begin{equation}
\label{fnhat}
\widehat{\F}(\nu)= \F(\nu)- 2\int  \z \, d\nu=\frac{1}{n\pi} W(\nab H_n', \indic_{\mr^2}).
\end{equation}
\ed{
In view of \eqref{idwnbu} the main task in our proof is to pass to the limit $n\to \infty$ in $W(\nab H_n', \indic_{\mr^2})$ and obtain a limiting energy, which will be (the average of) our Coulomb renormalized energy $W$. Passing to the limit in \eqref{Hp} will lead to solutions of $-\Delta H= 2\pi \(\sum_p \delta_p -cste\)$ where the sum is now infinite. The limit energy thus has to be defined on objects of this form, or equivalently (by taking $E=-\nab H$) solutions of $\div E = 2\pi \(\sum_p \delta_p -cste\)$, $\curl E=0$. The definition will be given just below. The passage to the limit is not obvious, for several reasons. The first is the lack of local charge neutrality, and the  fact that the energy density associated to $W(\nab H_n', \indic_{\mr^2})$ is not pointwise bounded below. The second is the need of the ``averaged formulation" alluded to above, this will be provided by an abstract method relying on the ergodic theorem, and inspired by Varadhan.

}

\subsection{The renormalized energy}\label{sec1.2}
We now define precisely the ``renormalized energy" $W$ introduced in \cite{ss1}, which is a way of computing the Coulomb interaction between an infinite number of point charges in the plane with a uniform neutralizing background of density $m$. We point out that, to our knowledge, each of the analogous Coulomb systems studied in the physics literature (e.g. \cite{sm,aj}) comprise a finite number of point charges, and hence implicitly extend only to a bounded domain on which there is charge neutrality.  Here we do not assume any local charge neutrality.

We denote by $B(x,R)$ or $B_R(x)$ the ball centered at $x$ with radius $R$ and let $B_R=B(0,R)$.
\ed{In all the paper, when $U$ is a measurable set, $|U|$ will denote its Lebesgue measure, and when $U$ is a finite set, $\# U$ will denote its cardinal. $\dashint $ will denote an integral average.

The point of the definition of $W$ below is that we would like to define $W(\nab H)$ for $H$ solving $-\Delta H=2\pi(\sum_p \delta_p-m)$ as $\limsup_{R\to \infty} \dashint_{B_R} |\nab H|^2$, however these integrals diverge because of the logarithmic divergence of $H$ near each point. Instead, we compute $\int |\nab H|^2$ in a ``renormalized" way or in ``finite parts", by cutting out holes around each $p$ and subtracting off the corresponding divergence, in the manner of \cite{bbh}, from which the name ``renormalized energy" is borrowed.
}

\begin{defi}Let $m$ be a nonnegative number. For any continuous  function $\chi$ and any vector-field $\j$ in $\mr^2$  such that
\begin{equation}\label{eqj}
\div  \j = 2\pi (\nu - m) , \qquad \curl \j =0\end{equation} where
 $\nu$ has the form

\begin{equation} \label{eqnnu}\nu=  \sum_{p \in\Lambda} \delta_{p}\quad \text{ for
some discrete set} \  \Lambda\subset\mr^2,\end{equation}
 we let
\begin{equation}\label{WR}W(\j, \chi) = \lim_{\eta\to 0} \(
\hal\int_{\mr^2 \backslash \cup_{p\in\Lambda} B(p,\eta) }\chi
|\j|^2 +  \pi \log \eta \sum_{p\in\Lambda} \chi (p) \).
\end{equation}\end{defi}
\ed{
To see that the limit $\eta\to 0$ exists, it suffices to  observe that in view of \eqref{eqj}--\eqref{eqnnu},  $\j$ is a gradient and near each $p\in \Lambda$ we may write  $\j =  \nab \log |\cdot - p| +\nab f(\cdot)$ where $f$  is $C^1$ by elliptic regularity.  The limit follows easily. It also follows that $\j$ belongs to $L^q_\loc$ for any $q<2$.
}

\begin{defi}\label{defA} Let $m$ be a nonnegative number.
Let  $\j$
be a vector field in $\mr^2$. We say $\j$
belongs to the admissible class
$\mathcal{A}_m $
 if \eqref{eqj}, \eqref{eqnnu}  hold  
and
\begin{equation}\label{densbornee}
\frac{    \nu (B_R )  } {|B_R|}\quad \text{ is bounded by a constant independent of $R>1$}.
\end{equation}
 \end{defi}

In the sequel $K_R$ will denote the two-dimensional squares $[-R,R]^2$.
We also   use the
notation $\chi_{K_R}$ for positive  cutoff functions satisfying, for some
constant $C$ independent of $R$,
\begin{equation}
\label{defchi} |\nab \chi_{K_R}|\le C, \quad \supp(\chi_{K_R})
\subset K_R, \quad \chi_{K_R}(x)=1 \ \text{if } d(x, {K_R}^c) \ge
1.\end{equation}

\begin{defi}The renormalized energy $W$ is defined,
for $\j \in \mathcal{A}_m$,  by
\begin{equation} \label{Wroi} W(\j)= \limsup_{R \to \infty}
\frac{W(\j, \chi_{K_R})}{|K_R|} ,
\end{equation} with $\{\chi_{K_R}\}_R$ satisfying \eqref{defchi}.
\end{defi}
\ed{We note that we have taken a slightly different definition from \cite{ss1}: first the vector-fields in \eqref{eqj} have been rotated by $\pi/2$, second
$\mathcal{A}_m $ here corresponds to $\mathcal{A}_{2\pi m}$ in \cite{ss1}. Finally, in \cite{ss1} we presented the definition with averages over general sets, here we have chosen for simplicity to introduce it only with square averages.}


In theory, many different $\j$'s could correspond to a given $\nu$ (one can always add the gradient of a harmonic function). But as it turns out, they only differ by a constant:
\begin{lem}\label{lem15} Let $m\ge 0$ and $\nu =  \sum_{p \in \Lambda}\delta_p $, where $\Lambda\subset\mr^2$ is discrete, and assume there exists $\j$ such that \begin{equation}\label{jfini}
\div  \j = 2\pi (\nu - m),\quad \curl \j =0,\quad\text{and}\quad W(\j)<+\infty.
\end{equation}
Then any other $\j'$ satisfying \eqref{jfini} is such that $\j-\j'$ is constant.

If there exists $\j$ such that \eqref{jfini} holds and such that
\begin{equation}\label{jnu} \lim_{R \to \infty} \dashint_{K_R} \j = 0,\end{equation}
then any other $\j'$ satisfying \eqref{jfini} is such that $W(\j')>W(\j)$.
\end{lem}

\begin{proof}\ed{ Let  $\j$, $\j'$ be as above. We may view them as complex functions of a complex variable.  From \eqref{jfini} we have $\div (\j-\j')=\curl (\j-\j')=0$ and thus $\j-\j'$ is  holomorphic. We can write it as a power series $\sum_{n=0}^\infty a_n z^n$ with infinite radius of convergence. On the other hand, from the finiteness of $W(\j)$ and $W(\j')$ we deduce easily that there exists $C>0$ such that  \begin{equation}\label{precau}\forall R>1, \quad
\int_{K_R} |\j-\j'|^2\le CR^2.
\end{equation}  But by Cauchy's formula we have, for any $R>0$ and $t\in[R,R+1]$
$$a_n= \frac{1}{2i\pi} \int_{\p B(0,t)} \frac{(\j-\j')(z)}{z^{n+1}} \, dz=\frac{1}{2i\pi}\int_{R}^{R+1} \int_{\p B(0,t)} \frac{(\j-\j')(z)}{z^{n+1}} \, dz\, dt. $$
It follows that
$$|a_n|\le =\frac{1}{2\pi R^{n+1}} \int_{B(0,R+1)\backslash B(0,R)}|\j-\j'|
\le \frac{C}{ R^{n+1}} R^{3/2}$$
where we have used the Cauchy-Schwarz inequality and \eqref{precau}.
Letting $R\to \infty$ we find that $a_n=0$ for any $n \ge 1$ and thus
$\j-\j'$ is constant.}
For the second statement, we deduce from the first statement that $\j' = \j + \vec{C}$ for some constant vector $\vec{C}\neq 0$, and then
$$W(\j',\chi_{K_R}) = W(\j,\chi_{K_R}) + \vec{C} \cdot\int \j\chi_{K_R} + \frac{|c|^2}{2} \int \chi_{K_R}, $$
so that dividing by $|K_R|$, passing to the limit as $R\to +\infty$ and in view of \eqref{jnu}, we find $W(\j') = W(\j) +\hal |\vec{C}|^2.$ \end{proof}

Note that given  $\nu$, the above lemma shows that either for all $\j$'s satisfying \eqref{jfini} the limit $\lim_{R \to \infty} \dashint_{K_R} \j$ exists, or it exists for none of them. Both cases may occur.





The following additional facts and remarks  about $W$ are mostly from \cite{ss1}:
\begin{itemize}
\item[-] \ed{ In \cite{ss1}, we introduced $W$ as  being computed with averages over general shapes (say balls, squares etc). We showed that the minimum of $W$  over $\mathcal{A}_m$ does not depend on the shape used. Since squares are the most useful ones, we restricted to them here for the sake of simplicity.}


\item[-] It was shown in \cite[Theorem 1]{ss1} that the value  of $W$ does not depend on the choice of $\{\chi_{K_R}\}_R$
 as long as  it satisfies \eqref{defchi}.


\item[-] $W$ is bounded below and admits  a minimizer over $\mathcal{A}_1$, cf. \cite[Theorem 1]{ss1}.
\item[-] It is easy to check that if $\j$ belongs to $\mathcal{A}_m$, $m>0$,  then
$\j'=\frac{1}{\sqrt{m}} \j(\cdot /\sqrt m)$ belongs to $\mathcal{A}_1 $ and
\begin{equation}\label{minalai1}
W(\j)= m \(W(\j') - \frac{\pi}{2} \log m\).\end{equation} Consequently  if $\j$
 is a minimizer of  $W$ over $\mathcal{A}_m$, then $\j'$ minimizes  $W$ over $\mathcal{A}_1$. In particular
 \begin{equation}\label{minalai2}
\min_{\mathcal A_m} W= m \(\min_{\mathcal A_1} W - \frac{\pi}{2} \log m\).\end{equation}

\item[-] Because the number of points is  in general infinite,
 the interaction over large balls needs to be normalized
 by the volume, as in a thermodynamic limit. Thus  $W$ does not feel compact perturbations of the configuration of points. Even though the interactions are long-range, this is not difficult to justify rigorously.

\item[-]In \cite{gs} some necessary and some sufficient conditions on the configuration of points for which $W(\j)<\infty$ are given.
\ed{
\item[-] We may define $W$ as a function of the point measure $\nu$  only, by setting for every $\nu$ satisfying \eqref{eqnnu}
\begin{equation}\label{neww}
\mathbb{W}(\nu)= \inf_{\j \ \text{such that } \eqref{eqj} \ \text{holds}} W(\j),\end{equation}
and $\mathbb{W}(\nu)=+\infty$ if $\nu$ is not of the form $\sum_{p\in \Lambda} \delta_p$.
This definition is somehow ``relaxed" since $\mathbb{W}(\nu) \le W(\j)$ for any $\j$ satisfying \eqref{eqj}. The main point to check is the measurability of $\W$, which we will discuss below in Section \ref{secnewdefw}.
}
\item[-] In the case  $m=1$ and when  the set of points $\Lambda$ is  periodic with respect to some lattice $
\mz \vu+ \mz \vv$ then  it can be viewed as a set of $n$ points $a_1, \dots , a_n$ over the torus
$\T_{(\vu,\vv)}:= \mr^2/(\mz \vu +\mz\vv)$ with $  |\T_{(\vu,\vv)}|=n$.  In this case,  the infimum of $W(\j)$ among $\j$'s which satisfy \eqref{jfini}  is achieved by  $\j_{\{a_i\}} = -\nab h$, where $h$ is the periodic solution to $-\Delta h = 2\pi (\sum_i\delta_{a_i} - 1)$, and
\begin{equation}\label{Wcasper}
W(\j_{\{a_i\}})= \frac{\pi}{|\T_{(\vu,\vv)}|}  \sum_{i \neq j} G(a_i-a_j) +  \pi \lim_{x\to 0} \(G(x)+\log |x|\) \end{equation}
where  $G$ is the Green function of the torus with respect to its volume form, i.e. the solution to
$$- \Delta G(x) = 2\pi \(\delta_0 - \frac{1}{|\T_{(\vu,\vv)}|}\)    \quad \text{in } \ \T_{(\vu, \vv)}.$$
An explicit expression for $G$ can be found via  Fourier series  and this leads to an explicit expression for $W$ of the form $\sum_{i\neq j} \mathrm{E}(a_i-a_j)$ where $\mathrm{E}$ is an Eisenstein series (for more details see \cite[Lemma 1.3]{ss1} and also \cite{bs}). In this periodic setting, the expression of $W$ is thus much simpler than \eqref{Wroi} and reduces to the computation of a sum of explicit pairwise interaction.

\item[-] When the set of points $\Lambda $ is itself exactly a lattice $\mz \vu + \mz \vv$ then $W$ can be expressed explicitly through the Epstein Zeta function of the lattice. Moreover, using results from number theory, it is proved in \cite[Theorem 2]{ss1},  that  the unique minimizer of $W$ over lattice configurations of fixed volume is the triangular lattice. This supports the conjecture that the Abrikosov triangular lattice is a global minimizer of $W$, with  a slight abuse of language since $W$ is here not a function of the points, but of their associated ``electric fields" $\j_{\{a_i\}}$.
\end{itemize}
This last fact allows us to think of $W$ as a way of measuring the disorder  and lack of  homogeneity of a configuration of points in the plane (this point of view is pursued in \cite{bs} with explicit computations for random point processes). Another way to see it is to view $W$ as measuring the distance between $\sum_{p\in\Lambda}  \delta_p $ and the constant $m$ in $H^{-1}$, the dual space to the Sobolev space $H^1_0$ (with $\|f\|_{H^1_0}=\|\nab f\|_{L^2}$) which only makes sense modulo the ``renormalization" as $\eta \to 0$ and modulo normalizing by the volume.

We may now define  the constant $\a$ which appears in Theorem \ref{valz} and in Theorem~\ref{th2} below:
\begin{equation}\label{defa}
\a:= \frac1\pi \int_\E \min_{\mathcal A_{\bm(x)}} W \,dx = \frac{1}{\pi}\min_{\mathcal{A}_1} W-\hal \int_\E \bm(x)\log\bm(x)\, dx,\end{equation}
where we have used  \eqref{minalai2} and the fact that, from \eqref{112},  $\int_\E\bm=1$. Note that $\a$ only depends on $V$, via the integral term, and on the (so far) unknown constant $\min_{\mathcal{A}_1} W$.

\subsection{Statement of main results}

Our first  result identifies the $\Gamma$-limit of $\{F_n\}_n$, defined in  \eqref{defF} or \eqref{lienfw}. This in particular allows a description at the microscopic level of  the weighted Fekete sets minimizing $\w$. Below we abuse notation by writing $\nu_n= \sum_{i=1}^n \delta_{x_i}$ when it should be $\nu_n= \sum_{i=1}^n \delta_{x_{i,n}}$.
For such a $\nu_n$, we let $\nu_n' = \sum_{i=1}^n\delta_{x_i'}$
be the measure in blown-up coordinates and  $\j_{n} = -
\nab  H'_n$ be the associated electric field, where $H'_n$ is defined by \eqref{Hp} --- equivalently $\j_{n}$
 is the solution of $\curl \j_{n} = 0$, $\div \j_{n} = 2\pi(\nu_n' - \bm' dx')$
 in $\mr^2$ which tends to $0$ at infinity.
 (To avoid confusion, we emphasize here that $\nu_n$ lives at the original scale while $\j_{n}$ lives at the blown-up scale and that $m_0'$ is the blown-up density of the equilibrium measure $\mu_0$.)
 We also let
\begin{equation}\label{pnun}P_{\nu_n} = \dashint_{\E} \delta_{(x,\j_{n}(\sqrt{n}x+\cdot))}\,dx,\end{equation}
i.e. the push-forward of the normalized Lebesgue measure on $\E$ by $x \mapsto (x, \j_{n} (\sqrt{n}x+\cdot)).$
It is a probability measure on $X:=\E \times L^p_\loc(\mr^2,\mr^2)$ (couples of (blow-up centers, blown-up current around this center)).
\ed{We emphasize that $P_{\nu_n}$ is a probability measure which has nothing to do with $\Q$. Each realization or configuration $(x_1, \dots, x_n)$ gives rise in a deterministic fashion to its $P_{\nu_n}$, which encodes all the blown-up profiles of associated electric fields. We denote by $i_n$ this mapping (or embedding)
\begin{equation}
\label{in}
\begin{array}{ll}
i_n: & \mc^n \to \P(X)
\\
&
(x_1, \dots, x_n) \mapsto P_{\nu_n}\end{array}
\end{equation}
where $\P(X)$ denotes the space of probability measures on $X= \E\times \Lp$. We view $\P(X)$ as endowed with the topology of weak convergence of probabilities.
}

The limiting object as $n\to +\infty$ in the $\Gamma$-limit of $w_n$ was the limit of $\frac{\nu_n}{n}$. In taking the  $\Gamma$-limit of $F_n$, the limiting object is more complex, it is the limit  $P$  of  $P_{\nu_n}$. This is a probability measure on all blown-up electric fields obtained from a given $(x_1, \dots, x_n)$. Thus  it is like a  Young measure akin to the Young measures on micropatterns introduced in \cite{am}.

We will  here and below use the notation
\begin{equation} \label{dxr} D(x',R)= \nu_n \(B\(x, \frac R{\sqrt{n}}\) \) -n \mo \( B\(x, \frac R{\sqrt{n}} \)\), \end{equation}
where $x' = \sqrt n x$ as usual, to denote the fluctuations of the number of points in a microscopic ball of radius $R$.
Note that $\widehat{\F}$ was defined in \eqref{fnhat} and the result below is slightly stronger than the $\Gamma$-convergence of $\F$ since $\F \ge\widehat{\F}$.


\begin{theo}[Microscopic behavior of weighted Fekete sets] \label{th2}  Let the potential $V$ satisfy assumptions \eqref{assumpV1}--\eqref{assumpV3}. Let $m_0$ be the density of the equilibrium measure $\mu_0$. Fix from now on $1<p<2$ and let $X= \E \times L^p_\loc(\mr^2,\mr^2)$.\medskip

{\bf A. Lower bound.}\ Let  $\nu_n= \sum_{i=1}^n \delta_{x_i}$ be a sequence such that $\widehat{ F_n}(\nu_n)  \le C$.
Then $P_{\nu_n}$ defined by \eqref{pnun} is a probability measure on $X$ and

\begin{enumerate}
\item Any subsequence of $\{P_{\nu_n}\}_n$ has a convergent subsequence converging to some $P\in \P(X)$ as  $n \to \infty$.
\item The first marginal of $P$ is the normalized Lebesgue measure on $\E$. $P$ is invariant  by $(x,\j)\mapsto\(x,\j(\lambda(x)+\cdot)\)$, for any $\lambda(x)$ of class $C^1$  from $\E$ to $\mr^2$ (we will say {\em  $T_{\lambda(x)}$-invariant}).
\item For $P$ almost every $(x,\j)$ we have $\j\in\mathcal A_{\bm(x)}$.
\item Defining $\a$ as in \eqref{defa}, it holds that
\begin{equation}\label{thlow}\liminf_{n \to \infty} \widehat{F_n}(\nu_n) \ge \frac{|\E|}{\pi}\int W(\j) \, dP(x,\j)\ge\a .\end{equation}
\end{enumerate}

{\bf B. Upper bound construction.}\ Conversely, assume $P$ is a $T_{\lambda(x)}$-invariant probability measure on $X$ whose first marginal is $\frac{1}{|\E|}dx_{|\E}$ and such that for $P$-almost every $(x,\j)$ we have $\j\in\mathcal A_{\bm(x)}$. Then there exists a sequence $\{\nu_n= \sum_{i=1}^n \delta_{x_i}\}_n$ of empirical measures on $\E$ and a sequence $\{\j_n\}_n$ in $L^p_\loc(\mr^2,\mr^2)$ such that $\div \j_n = 2\pi(\nu_n' -\bm')$ and   such that defining $P_n$ as in \eqref{pnun}, we have $P_n\to P$ as $n \to \infty$  and
\begin{equation}\label{thup}\limsup_{n \to \infty} F_n(\nu_n)  \le \frac{|\E|}{\pi}\int W(\j) \, dP(x,\j). \end{equation}
 \medskip

{\bf C. Consequences for minimizers.}  If $(x_1,\dots, x_n)$ minimizes $w_n$ for every $n$ and  $\nu_n=\sum_{i=1}^n \delta_{x_i}$, then  the limit $P$ of $P_{\nu_n}$ as defined in \eqref{pnun} satisfies the following.
\begin{enumerate}
\item For $P$-almost every $(x,\j)$,  $\j$ minimizes $W$ over $\mathcal{A}_{\bm(x)}.$
\item We have $$ \lim_{n \to \infty} F_n(\nu_n) =\lim_{n\to \infty} \widehat{\F}(\nu_n)= \frac{|\E|}{\pi}\int W(\j) \, dP(x,\j)= \a, \qquad \lim_{n\to \infty}\sum_{i=1}^n \dist^2(x_i,\E)=0.$$
\item There exists $C>0$ such that  for every $x'\in \mr^2$,
$R>1$ and using the notation \eqref{dxr} we have
\begin{equation}\label{dxr2}
D(x', R)^2 \min \(1, \frac{|D(x',R)|}{R^2}\)\le C n.\end{equation}
\end{enumerate}
\end{theo}
Note that part~B of the theorem is only a partial converse to part~A because the constructed $\j_n$ need not be curl  free.
\ed{
\begin{remark}\label{rem15}
Defining $Q_{\nu_n} = \dashint_{\E} \delta_{x, \nu_n'(\sqrt{n} x+\cdot))}, dx$, or equivalently  as the push-forward of $P_{\nu_n}$ by the map $(x,\j) \mapsto  \frac{1}{2\pi} \div \j + m_0'(\sqrt{n} x+\cdot) dx'$, we can also express this limiting result in terms of the  limit $Q$ to $Q_{\nu_n}$, which is  the push-forward of $P$ by $(x,\j) \mapsto \frac{1}{2\pi} \div \j + m_0(x)$.
The limiting energy for both the upper bound and the lower bound is then  $$\frac{|\E|}{\pi} \int \mathbb{W}(\nu) \, dQ(x,\nu).$$
Of course such a statement is a bit weaker than Theorem \ref{th2} since some information is lost: namely we do not keep the information of  which $\j$ corresponded to $\nu$.
\end{remark}

}

This  theorem is the analogue of the main result
 of \cite{ss1} but for $\w$ rather than the Ginzburg-Landau energy. It is technically simpler to prove, except for the possibility of a nonconstant weight  $\bm(x)$ which was absent from  \cite{ss1}.  It can be stated as the fact that $\frac{|E|}{\pi}\int W dP$, which can be understood as the average of $W$ with respect to all possible blow-up centers in $\E$ (chosen uniformly at random), is the $\Gamma$-limit of $w_n$ at next order. Its minimum over all admissible probabilities is  $\a$.

 The estimate \eqref{dxr2} gives a control on the ``discrepancy" $D$ (between the effective   number of points  and the expected one)  at the scale $ R/\sqrt{n}$.
Note that in  a  recent paper \cite{ao}, the authors  also study the fine  behavior  of weighted Fekete sets. Using completely different methods, based on Beurling-Landau densities and techniques going back to \cite{landau},  they are able to show the very strong result that
$$ \limsup_{R\to \infty} \limsup_{n\to \infty} \frac{D(x_n,R)}{R^2}=0,$$ as long as $\dist(x_n, \p \E') \ge \log^2 n  $.
 This shows that the density of points follows $\mu_0$ at the microscopic scale  $1/\sqrt{n}$ and thus the configurations are very rigid. This still leaves however  some uncertainty about the patterns they should follow. On the contrary, our result is less precise about $D(x,R)$ since we only recover the optimal estimate when $R$ grows faster than $n^{1/4}$,  but it  connects the  pattern formed by the points to the  Abrikosov triangular lattice via the minimization of $W$.

\medskip

We now turn to Coulomb gases, i.e. to the case with temperature.  It is straightforward from the form \eqref{loi2} and the estimate (provided by Theorem \ref{valz})
$\log \K=O(n\beta)$ where $\K$ is defined in \eqref{defK}, to deduce  that $\F\le C$  except on a set of small probability, because $F_n$ controls the deviation between $\nu_n$ and $n\mu_0$ and controls $W$.
This fact allows to derive various consequences, the first being estimates on the probability of certain rare events.

\begin{theo} \label{th3} Let $V$ satisfy assumptions  \eqref{assumpV1}--\eqref{assumpV4}.

There exists a universal constant $R_0>0$ and $c,C>0$ depending only on $V$ such that: For any $\beta_0>0$, any $n$ large enough depending on $\beta_0$, and any $\beta>\beta_0$, for   any $x_1,\ldots,x_n\in \mr^2$, any $R>R_0$,  any $x_0' = \sqrt{n} x_0\in \mr^2$ and any $\eta>0$, letting  $\nu_n=\sum_{i=1}^n \delta_{x_i}$,  we have the following:
\begin{equation}\label{local} \log \Q \( \left|D(x_0',R)\right|\ge \eta R^2   \) \le
- c\beta \min(\eta^2,\eta^3) R^4 + C \beta (R^2+ n) + C  n .\end{equation}
\begin{equation}\label{controlez}
\log\Q\( \int\x \, d\nu_n \ge \eta\) \le - \hal  n \beta\eta+ C  n (\beta+1) .\end{equation}
Moreover, for any smooth bounded  $U' = \sqrt n U \subset \mr^2$,
\begin{equation}\label{fubi}
\log\Q\( \int_{U'}   \frac{ D(x',R)^2}{R^2} \min \(1,\frac{|D(x',R)|}{R^2} \) \, dx'  \ge
  \eta \) \le   n \beta ( - c\eta+ C |U| + C ) +Cn.
\end{equation}
Finally, if $q\in[1,2)$ there exists $c,C>0$ depending on $V$ and $q$ such that $\forall\eta\ge 1$, $R>0$,
\begin{equation}\label{cv}
\log\Q\(\(1+\frac{R^2}n\)^{\frac12-\frac1q} \|\nu - n \mo\|_{W^{-1,q}(B_{R/\sqrt n})}\ge
\eta \sqrt{n}\)\le -c n\beta\eta^2 + Cn(\beta+1),
\end{equation}
where $W^{-1,q}(\Omega)$ is  the dual of the Sobolev space $W^{1,p}_0(\Omega)$ with $\frac{1}{p}+\frac{1}{q}=1$; in particular $W^{-1,1}$ is the dual of Lipschitz functions.
\end{theo}
These estimates can roughly be read in the following way: as soon as $\eta$ is large enough, the events in parentheses have probability decaying like $e^{-cn}$. More precisely,  we bound the probability that a  ball contains too many or too few points compared to the expected number $n\mu_0(B)$, but whereas the circular law does it for a {\em macroscopic ball}, i.e. for $R$ comparable to $\sqrt n$, the estimate \eqref{local} is effective at intermediate scales, of the order of $n^{1/4}$. This is sometimes called in this context  ``undercrowding" or ``overcrowding" of points, see \cite{jlm,nsv,krish}. In view of similar results in \cite{jlm}  and the result of \cite{ao}, we can  expect this to hold as soon as $R\gg 1$, but this seems out of reach by our method. This can also be compared with analogous estimates without error terms proven in the case of Hermitian matrices, cf. \cite{esy}.  These results, in the Hermitian case, are proven in the general setting of Wigner matrices, i.e. Hermitian matrices with random i.i.d. entries, which do not need to be Gaussian.  They only concern some fixed $\beta$ however.

The estimate \eqref{fubi} gives a global version of this result: it  expresses a control on the average  microscopic ``discrepancy" $D$. This control is in $L^2$ for large values of the discrepancy, and in $L^3$ for small values. The estimate \eqref{controlez} allows, in view of \eqref{lemz}, to control (again with some threshhold to be beaten) the probability that some points may be far from the set $\E$.
Note that since $\nu_n$ is a non-normalized empirical measure, \eqref{controlez} ensures for example that the probability that a single  point lies at a distance $\eta$ from $\E$ is exponentially small as soon as $\eta$ is larger than some constant.
All these estimates rely on controlling $D$ by $\F$.

Finally,
\eqref{cv} tells us that fluctuations around the law $n \mo $ can be globally controlled   (take for example $R=\sqrt{n}$)  by
$O(\sqrt{n})$ (except with exponentially small probability).  We believe this estimate to be optimal. Its proof  uses in a crucial manner the result of \cite{st}, which controls, via Lorentz spaces, the difference $\nu_n - n \mo$ in terms of  $W$  \medskip or $\widehat{\F}$.

Our last result mostly expresses Theorem~\ref{th2} in a ``moderate" deviations language. Before stating it, let us recall for comparison the result of \cite{bz}:
\begin{theo}[Ben Arous - Zeitouni] \label{thebz} Let $\beta=2$ and $V(x)=|x|^2$. Denote by $\wQ$ the image of the law \eqref{loi} by the map $(x_1,\dots,x_n)\mapsto \nu_n$, where $\nu_n=\frac{1}{n}\sum_{i=1}^n \delta_{x_i}$. Then for any subset  $A$ of the set of  probability measures on $\mr^2$ (endowed with the topology of weak convergence), we have
$$
-\inf_{\mu \in \mathring{A}} \widetilde I(\mu) \leq \liminf_{n\to \infty} \frac{1}{n^2} \log \wQ(A) \leq \limsup_{n\to\infty} \frac{1}{n^2} \log \wQ(A) \leq -\inf_{\mu\in \bar A} \widetilde I(\mu),
$$
where $\widetilde I = I - \min I$.
\end{theo}

\begin{theo}\label{th4} Let $V$ satisfy assumptions  \eqref{assumpV1}--\eqref{assumpV4}. For any $ \beta>0$, the following holds. 
For any $n>0$ let $A_n\subset \mc^n$. Denote
\begin{equation}\label{ainf} A_\infty = \bigcap_{n>0}\overline{\bigcup_{m>n} i_m(A_m)},\end{equation}
where $i$ is as in \eqref{in}, and the topology is the weak convergence on $\P(X)$.
Then for any $\eta>0$ there is $C_\eta>0$ depending on $V$ and $\eta$ only such that $\a$ being as in \eqref{defa},
\begin{equation}
\label{ldr}
\limsup_{n\to \infty} \frac{ \log \Q(A_n)}{n} \le -\frac{\beta}{2}
 \Big(\frac{|\E|}\pi \inf_{P\in A_\infty}\int W(\j) dP(x,\j) - \alpha- \eta - \frac{C_\eta}\beta \Big).
\end{equation}

Conversely, let  $A\subset\P(X)$ be a set of $T_{\lambda(x)}$-invariant probability measures on $X$ and let $\mathring{A}$ be the interior of $A$. Then for any $\eta>0$, there exists a sequence of subsets  $A_n\subset \E^n$ such that
\begin{equation} \label{ldlb}
-\frac{\beta}{2} \(\frac{|\E|}\pi \inf_{P\in \mathring{A}}\int W(\j) dP(x,\j) - \alpha+ \eta+\frac{C_\eta}\beta\) \le \liminf_{n\to\infty} \frac{\log \Q(A_n)}{n},
\end{equation}
and such that for any sequence $\{\nu_n = \sum_{i=1}^n\delta_{x_i}\}_n$ such that  $(x_1,\dots,x_n)\in A_n$ for every $n$ there exists a sequence of fields $\j_n\in\Lp$ such that $\div \j_n = 2\pi(\nu_n' - \bm')$ and such that --- defining $P_n$ as in \eqref{pnun} with $\j_n$ replacing $\j_{\nu_n}$ --- we have
\begin{equation}\label{fauxi} \lim_n P_n\in \mathring A.\end{equation}
\end{theo}
Note that if $P_n$ was $P_{\nu_n}$, then \eqref{fauxi} would be equivalent to saying that $\cap_n\overline{\cup_{m>n} i_m(A_m)}\subset\mathring A$. The difference betwen $P_{\nu_n}$ and $P_n$ is that the latter is generated by a field $\j_n$ which is not necessarily a gradient.

Compared to Theorem \ref{thebz}  this result can  be seen as a next order (speed $n$ instead of $n^2$) deviations result, where the average of $W$ over blow-up centers plays  the role of a rate function, with a margin which becomes small as $\beta\to \infty$. While Theorem \ref{thebz} said that empirical measures at macroscopic scale converge to $\mo$, except for a set of exponentially decaying probability,  Theorem \ref{th4} says  that within the empirical measures which do converge to $\mo$, the ones
with large average of $W$ (computed after blow-up) also have exponentially decaying probability, but at the slower rate $e^{-n}$ instead of $e^{-n^2}$. More precisely, there is a threshhold $C/\beta$ for some $C>0$,  such that configurations  satisfying
$$\frac{|\E|}{\pi}\int  W\, dP \ge \a +\frac{C}{\beta}$$
have exponentially small probability, where we recall $\a$ is also the minimum possible value of $\frac{|\E|}{\pi}\int  W\, dP$.
Since we believe that $W$ measures the disorder of a (limit) configuration of (blown up) points in the plane, this means that most configurations have a certain order.
The threshhold, or gap, $C/\beta$ tends to $0$ as $\beta $ tends to $\infty$, hence in this limit, configurations have to be closer and closer to the minimum of the average of $W$, or have more and more order.

Modulo the conjecture that the minimum of $W$ is achieved by the perfect ``Abrikosov" triangular lattice, this constitutes a crystallisation result.  Note that to solve this conjecture, it would suffice to evaluate $\a$, which in view of Theorem \ref{valz} is equivalent to being able to compute the asymptotics of $\Z$  as $ \beta\to \infty$.

\ed{At nonzero temperature, the probabilities  $P$ are not expected to be concentrated on minimizers of $W$, indeed numerical simulations of the Ginibre ensemble\footnote{cf.  e.g. Benedek Valko's webpage {\tt http://www.math.wisc.edu/~valko/courses/833/833.html}} corresponding to $\beta=2$ show patterns of points with a certain microscopic disorder, which are  certainly not crystalline.  This is probably explained by the fact that at finite temperature, and in this order $n$, an entropy term should come to compete with the minimization of $W$.
One may wonder if at least there exists a limiting law on the probabilities $P_n$, and which it is.
The following theorem answers positively the question of  the existence:
\begin{theo}\label{additi}
For each integer $n$, and a given $\beta >0$, let $\widetilde{\mathbb{P}_n^\beta}$ denote the push-forward of $\Q$ by $i_n$ defined in \eqref{in}. It is an element of $\P(\P(X))$.
Then $\{ \widetilde{\mathbb{P}_n^\beta } \}_{n}$ is tight and converges, up to a subsequence, to a probability measure $\widetilde{\mathbb{P}^\beta}$  on $\P(X).$
\end{theo}
This shows the existence of a limiting ``electric field process".

\begin{remark}
Pushing forward by $(x,\j) \mapsto \frac{1}{2\pi} \div \j + m_0'(\sqrt{n} x+\cdot) $ as in Remark \ref{rem15}
gives the existence  of  a limiting  point process $\widetilde{\mathbb{Q}^\beta}$ i.e.  a probability on the limiting $Q$'s, which themselves encode all the $(x,\nu)$'s.
\end{remark}
}

To conclude the following open questions  naturally arise in view of our results, and are closely related  to one another:
\begin{itemize}
\item[-]  Prove that $\min_{\mathcal{A}_1} W$ is achieved by the triangular lattice.
\item[-]  Find whether a large deviations statement is true at speed $n$, and if it is, find the  rate function.
\item[-] Characterize the limiting processes $\widetilde{\mathbb{P}^\beta}$ and $\widetilde{\mathbb{Q}^\beta}$.
\end{itemize}

In \cite{ss2} we show  that all the results we have obtained here are also true in the case of points on the real line, i.e. for 1D log gases  or Hermitian random matrices. There the minimization of $W$ is solved (the minimum is the perfect lattice $\mz$) and the crystallisation result is complete.
\\

The rest of the paper is organized as follows:
Section 2 contains the proof of the ``splitting formula". In Section 3, we present the ``spreading result" from \cite{ss1} and some first corollaries. In Section 4, we present an explicit construction which yields the lower bound on $\Z$, whose proof is postponed to Section 7. In Section 5, we show how $W$ controls the overcrowding/undercrowding of points, and prove Theorem \ref{th3}.
In Section 6 we present the ergodic averaging approach (the abstract result) and apply it to conclude  the proofs of Theorem \ref{th2}, \ref{th4}  and \ref{valz}. \medskip

{\it Acknowledgements:}
We are grateful to  Alexei Borodin, as well as  G\'erard Ben Arous, Amir Dembo, Percy Deift, and  Alice Guionnet for their interest and  helpful discussions. We thank Peter Forrester for useful comments and references.
E. S. was supported by the Institut Universitaire de France and S.S. by a EURYI award.

\section{Proof of the splitting formula}\label{secproofsplit}
The connection between $w_n$ and $W$ originates in the following computation
\begin{lem}
For any $x_1, \dots, x_n$ and letting $\nu_n = \sum_{i=1}^n\delta_{x_i}$ the following holds
\begin{multline}\label{idwn}
\F(\nu) =  \frac{1}{n\pi} W(\nab  H_n', \indic_{\mr^2} )  +2 \sum_{i=1}^n  \z (x_i) \\
= \frac1n\(\w(x_1, \dots, x_n) - n^2 \I (\mo) + \frac{n}{2} \log n \) , \end{multline}
where  $\F$ is defined in \eqref{defF},  $W$ is defined in \eqref{WR},  and $H'_n$ is defined in \eqref{Hp}.
\end{lem}
\begin{proof} Let $\nu_n = \sum_{i=1}^n \delta_{x_i}$, and let $H_n$ be defined by
$$H_n= - 2\pi \Delta^{-1} \( \nu_n - n \mu_0\).$$
First we note that  since $\nu_n $ and $n \mo$ have same mass and  compact support we have $H_n(x) = O(1/|x|)$ and $\nab H_n(x) = O(1/|x|^2)$ as $|x|\to +\infty$.

We  prove  that, denoting by $D$  the diagonal in $\mr^2\times\mr^2$, we have
\begin{equation}\label{enren} \int_{(\mr^2\times \mr^2 ) \sm D}\,- \log |x-y| \, d(\nu_n - n \mo)(x)\, d(\nu_n- n \mo)(y)= \frac{1}{\pi} W(\nab H_n ,\indic_{\mr^2}).\end{equation}
First, using Green's formula, we have
\begin{multline}
\int_{B_R\sm\cup_{i=1}^n B(x_i,\eta)} |\nab H_n|^2 = \int_{\p B_R} H_n\nab H_n \cdot \vec{\nu} + \sum_{i=1}^n \int_{\p B(x_i,\eta)} H_n \nab H_n \cdot \vec{\nu}  \\
+ 2\pi\int_{B_R\sm \cup_{i=1}^n B(x_i,\eta)}H_n\,d(\nu_n - n\mo).\end{multline}
 Here, and in all the paper, $\vec{\nu}$ denotes the outer unit normal vector.

Let
$H^i(x) := H_n(x)+\log|x-x_i|$. We have $H^i = -\log\ast (\nu^i - n\mo)$, with $\nu^i = \nu_n-\delta_{x_i}$, and near $x_i$, $H^i$ is $C^1$. Therefore, using \eqref{Hp} and the boundedness of $\bm$ in $L^\infty$, we have that, as $\eta\to 0$
$$\int_{\p B(x_i,\eta)} H_n \cdot \vec{\nu} = -2\pi\log\eta + 2\pi H^i(x_i) + o(1),$$
 while the integral on $\p B_R$ tends to $0$ as $R\to +\infty$ from the decay properties of $H_n$. We thus obtain, as $\eta\to 0$ and $R\to +\infty$,
$$\int_{B_R\sm\cup_i B(x_i,\eta)} |\nab H_n|^2 = -2\pi n\log\eta+ 2\pi \sum_{i=1}^n H^i(x_i) - 2\pi n\int_{\mr^2}  H_n\,d\mo + o(1),$$
and therefore, by definition of $W$,
\begin{equation}\label{intermed} W(\nab H_n,\indic_{\mr^2}) = \pi \sum_{i=1}^n  H^i(x_i) - \pi n\int H_n\,d\mo.\end{equation}

Second we note that
$$\int_{\mr^2\sm\{x_i\}} - \log|x_i-y|\,d(\nu_n - n\mo)(y) = H^i(x_i),$$
and if $x \notin \{x_i\}$ then
$$\int_{\mr^2\sm\{x\}} - \log|x-y|\,d(\nu_n - n\mo)(y) = H_n(x).$$
It follows that
$$\int_{D^c}\,- \log |x-y| \, d(\nu_n - n \mo)(x)\, d(\nu_n - n \mo)(y)= \sum_{i=1}^n  H^i(x_i) - n\int_{\mr^2} H_n(x)\,d\mo(x),$$
which together with \eqref{intermed} proves \eqref{enren}.

On the other hand, we may rewrite $\w$ as
$$\w(x_1, \dots, x_n) =\int_{D^c}\,- \log |x-y| \, d\nu_n(x)\, d\nu_n(y) + n\int V(x)\, d\nu_n(x)$$
and, splitting $\nu_n$ as $n \mo + \nu_n - n\mo$ and using the fact that $\mo\times\mu_0 (D) = 0$, we obtain
\begin{multline*}
w(x_1, \dots, x_n) = n^2 \I(\mo)+  2n \int U^{\mo}(x)\, d(\nu_n - n \mo)(x)+ n\int V(x)\, d(\nu_n - n \mo)(x)\\+
\int_{D^c}\,- \log |x-y| \, d(\nu_n - n \mo)(x)\, d(\nu_n - n \mo)(y) .
\end{multline*}
Since $U^{\mo} + \frac{V}{2}= c + \z$ and since $\nu_n $ and $n \mo$ have same mass $n$, we have
$$2n \int U^{\mo}(x)\, d(\nu_n - n \mo)(x)+ n\int V(x)\, d(\nu_n- n \mo)(x)=  2 n \int \z \,d(\nu_n  - \mo)= 2n \int \z \,d\nu_n,$$
using the fact that  $\z=0$ on  the support  of $\mo$. Therefore,   in view of \eqref{enren} we have  found
\begin{equation}\label{avantchvar}w(x_1, \dots, x_n) = n^2 \I(\mo)+  2n \int \z \,d\nu_n +\frac{1}{\pi} W(\nab H_n ,\indic_{\mr^2}).
\end{equation}
But, changing variables, we find
$$\hal\int_{\mr^2 \sm \cup_{i=1}^n B(x_i, \eta) } |\nab H_n|^2 =\hal \int_{\mr^2 \sm \cup_{i=1}^n B(x_i',\sqrt{n} \eta)}  |\nab H_n'|^2,$$
and by  adding $\pi n\log\eta$ on both sides and letting $\eta\to 0$ we deduce that $W(\nab H_n , \indic_{\mr^2})= W(\nab H'_n, \indic_{\mr^2} ) -  \frac{ \pi}{2} n \log n$. Together with \eqref{avantchvar} this proves \eqref{idwn}.
\end{proof}

\section{A first lower bound on $\F$ and upper bound on $\Z$}
The crucial fact that we now  wish to exploit  is that, even though $W(\j, \chi)$ or rather its associated energy density  does not have a sign,  there are good lower bounds for $\F$. This follows from the analysis of \cite{ss1}, more specifically from the following
 ``mass spreading result", adapted  from  \cite{ss1},  Proposition~4.9 and Remark~4.10 (with  slightly different notation), which itself is based on the so-called ``ball construction method", a crucial tool in the analysis of Ginzburg-Landau equations.  This result, that we will use here as a black box, says that even though the energy density  associated to $W(\j, \chi)$ is not positive (or even bounded below), it can be replaced by an energy-density $g$ which is uniformly bounded below, at the expense of a negligible error.

 For any set $\om$, $\widehat{\om}$ denotes  its 1-tubular neighborhood, i.e.
 $\{x\in \mr^2, \dist (x, \om) < 1\}$.
\begin{pro}\label{wspread} Assume $\om\subset\mr^2$ is open
 and $(\nu,\j)$ are such that $\nu = 2\pi \sum_{p\in\Lambda}
  \delta_p$ for some finite subset $\Lambda$ of $\widehat \om$
   and  $\div \j =2\pi ( \nu -\ed{a(x)}dx) $, $\curl \j = 0$ in $\widehat \om$,
    where $a\in L^\infty(\widehat \om)$. Then, given any $\ro>0$ there exists a measure  $g$ supported  on $\widehat \om$ and such that
\begin{itemize}
\item[-]
there exists a family $\mathcal{B}_\ro$ of disjoint closed balls  covering $\supp(\nu)$, with the  sum of  the radii  of the balls in $\mathcal{B}_\ro$ intersecting with any ball of radius $1$ bounded by $\ro$, and  such that
\begin{equation}\label{lbg}
g\ge - C (\|a\|_{L^\infty}+1) + \frac14 |\j|^2\indic_{\om\sm \mathcal{B}_\ro} \quad \text{in} \ \widehat\om,\end{equation}
  where $C$ depends only on $\ro$.
 \item[-]\begin{equation} \label{gcj} \text{$\D g = \hal |\j|^2$ outside $\cup_{p\in\Lambda} B(p,\lambda)$}\end{equation} where $\lambda$ depends only on $\ro$.
\item[-]  For any function $\chi$ compactly supported in $\om$ we have
\begin{equation}\label{wg}\left|W(\j,\chi) - \int \chi\,dg\right|
\le C N (\log N + \|a\|_{L^\infty})\|\nab\chi\|_\infty,\end{equation}
where  $N=\#\{p\in\Lambda: B(p,\lambda)\cap\supp(\nabla\chi)\neq\varnothing\}$ for some   $\lambda$  and $C$ depending only on $\ro$.
\item[-] For any  $U\subset \om$,
\begin{equation}\label{bgnalpha}\#(\Lambda\cap U) \le C\(1+\|a\|_{L^\infty}^2|\widehat U| + g(\widehat U)\).\end{equation}

\end{itemize}
\end{pro}
Note that the result in \cite{ss1} is not stated for any $\rho$ but a careful inspection of the proof there allows to show that it can be readapted to make $\rho$ arbitrarily small. \ed{From now on, we take some  $\ro<1/8$.}

\begin{defi}\label{defG} Assume  $\nu_n = \sum_{i=1}^n\delta_{x_i}$. Letting $\nu_n' = \sum_{i=1}^n\delta_{x_i'}$ be the measure in blown-up coordinates and  $\j_{\nu_n} = -\nab  H_n'$, where $H_n'$ is defined by \eqref{Hp}, we denote by  $g_{\nu_n}$ the result of applying the previous proposition to $(\nu_n',\j_{\nu_n})$ in $\mr^2$.
\end{defi}
Even though we will not use the following result in the sequel,  we state it to show how we can quickly derive a first upper bound on $\Z$ from what precedes.
\begin{pro}\label{lowZ} We have
\begin{equation}\label{ubk}
\log \K \le    C n \beta + n (\log |\E|+o(1))
\end{equation}
where we recall $\E=\supp(\mu_0)$,
and
\begin{equation}\label{ubz}\log  \Z\le - \frac{\beta}{2} n^2 \I(\mo) + \frac{\beta n}{4} \log n + C n \beta + n (\log |\E|+o(1)) \end{equation}
where $o(1)\to 0$ as $n \to \infty$ uniformly with respect to $\beta>\beta_0$, for any $\beta_0>0$,   and  $C$ depends only on $V$.
\end{pro}

The  proof uses two lemmas.

\begin{lem} For any $\nu_n=\sum_{i=1}^n \delta_{x_i}$, we have
\begin{equation}\label{fg}\F(\nu_n) = \frac{1}{n\pi} \int_{\mr^2} dg_{\nu_n} + 2\int \z\,d\nu_n,\end{equation} where $\F$ is as in \eqref{defF}.
\end{lem}
\begin{proof} This follows from \eqref{wg} applied to $\chi_{K_R}$, where $\chi_{B_R}$ is as in \eqref{defchi}. If $R$ is large enough then, $\lambda$ being the constant of Proposition \ref{wspread},  $\#\{p\in\supp(\nu_n): B(p,\lambda)\cap\supp(\nabla\chi_{K_R})\neq\varnothing\} = 0$ and therefore \eqref{wg} reads
$$W(\j_{\nu_n},\chi_{K_R}) = \int \chi_{K_R}\,dg_{\nu_n}.$$
Letting $R\to +\infty$ yields $W(\j_{\nu_n},\indic_{\mr^2}) = \int\,dg_{\nu_n}$ and the result, in view of \eqref{defF}.
\end{proof}

\begin{lem}\label{lemintxi}
Letting $\nu_n$ stand for $\sum_{i=1}^n \delta_{x_i}$ we have, for any constant $\gamma>0$ and uniformly w.r.t. $\beta$ greater than any arbitrary positive constant $\beta_0$, we have 
 \begin{equation}\label{lex}
\lim_{n \to \infty} \( \int_{\mc^n} e^{-\gamma\beta n \int \zeta \,d\nu_n }  \, dx_1 \dots dx_n \)^\frac{1}{n} = |\E|.\end{equation}\end{lem}
\begin{proof}
\ed{This is the place where we use the assumption \eqref{assumpV4}. 
We recall that  $\zeta= U^{\mu_0} + \frac{V}{2}- c$, and note that since $\mu_0$ is a compactly supported probability measure $U^{\mu_0}(x)= - \int \log |x-y| \, d\mu_0(y)$ behaves asymptotically like $-\log |x|$ when $|x|\to \infty$, more precisely  one can easily show that there exists $C$ such that $|U^{\mu_0}(x)+  \log |x||\le C$ for $|x|$ large enough.
It thus follows that $\zeta(x) \ge - \log |x|+ \frac{V}{2}(x) - C$ for $|x|$ large enough,  and in view of \eqref{assumpV4}, this implies that for some $\beta_2>0$,  $\int_{\mc} e^{- \beta_2\zeta(x)}\, dx$ converges.}

Next, by separation of variables, we have
$$  \int_{\mc^n} e^{-  \gamma \beta n   \sum_{i=1}^n \zeta(x_i)   } \, dx_1 \dots dx_n= \( \int_{\mc}
e^{-\gamma\beta n \zeta(x)} \, dx\)^n $$
On the other hand, we have  $\zeta\ge 0$ and  $\{\z=0\}= \E$ by \eqref{eqz}, hence we have
 $e^{-\gamma\beta n\zeta(x)}\to \indic_\E$
pointwise,  as $\beta n \to \infty$. \ed{In addition, if $\beta\ge \beta_0>0$,  for $n$ large enough depending on $\beta_0$,  $e^{-\gamma\beta n\zeta(x)} $ is dominated by $e^{-\beta_2 \zeta(x)}$ which is integrable.  Therefore, by dominated convergence theorem,} it follows that \eqref{lex} holds uniformly w.r.t. $\beta\ge\beta_0$, for any $\beta_0>0$.
\end{proof}

\begin{proof}[Proof of Proposition~\ref{lowZ}]
Let again $\nu_n$ stand for $\sum_{i=1}^n \delta_{x_i}$. From  \eqref{gcj} we have $g_{\nu_n} \ge 0$ outside $\cup_iB(x_i,\lambda)$ and from \eqref{lbg} we have $g_{\nu_n}\ge - C$ (depending only on $\|\bm\|_{L^\infty}$ hence on $V$) in $\cup_iB(x_i,\lambda)$. Inserting into  \eqref{fg} we deduce that
$$\F(\nu_n) \ge - C + 2\int \zeta\,d\nu_n ,$$ where $C$ depends only on $V$.
Inserting into \eqref{loi2} and integrating over $\mc^n$, we find
$$1\le \frac{1}{\K} e^{Cn \beta  } \int_{\mc^n} e^{- n \beta \int \zeta\,d \nu_n} \, dx_1 \dots dx_n.$$ Inserting \eqref{lex} and taking logarithms, it follows that
$$\log \K \le  C n \beta + n (\log |\E|+o(1)). $$
The relation \eqref{ubz} follows using \eqref{defK}.\end{proof}

\section{A construction and a lower bound for $\Z$ }
In this section, we construct a set of explicit configurations whose $W$ is not too large,  and show that  their probability is not too small, which will lead to a lower bound on $\Z$. This is the longest part of our proof.  The method is borrowed from \cite{ss1} but requires various adjustments  that we shall detail  in Section \ref{sec5}.
We will need \eqref{assumpV3} in order to simplify the construction and estimates near the boundary.



The following  proves  Theorem~\ref{th2}, part B and contains a bit more information useful for proving Theorem~\ref{th4}.
\begin{pro} \label{construct} Let $P$ be a $T_{\lambda(x)}$-invariant  probability measure on $X = \E\times\Lp$ with  first marginal $dx_{|\E}/|\E|$ and such that for $P$ almost every $(x,\j)$ we have $\j\in\mathcal A_{\bm(x)}$. Then, for any  $\eta>  0$, there exists $\delta >0$ and for any $n$ a subset $A_n\subset\mc^n$ such that $|A_n|\ge n!(\pi\delta^2/n)^n$ and for every sequence $\{\nu_n=  \sum_{i=1}^n \delta_{y_i}\}_n$ with $(y_1,\dots,y_n)\in A_n$ the following holds.

i) We have the upper bound
\begin{equation}\label{bsw}\limsup_{n \to \infty}\frac{1}{n} \(\w(y_1, \dots, y_n) - n^2 \I (\mo)   + \frac{n}{2}\log n \)  \le \frac{|\E|}{\pi}\int W(\j) \, dP(x,\j)+\eta.\end{equation}

ii) There exists $\{\j_n\}_n$ in $L^p_\loc(\mr^2,\mr^2)$ such that $\div \j_n = 2\pi( \nu_n' -\bm')$ and such that  the image $P_n$ of $dx_{|\E}/|\E|$ by the map $x\mapsto \(x,\j_n(\sqrt n x+\cdot)\)$ is such that
\begin{equation}\label{convpn} \limsup_{n \to \infty} \dist(P_n,P) \le \eta, \end{equation}
where $\dist$ is a distance which metrizes the topology of weak convergence on $\P(X)$.\\
\ed{Respectively,  the image $Q_n$ of  $dx_{|\E}/|\E|$ by the map $x\mapsto \(x,\nu_n'(\sqrt n x+\cdot)\)$ is such that
\begin{equation}\label{convpn} \limsup_{n \to \infty} \dist(Q_n,Q) \le \eta, \end{equation}
where $\dist$ is a distance which metrizes the topology of weak convergence, and $Q$ is the push-forward of $P$ by $(x,\j) \mapsto \frac{1}{2\pi}\div \j + m_0'(x).$}
\end{pro}

Applying the above proposition with $\eta = 1/k$ we get a subset $A_{n,k}$ in which we choose any $n$-tuple $(y_{i,k})_{1\le i\le n}$. This yields in turn a family $\{P_{n,k}\}$  of probability measures on $X$. A standard  diagonal  extraction argument then yields
\begin{coro}[Theorem~\ref{th2}, Part B]\label{etanul} Under the same assumptions as Proposition~\ref{construct},   there exists a sequence $\{\nu_n= \sum_{i=1}^n \delta_{x_i}\}_n$ and a sequence $\{\j_n\}_n$ in $L^p_\loc(\mr^2,\mr^2)$ such that $\div \j_n = 2\pi( \nu_n' -\bm'(x') \,dx')$ and
\begin{equation}\label{bswbis}\limsup_{n \to \infty}\frac{1}{n} \(\w(x_1, \dots, x_n) - n^2 \I (\mo)   + \frac{n}{2}\log n \)  \le \frac{|\E|}{\pi}\int W(\j) \, dP(x,\j). \end{equation}
Moreover,   denoting  $P_n$ the image of $dx_{|\E}/|\E|$ by the map $x\mapsto \(x,\j_n(\sqrt n x+\cdot)\)$, we have $P_n\to P$ as $n\to +\infty$.
\end{coro}
Another  consequence of Proposition~\ref{construct} is, recalling  \eqref{defa} and \eqref{defK}:
\begin{coro}[Lower bound part of Theorem \ref{valz}]\label{coro43} For any $\eta>0$ there exists $C_\eta>0$ such that for any $\beta>0$ we have
\begin{equation}\label{lbk}
\liminf_{n\to +\infty} \frac{\log \K}n \ge - \frac{\beta}{2}  (\alpha +\eta) - C_\eta .\end{equation}
\end{coro}

\begin{proof}[Proof of the corollary]  Choose $\j_0\in\mathcal A_1$ to be a minimizer for $\min_{\mathcal{A}_1} W$, which exists by \cite[Theorem 1]{ss1}, and let $P$ be the image of the normalized Lebesgue measure on $\E$ by the map $x\mapsto (x, \sigma_{\bm(x)} \j_0)$, where
\begin{equation}\label{sigma} \sigma_{m} \j (y) := \sqrt m\,\j(\sqrt m y).\end{equation}
Then by construction $P$-almost every $(x,\j)$ satisfies $\j\in \mathcal A_{\bm(x)}$ and the first marginal of $P$ is $dx_{|\E}/|\E|$.

Given $\eta>0$, applying Proposition~\ref{construct} and using the notation there,  we have
$|A_n|\ge n! (\delta^2/n)^n$ and from \eqref{loi2} we have
\begin{equation}\label{igesf}1\ge \int_{A_n} \frac{1}{\K} e^{- n \frac{\beta}{2} \F(\nu_n)} \, dy_1 \dots \, dy_n,\end{equation}
where $\nu_n =\sum_{i=1}^n \delta_{y_i}$. From \eqref{lienfw} and \eqref{bsw}, when $(y_1,\dots,y_n)\in A_n$ we have
$$\F(\nu_n)\le \eta + \frac{|\E|}{\pi}\int W(\j) \, dP(x,\j) =\eta +  \frac{1}{\pi}\int_\E W(\sigma_{\bm(x)}\j_0)\,dx.$$
From  \eqref{minalai1}, and since  $\int_\E\bm=1$, we obtain
$$ \frac{1}{\pi}\int_\E W(\sigma_{\bm(x)}\j_0)\,dx= \frac1\pi W(\j_0) - \frac1{2}\int_\E \bm(x)\log\bm(x)\,dx=\a,$$
by definition  \eqref{defa}. We deduce
$$\F(\nu_n)\le \alpha+\eta.$$
Together with \eqref{igesf} we find $1\ge \frac{|A_n|}{\K} e^{- n \frac{\beta}{2} (\eta+\alpha)}.$
Taking logarithms, we are led to
$$\log \K \ge \log n! + n \log \delta^2 - n \log n -\hal  n\beta(\eta +\alpha).$$
From Stirling's formula, $\log n!\ge  n \log n -C n$ and   we deduce  \eqref{lbk}, with $C_\eta = - \log \delta^2 +C$. Note that the dependence on $\eta$ comes from $\delta$. \end{proof}

\section{Consequences for fluctuations: proof of Theorem \ref{th3}}

In this section, we prove Theorem \ref{th3}. The first step is to find, via the method first introduced in \cite{gl7} and tools from \cite{ss1,st}  how $\F$ and $W$  control the
discrepancy between $\nu_n$ and the  measure $n \mo$,
\ed{as can be seen in the following
\begin{lem}\label{lemnew}
Let  $\nu_n=\sum_{i=1}^n \delta_{x_i}$ and  $\j_{\nu_n}=-\nab H_n'$ be associated  through \eqref{Hp}. Let $B_R$ be any ball of radius $R$ (not necessarily centered at $0$).
Assume $\chi$ is a smooth nonnegative function  compactly supported in $U$.  Then for any $1<q<2,$ we have
\begin{equation}\label{firstitem}
\|\sqrt{\chi} \, \j_{\nu_n}\|_{L^q(U)} \le C_q |U|^{\frac1q-\hal}   \(W(\j_{\nu_n},\chi)+ \nu_n'(\widehat{U} ) (\|\chi\|_{L^\infty}+ \|\nab \chi
\|_{L^\infty} ) + N\log N  \)^{\frac 12},
\end{equation}
where $N= \#\{x_1, \dots, x_n|  0< \chi(x_i) \le \hal \chi \} $.
Thus
\begin{equation}\label{jpp}
\int_{B_R} |\j_{\nu_n}|^q
\le C_q (n+R^2)^{1-\frac{q}2} n^{\frac{q}{2}}\( \widehat{\F}(\nu_n) +1\)^{\frac{q}{2}}.
\end{equation}
and
\begin{equation}\label{fluctu}
\|\nu_n-n\mu_0\|_{W^{-1, q}(B_{R}) } \le   C_q  (1+R^2)^{\frac{1}{q}-\frac{1}{2}} n^{\hal} \( \widehat{\F}(\nu_n) +1\)^{\frac{1}{2}}.
\end{equation}

\end{lem}
\begin{proof}
The first item is a rewriting of \cite[Corollary 1.2]{st}.
We then choose $\chi$   such that  $\chi  = 1$ on $U:=\E'\cup \(\cup_{i=1}^n B(x_i',\hal)\)\cup B_R $ and $\|\chi\|_\infty$, $\|\nab \chi\|_\infty\le 1$, compactly supported on $\widehat {U} = \{x: d(x,U)\le 1\}$. Using the fact that    $| \widehat {U}| \le  C (n + R^2)$ where $C$ depends only on $\E$,
  from \eqref{firstitem} we find
\begin{equation}\label{jp}
\|\sqrt{\chi} \, \j_{\nu_n}\|_{L^q(U)} \le C_q ( n+R^2)^{\frac1q-\hal}   \(W(\j_{\nu_n},\chi)+ n\)^{\frac 12}.
\end{equation}

  Since  $\nu_n'=0$ in  the support of $1-\chi$, we have
$$W(\j_{\nu_n},1-\chi) = \hal\int(1-\chi)|\j_{\nu_n}|^2\ge 0.$$
In particular $W(\j_{\nu_n},\chi)\le W(\j_{\nu_n},\chi)+ W(\j_{\nu_n},1-\chi) = W(\j_{\nu_n},\indic_{\mr^2})$. It then follows from \eqref{jp} and the fact that
$\widehat{\F}(\nu_n) = \frac1{\pi n} W(\j_{\nu_n},\indic_{\mr^2}) $  (cf. \eqref{fnhat})  that \eqref{jpp} holds.

By scaling we have $\int_\om|\nab H_n|^q =n^{\frac q2-1} \int_{ \om'}|\j_{\nu_n}|^q$, where $H_n= - 2\pi\Delta^{-1} \(\nu_n - n \mo\)$, while $\|\nu_n - n \mo\|_{W^{-1,q}(\om)}\le C \|\nab H_n\|_{L^q(\om)}$.  Thus \eqref{fluctu} follows. \end{proof}

The next proposition, whose proof will be given below, shows how $\widehat{F_n}$ controls $D(x_0', R)$.
}
\begin{pro}\label{5.1} Let $\nu_n=\sum_{i=1}^n \delta_{x_i}$, and $g_{\nu_n}$ be as in Definition~\ref{defG}. There exists a universal constant $R_0>0$ such that for any $R>R_0$, and any $x_0' = \sqrt{n} x_0\in \mr^2$,
we have
\begin{equation}\label{lgnu} \int_{B_{2R}(x_0')} \,d g_{\nu_n} \ge  c D(x_0',R)^2 \min\(1, \frac{|D(x_0',R)|}{R^2}\)-C R^2,\end{equation}
where $c>0$ and $C$ depend only on $V$, and where $D$ was defined in \eqref{dxr}.
\footnote{In fact $R_0$ could be any positive constant, and then $c,C$ would depend on $R_0$ as well, but this requires to adjust $\rho$ accordingly and we omit for simplicity to prove this fact.}\end{pro}

We now proceed to the
\begin{proof}[Proof of Theorem \ref{th3}]
We start by proving   \eqref{local}. If $R>R_0$ and $|D(x_0',R)|\ge \eta R^2$ then from Proposition \ref{5.1} and using the fact  ---  from Proposition \ref{wspread} --- that $g_{\nu_n}$  is positive outside $\cup_{i=1}^n B(x_i',\lambda)$ and that $g_{\nu_n}\ge -C$ everywhere, we deduce from   \eqref{fg} and \eqref{lgnu} that
\begin{equation}\label{pointf}\F(\nu_n) \ge  \frac{1}{n} \(- CR^2 + c \min( \eta^2 ,  \eta^3) R^4 \)+ 2\int \x\,d\nu_n .\end{equation}
Inserting into \eqref{loi2} we find
$$\Q\( \left| D(x_0',R)\right|\ge \eta R^2\)
\le\frac{1}{\K}  \exp\(C \beta  R^2 -c \beta\min(\eta^2, \eta^3)  R^4\) \int  e^{-n\beta\int\x\,d\nu_n} \, dx_1\dots dx_n.$$
Then, using   the lower bound \eqref{lbk} and Lemma~\ref{lemintxi} we deduce that if $\beta>\beta_0>0$ and $n$ is large enough depending on $\beta_0$ then
$$\log \Q\( \left|D(x_0',R)\right|\ge  \eta R^2\) \le  -c\beta\min(\eta^2,\eta^3) R^4+ C\beta R^2 + Cn\beta+ Cn , $$
where $c, C>0$ depend only on $V$. Thus \eqref{local} is \smallskip established.

We next prove \eqref{fubi}.
By Fubini's theorem, and using again  the facts that  $g_{\nu_n}$ is positive outside $\cup_{i=1}^n B(x_i',C)$ and $\ge -C$ everywhere  we have
$$\int_{ \mr^2} \,dg_{\nu_n} \ge \int_{U'}\( \dashint_{B(x', 2R)} \,dg_{\nu_n}\)\, dx'  - C | U'|- Cn .$$
Combining with Proposition \ref{5.1} it follows that
\begin{equation*}\int_{ \mr^2} \,dg_{\nu_n} \ge -C (|U'|+n) + \frac{1}{ R^2}\int_{U'} -C R^2 +cD(x',R)^2 \min \(1,\frac{|D(x',R)|}{R^2} \)   \, dx'.\end{equation*}
 i.e., changing the constants if necessary,
 \begin{equation}\label{Fmoy}
\int_{ \mr^2} \,dg_{\nu_n}\ge  -C(|U'|+n) + \frac{c}{R^2} \int_{U'} D(x',R)^2 \min \(1, \frac{|D(x',R)|}{R^2} \) \, dx .\end{equation}
 It follows, using as above \eqref{loi2}, \eqref{fg}, \eqref{lbk} and Lemma~\ref{lemintxi}, and since  $|U'| = n|U|$,  that
  $$\log \Q\( \int_{U'}   \frac{ D(x',R)^2}{R^2} \min \(1,\frac{|D(x',R)|}{R^2} \) \, dx  \ge
  \eta \) \le -c n\beta \eta + C n \beta \(|U|+1\)+ Cn,$$
  where $c,C>0$ depend only on $V$, where $\beta>\beta_0>0$ and where \smallskip $n>n_0(\beta_0)$.

 We next turn to \eqref{controlez}.  Arguing as above, from \eqref{fg} we have $\F(\nu_n) \ge - C  + 2  \int\x\,d\nu_n.$
Splitting $2\int \x \, d\nu_n $ as $\int \x\, d\nu_n+\int \x\, d\nu_n$, inserting into \eqref{loi2} and using \eqref{lbk} we are led to
$$\Q\( \int\x \, d \nu_n \ge \eta\) \le e^{ -\hal n \beta \eta + C n (\beta  +1)} \int e^{-n\beta \int \x\,d\nu_n} \, dx_1\dots dx_n,$$
where $C$ depends only on $V$. Then, using  Lemma \ref{lemintxi} we deduce \smallskip
\eqref{controlez}.

There remains to prove \eqref{cv}.
Reasoning as above after \eqref{pointf}, the probability that $\widehat{F_n}(\nu_n) \ge \eta$ is bounded above for any $\beta>\beta_0>0$ and $n$ large enough depending on $\beta_0$  by
$\exp(-\hal n\beta\eta + Cn(\beta+1))$, where $C$ depends on $V$ only. In view of  \eqref{fluctu}, it  follows that
$$\Q\(\(1+\frac{R^2}n\)^{\frac12-\frac1q} \|\nu_n - n \mo\|_{W^{-1,q}(B_{R/\sqrt n})}\ge C_q n^\hal (1+\eta)^\hal  \)\le \exp(-\hal n\beta\eta + Cn(\beta+1)).$$
After a slight rewriting, this concludes the proof of Theorem~\ref{th3}.

\end{proof}

\ed{
We now turn to the proof of Proposition \ref{5.1}.
The idea of the proof is the following: if $\div
\j=a(x) $ in a ball centered, say at $0$, then we can bound from below the contribution to the energy on circles as follows, using the Cauchy-Schwarz inequality and integration by parts
$$
\int_{\p B(0,t)} |\j|^2 \ge \frac{1}{2\pi t}\(\int_{\p B(0,t)} \j \cdot \vec{\nu}\)^2 \ge \frac{1}{2\pi t} \(\int_{B(0,t)} a(x)\, dx\)^2
.$$
Thus if we can bound from below the total charge in $B(0,t)$, i.e. $\int_{B(0,t)} a$, by integrating over circles we get a bound from below on $\int|\j|^2 $ which scales like the square of that charge. This is roughly how we expect the square of the discrepancy to appear in the right-hand side of \eqref{lgnu}.
This idea needs however two modifications in order to truly work: the first is that if there is a total charge, or charge discrepancy $D$ in a ball of radius $R$, we cannot be sure that the same holds in balls of radius $t$ different from $R$.
However, our charge density $a=\sum \delta_{x_i'}-\mu_0'$ has a particular structure: its negative part is bounded in $L^\infty$, so the charge discrepancy in $B(0,t)$ cannot decrease too quickly as $t $ moves away from $R$.
The second point is that what we need to bound from below is not $\int |\j|^2$ but $\int g$, via \eqref{lowgenu}, and so one may only use such lower bounds on circles that do not intersect the balls of $\B_\ro$. However it is not true in general  that one can find enough such circles. In order to go around this difficulty, instead of circles, we shall use curves that are defined as level lines of the distance to $\B_\ro$, this way they will automatically avoid the balls of $\B_\ro$. The co-area formula (see e.g. \cite{evgar}) will then be used to relate $\int  |\j|^2$ to the integrals along these curves.}

More precisely,
let us introduce a modified distance function to $x_0'$, as follows.
Two cases can be distinguished: either $D(x_0', R)>0$ or $D(x_0', R)\le 0$.
If $D(x_0', R)>0$, we let, for any $x$, $f(x)$ be the infimum over the set of curves $\gamma$ joining a point in $B_R(x_0')$ to $x$ of the length of $\gamma\sm\mathcal B_\rho$. This is also the distance to $x_0'$ for  the degenerate metric which is Euclidean outside $B_R(x_0')\cup \mathcal B_\rho$ and vanishes on $B_R(x_0')\cup\mathcal B_\rho$.

If $D(x_0', R)\le 0$ we define $f(x)$ to be the distance of $x$ to the complement of $B_R(x_0')$ with respect to  the metric which is Euclidean on $B_R(x_0')\sm\mathcal B_\rho$.

 We claim the following:
\begin{lem}\label{lemdistf}
In the first case,
if $|x-x_0'|\ge R+2$ then
$$\frac{|x-x_0'|-R}{4}\le f(x) \le |x-x_0'|-R.$$
In the second case, if $|x-x_0'|< R-2$ then
$$\frac{R -|x-x_0'|}{4}\le f(x) \le R - |x-x_0'|.$$
\end{lem}
\begin{proof} We start with the first case. The upper bound is obvious so we turn to the lower bound. Let $\gamma(t)$ be a continuous curve joining $x=\gamma(0)$ to $B_R(x_0')$. Let us build by induction a sequence $x^0=x, \dots, x^{K}$ with $x^{k+1}$ defined as follows: let $t_{k+1}$ be the smallest $t>t_k$ such that $\gamma(t)\notin B_1(x^k)\cap\gamma$ or $\gamma(t)\in B_R(x_0')$. This procedure terminates after a finite number of steps at $x^K\in\p  B_R(x_0')$.
By triangle inequality we have
\begin{equation}\label{triineq}
|x-x_0'|\le \sum_{k=0}^{K-1} |x^{k+1}-x^k| + |x^K - x_0'| \le K+ R.\end{equation}
On the other hand, by property of $\mathcal B_\rho$, for any $0\le k\le K-1$ we have
\begin{equation*}
\ell(\gamma[t_k,t_{k+1}]\sm\mathcal B_\rho) \ge |x^{k+1}-x^k| - 2\rho,\end{equation*} where $\ell$ denotes the length.
Summing this over $k$ and using \eqref{triineq}, we find
$$\ell(\gamma\sm\mathcal B_\rho)\ge K-1 -2 K \rho \ge |x-x_0'| - R-1 - 2\rho \(|x-x_0'| - R \).$$
Taking the infimum over all curves $\gamma$ we deduce that
$$f(x)\ge \(|x-x_0'| - R\) (1-2\rho) - 1.$$
Since by assumption $|x-x_0'|-R\ge 2$, we obtain $f(x)\ge \(|x-x_0'| - R\) (1-2\rho   - \frac{1}{2}) $ and the result follows since $\rho<1/8$.

The proof in the  second case is analogous.
\end{proof}

\begin{proof}[Proof of the proposition]
From Proposition~\ref{wspread}, defining $\j_{\nu_n}$ and $g_{\nu_n}$ as in Definition~\ref{defG}, we have
\begin{equation}\label{lowgenu}g_{\nu_n} \ge - C + \frac14|\j_{\nu_n}|^2\indic_{\mr^2 \sm \mathcal{B}_\ro},\end{equation}
where $\mathcal{B}_\ro$ is a set of
disjoint closed balls covering $\supp(\nu'_n)$,  and the  sum of the
radii of the balls in $\mathcal{B}_\ro$  intersecting any given ball
of radius $1$ is bounded by  $ \ro<\frac18$.

We distinguish again the  two cases $D(x_0',R)\ge 0$ or $D(x_0',R)<0$.
Let us start with the first case.


 Since  $f$, as defined above, is Lipschitz  with constant $1$, almost every $t$ is a regular value of $f$. For such a $t$ the curve $\gamma_t:=\{f = t\}$ is Lipschitz and does not intersect $\mathcal B_\rho$, since $\nabla f = 0$ there. Moreover, restating Lemma~\ref{lemdistf} we have
\begin{equation}\label{inclusion} f(x) < t\implies x\in B_{R+4t}(x_0')\cup B_{R+2}(x_0'),\end{equation}
thus $\gamma_t\subset B_{2R}(x_0')$ if $R+4t < 2R$, i.e. if $t<R/4$, and $R+2<2R$, i.e. if $R>2$. It follows from \eqref{lowgenu} that if $R>2$ then
\ed{
\begin{equation}\label{unelow}\int_{B_{2R}(x_0')}\,d g_{\nu_n}\ge -CR^2 +\frac{1}{4} \int_{\{0<f(x)<R/4\}}|\j_{\nu_n}|^2 .\end{equation}
On the other hand the co-area formula (see e.g. \cite{evgar}) asserts that
$$\int_{\{0<f(x)<R/4\}}|\j_{\nu_n}|^2= \int_{0}^{R/4} \(\int_{\gamma_t} \frac{|\j_{\nu_n}|^2}{|\nabla f|}\, d\mathcal{H}^1\) dt.$$ Since $f$ is $1$-Lipschitz, it follows that }
\begin{equation}\label{unelow}\int_{B_{2R}(x_0')}\,d g_{\nu_n}\ge -CR^2 + \frac14\int_{t=0}^{R/4}\(\int_{\gamma_t}|\j_{\nu_n}|^2\, d\mathcal{H}^1\)\,dt.\end{equation}
We proceed by estimating the innermost integral on the right-hand side. If $t>1/2$ then $B_{R+2}\subset B_{R+4t}$ and using \eqref{inclusion} we find,  using Definition \ref{defG} and writing $D$ instead of $D(x_0',R)$ in the course of this proof,
\begin{multline*}\int_{\gamma_t}
\j_{\nu_n}\cdot \vec{\nu} = \int_{\{f< t\}}\div \j_{\nu_n} \ge  2\pi\nu_n\(B_{\frac R{\sqrt n}}(x_0)\) - 2\pi  n\mo\(B_\frac{R+4t}{\sqrt n}(x_0)\)\\ \ge 2\pi D - C \((R+4t)^2 - R^2\),\end{multline*} where we have used \eqref{bornmo}.
The right-hand side is bounded below by $\pi D$ if
$R+4t < \sqrt{R^2+cD}$ and $c$ is small enough. Thus, if
\begin{equation}\label{Dcondition} 1/2 < T:=\frac R4\(\sqrt{1+c\frac D{R^2}} - 1\),\end{equation}
it follows using Cauchy-Schwarz's inequality that for every  $t\in (1/2,T)$ we have
\begin{equation}\label{cercles}\int_{\gamma_t} |\j_{\nu_n}|^2\, d\mathcal{H}^1 \ge  \frac{\pi^2 D^2}{\mathcal{H}^1(\gamma_t)},\end{equation}
Inserting into \eqref{unelow} while reducing integration to  $1/2<t<\min(T,R/4)$ we obtain
\begin{equation}\label{avJ}\int_{B_{2R}(x_0')}\,d g_{\nu_n}\ge -CR^2 + \pi^2 D^2 \int_{t=1/2}^{\min(T,R/4)} \frac{1}{\mathcal{H}^1(\gamma_t)}\,dt.\end{equation}
The evaluation of the last integral is complicated by the fact that $\gamma_t$ is the level set for $f$ rather than the usual circle, however the result will be comparable to the one we would get if we had $\mathcal{H}^1(\gamma_t) = 2\pi (R+t)$, this is proven as follows:
From Lemma~\ref{lemdistf} we have for every $t\in [1/2, \min(T,R/4)]$ that $\gamma_t\subset\{x: 0<|x-x_0'|-R < \min(4T,R)\}$. From the coarea formula and  the fact that $|\nabla f| \le 1$ it follows that
$$\int_{\hal}^{\min(T,R/4)} \mathcal{H}^1(\gamma_t) \,dt \le |\{x: R<|x-x_0'|< R+  \min(4T,R)\}| \le C R \min(4T,R).$$Then, using the convexity of $x\mapsto1/x$ and Jensen's inequality in \eqref{avJ} we obtain for some $c>0$ that
$$\int_{\hal}^{\min(T,R/4)} \frac1{ \mathcal{H}^1(\gamma_t)}\,dt\ge c\frac{\(\min(T,R/4) - \hal\)^2}{R \min(4T,R)}.$$
Inserting into \eqref{avJ} we obtain assuming \eqref{Dcondition} and
\begin{equation}\label{dcb} 1<\min(T,R/4)\end{equation} that
\begin{equation}\label{apJ}\int_{B_{2R}(x_0')}\,d g_{\nu_n}\ge -CR^2 + c \frac{D^2}{R} \(\min(T,R/4) - \hal\).\end{equation}

One may check that  \eqref{Dcondition}, \eqref{dcb} are satisfied if $R>4$ and $D>C_0R$ for a large enough $C_0>0$. Then it is not difficult to deduce \eqref{lgnu} from \eqref{apJ} by distinguishing the cases $T<R/4$ and $T\ge R/4$, i.e. $D<C_1R^2$ and $D\ge C_1R^2$ for a well chosen $C_1$. Finally,  if $D<C_0 R$ then \eqref{lgnu} is trivially satisfied, if $C$ is chosen large enough.

Let us turn to the case $D(x_0',R)\le 0$, which implies $|D(x_0',R)|\le n \mu_0(B_{R/\sqrt n}(x_0)) \le \pi  R^2  \wb$.
As above,  if $R>2$ and for almost every $1/2<t<R/4$ the curve $\gamma_t = \{f=t\}$ is a Lipschitz curve which  does not intersect $\mathcal B_\rho$ and $ \{f<t\}\subset B_R(x_0')\sm B_{R-4t}(x_0')$.  It follows that, writing as before $D$ for $D(x_0,R)$
\begin{multline}\label{neg}\int_{\gamma_t} \j_{\nu_n}\cdot \tau = \int_{\{ f<t\}}\div \j_{\nu_n} =
\int_{B_R(x_0')}\div \j_{\nu_n} -\int_{B( x_0', R)\sm \{f>t\}} \div \j_{\nu_n}
\\ \le  2\pi D  + 2\pi n \mo\big(B_\frac{R}{\sqrt{n}}\sm B_{ \frac{ R-4 t}{\sqrt{n}} }\big) \le 2\pi D + C  \(R^2 - (R-4t)^2\).\end{multline}
The proof then proceeds as in the first case by using Cauchy-Schwarz's inequality and integrating with respect to $t\in[1/2,\min(T,R/4)]$, where
$$ T=\frac R4\(1 - \sqrt{1+c\frac D{R^2}}\),$$
which ensures that the right-hand side in \eqref{neg} is bounded above by $\pi D$. Note that $D$ is nonpositive, but bounded below by $-CR^2$ hence if $c>0$ is small enough the quantity inside the square root above is positive.
\end{proof}

\section{Lower bounds via the ergodic theorem and conclusions}

\subsection{Abstract result via the ergodic theorem}\label{ergo}
In this section, we present the ergodic framework introduced in \cite{ss1} for obtaining ``lower bounds for 2-scale energies" and inspired by Varadhan. We cannot directly use the result there because it is written for a uniform ``macroscopic environment", which would correspond to the case where $m_0(x)$ is constant on its support (as in the circular law). To account for the possibility of varying environment or weight at the macroscopic, we can however adapt Theorem~3  of \cite{ss1} and easily prove the following variant:

Let $X$ denote a Polish metric space, when we speak of measurable functions on $X$ we will always mean Borel-measurable. We assume there is a $d$-parameter group of transformations  $\theta_\lambda$ acting continuously on $X$. More precisely we require that
\begin{itemize}
\item[-] For all $u\in X$ and $\lambda, \mu\in \mr^d$, $ \theta_\lambda (\theta_\mu u)=\theta_{\lambda+\mu} u$, $\theta_0 u=u$.
\item[-] The map $(\lambda,u)\mapsto \theta_\lambda u$ is continuous  with respect to each variable (hence measurable with respect to both).
\end{itemize}
Typically we think of $X$ as a space of functions defined on $\mr^d$ and $\theta$ as the action of translations, i.e. $\theta_\lambda
u(x)=u(x+\lambda)$. Then we consider the following   $d$-parameter  group of transformations
$T^\ep_\lambda $ acting continuously on $\mr^d \times X$ by
$T^\ep_\lambda (x, u)= (x+ \ep \lambda , \theta_\lambda u)$. We also define $T_\lambda(x, u )= (x, \theta_\lambda u)$.

For a probability measure $P$ on $\mr^d\times X$ we say that $P$ is translation-invariant if it is invariant under the action $T$, and we say it is $T_{\lambda(x)}$-invariant if for every function $\lambda(x) $ of class $C^1$, it is invariant under the mapping $(x,u)\mapsto(x, \theta_{\lambda(x)} u)$. Note that $T_{\lambda(x)}$-invariant implies translation-invariant. \smallskip


Let $G$ denote a  compact set in $\mr^d$ such that
\begin{equation}\label{pG}|G|>0,\quad \lim_{\ep\to 0} \frac{|(G+\ep x)\sd G|}{|G|} = 0,\end{equation}
for every $x\in\mr^2$, where $\sd$ denotes the symmetric difference of sets.  We let $\{f_\ep\}_\ep$ and $f$ be measurable nonnegative
functions on $G \times X$, and assume that for any family $\{(x_\ep, u_\ep)\}_\ep$
such that
$$\forall R>0, \quad \limsup_{\ep\to 0}\int_{B_R} f_\ep(T^\ep_\lambda( x_\ep,u_\ep ))\, d\lambda<+\infty$$
the following holds.
\begin{enumerate}
\item(Coercivity) $\{(x_\ep, u_\ep)\}_\ep$ admits a convergent subsequence (note that $\{x_\ep\}_\ep$ subsequentially converges since $G$ is compact).
\item ($\Gamma$-liminf)   If 
$\{(x_\ep,u_{\ep})\}_\ep$  converges to $(x,u)$ then
$$\liminf_{\ep \to 0} f_\ep(x_\ep, u_\ep) \ge f(x,u).$$
\end{enumerate}

Then for the sake of generality we consider an increasing  family of bounded open sets $\{\UR\}_{R>0}$ such that
\begin{equation}\label{hypsets}
\text{ (i) $\{\UR\}_{R>0}$ is a Vitali family,}\quad (ii) \text{$\D\lim_{R\to +\infty} \frac{|(\lambda + \UR)\sd \UR|}{|\UR|} = 0$} \end{equation}
for any $\lambda\in\mr^d$, where Vitali means  (see \cite{nmriv})  that the intersection of the closures is $\{0\}$, that $R\mapsto|\UR|$ is left continuous, and that $|\UR - \UR|\le C|\UR|$.

We have

\begin{theo}\label{gamma} Let $G$, $X$, $\{\theta_\lambda\}_\lambda$, and $\{f_\ep\}_\ep$, $f$ be as above. For any $u\in X$, let
$$F_\ep (u) = \dashint_{G} f_\ep(   x,  \theta_{\frac
{x}{\ep}}  u ) \,dx.$$
\ed{and  let $\phi_\ep(u)$ be the probability on $G \times X$ which is the image of the normalized Lebesgue measure on $G$  under
the map $x \mapsto(x , \theta_{\frac{x}{\ep} } u)$.}\\
A. Assume that \ed{$\{u_\ep\}_{\ep}$, a family of elements of $X$,  is such that  $\{F_\ep(u_\ep)\}_\ep$ is bounded, and  let  $P_\ep  =\phi_\ep(u_\ep)$.}
Then $P_\ep$
converges to  a Borel probability measure $P$ on $G \times X$ whose first marginal is the normalized Lebesgue measure on $G$, which is $T_{\lambda(x)}$-invariant, such that $P$-a.e. $(x,u)$ is of the form  $\lim_{\ep\to 0}(x_\ep, \theta_{\frac{x_\ep}{\ep}} u_\ep)$  and  such that
\begin{equation}
\label{rg1} \liminf_{\ep \to 0} F_\ep(u_\ep) \ge \int f( x, u)\, dP(x,u).
\end{equation}
Moreover,
\begin{equation}
\label{rg2} \int f(x,u)\, dP(x,u)= 
\mathbf{E}^P\(\lim_{R\to
+\infty}\dashint_{\UR} f(x,\theta_\lambda  u  )\,d\lambda\)
,\end{equation} where $\mathbf{E}^P$ denotes the
expectation under the probability $P$.\\
\ed{B. Let $\mathbb{P}_\ep$ be a probability on $X$  such that $\lim_{M\to +  \infty}\lim_{\ep\to 0}\mathbb{P}_\ep\(\{F_\ep(u)\ge M\}\)=0$, then $\{\phi_\ep \# \mathbb{P}_\ep\}_{\ep}$ is tight, i.e. converges up to a subsequence to  a probability measure on $\P(G\times X)$.}
\end{theo}

The proof  uses the following simple lemma,  whose statement and  proof can be found in \cite[Lemma 2.1]{ss1}.
\begin{lem}\label{lesigne}
Assume $\{P_n\}_n$ are Borel probability measures on a Polish metric space $X$
and that for any $\delta>0$ there exists $\{K_n\}_n$ such that $P_n(K_n)\ge 1-\delta$ for every $n$ and such that if $\{x_n\}_n$ satisfies for every $n$ that $x_n\in K_n$, then any subsequence of $\{x_n\}_n$ admits a convergent subsequence (note that we do not assume $K_n$ to be compact).  Then $P_n$ admits  a subsequence which converges tightly, i.e. converges weakly to a probability measure $P$.\end{lem}

\begin{proof}[Proof of the theorem] It follows the steps of \cite[Section 2]{ss1}:
\begin{enumerate}
\item  $P_\ep$ is tight hence has a limit $P$. This follows from the coercivity property of $f_\ep$ as in \cite[Section 2, Step 1]{ss1} and uses Lemma \ref{lesigne}.
\item  $P$ is $T_{\lambda(x)}$-invariant. Let $\Phi$ be bounded and continuous, and let $P_\lambda $ be the push-forward of $P$ by $(x, u) \mapsto (x, \theta_{\lambda(x)} u)$. Then from the
definition of  $P_\lambda$, $P$, $P_\ep$, we have
\begin{multline*}\int \Phi(x,u) \, dP_\lambda(x,u) =   \int \Phi(x,\theta_{\lambda(x)} u) \, dP(x,u)  = \lim_{\ep \to 0} \int \Phi(x, \theta_{\lambda(x)} u) \, dP_\ep(x,u)=\\  \lim_{\ep \to 0} \dashint_G \Phi(x, \theta_{\frac x\ep+\lambda(x)} u_\ep) \, dx = \lim_{\ep\to 0}\dashint_{(I+\ep\lambda )(G)} \frac{\Phi( (I+\ep\lambda)^{-1}(y), \theta_{\frac y\ep} u_\ep)}{|\det(I +\ep D \lambda((I+\ep \lambda)^{-1}(y)) |} \, dy,\end{multline*}
where the last equality follows by the change of  variables $y= (I+\ep \lambda)(x)$. Using the boundedness of $\Phi$, the $C^1$ character of $\lambda$, the compactness of $G$ and \eqref{pG}, we may replace $(I+\ep \lambda) (G)$ by $G$ and the denominator by $1$ in the last integral  and we find, using the definition of $P_\ep$
\begin{equation}\label{multi} \int \Phi(x,u) \, dP_\lambda(x,u) =    \lim_{\ep\to 0}\int \Phi((I+\ep\lambda)^{-1}(x),u) \, dP_\ep(x,u).\end{equation}
Since $\{P_\ep\}_\ep$ is tight, for any $\delta>0$ there exists $K_\delta$ such that $P_\ep({K_\delta}^c)<\delta$ for every $\ep$. Then by uniform continuity of $\Phi$ on $K_\delta$ the map $(x,u)\mapsto\Phi( (I+\ep \lambda)^{-1} (x),u)$ converges uniformly on $K_\delta$ to $(x,u)\mapsto\Phi(x,u)$ and thus
$$\lim_{\ep\to 0}\int_{K_\delta} \Phi((I+\ep\lambda )^{-1}(x),u) \, dP_\ep(x,u) = \lim_{\ep\to 0}\int_{K_\delta} \Phi(x,u) \, dP_\ep(x,u).$$
Since this is true for any $\delta>0$, and using the boundedness of $\Phi$  we get
$$\lim_{\ep\to 0}\int\Phi( ( I+\ep\lambda )^{-1}(x),u) \, dP_\ep(x,u) = \lim_{\ep\to 0}\int\Phi(x,u) \, dP_\ep(x,u) = \int \Phi(x,u)\,dP(x,u),$$
by definition of $P$. Thus in view of \eqref{multi} we have $P_\lambda = P$  and $P$ is thus $T_{\lambda(x)}$-invariant.

\item $\liminf_{\ep \to 0} F_\ep(u_\ep) \ge \int f \, dP$. This follows  from
 \cite[Lemma 2.2]{ss1}, since
 $F_\ep(u_\ep) = \int f_\ep \, dP_\ep$.
\end{enumerate}
To conclude, as in \cite[Section 2]{ss1}, the fact that $P$ is $T_{\lambda(x)}$-invariant (which implies $T_\lambda$-invariant) and Wiener's multiparametric ergodic theorem  (see e.g. \cite{becker}) imply that
$$\int f(x,u)\, dP(x,u)= \mathbf{E}^P\(\lim_{R\to
+\infty}\dashint_{\UR} f(T_\lambda(x,
u))\,d\lambda\)=\mathbf{E}^P\(\lim_{R\to
+\infty}\dashint_{\UR} f(x,\theta_\lambda  u  )\,d\lambda\).$$
\ed{We now turn to the proof of B.
Let $A_{M,\ep}=\{u\in X, F_\ep(u) \le M\}$. Then we have
$\phi_\ep\#\mathbb{P}_\ep( \phi_\ep(A_{M,\ep}^c))=\mathbb{P}_\ep(A_{M,\ep}^c)\to 0$ as $\ep \to 0$ and $M\to \infty$. In view of Lemma \ref{lesigne} applied with $K_n= \phi_\ep(A_{M,\ep})$, in order to prove the tightness of $\phi_\ep\#\mathbb{P}_\ep$ it suffices to take $M$ large enough and check that if $P_\ep \in \phi_\ep (A_{M,\ep})$ then $P_\ep $ has a convergent subsequence. But this is a direct application of what we have established in part A, since such a $P_\ep$ is the image by $\phi_\ep$ of a family $u_\ep$ for which $F_\ep(u_\ep)\le M$.  Therefore $P_\ep$ is tight and $\phi_\ep\#\mathbb{P}_\ep$ as well by the lemma.
}

\end{proof}

We now apply this abstract framework to our specific situation to obtain the lower bound on $\widehat{\F}$.

\subsection{Proof of Theorem \ref{th2}, part A} The proof follows essentially \cite{ss1}, Proposition 4.1 and below. Let  $\{\nu_n\}_n$ and $P_{\nu_n}$ be as in the statement of Theorem~\ref{th2}. We need to prove that any subsequence of $\{P_{\nu_n}\}_n$ has a convergent subsequence and that the limit $P$ is a $T_{\lambda(x)}$-invariant  probability measure such that $P$-almost every $(x,\j)$ is such that  $\j\in\mathcal A_{\bm(x)}$ and  \eqref{thlow} holds. Note that the fact that the first marginal of $P$ is $dx_{|\E}/|\E|$ follows from the fact that, by definition, this is true of $P_{\nu_n}$.

We thus take a subsequence of $\{P_{\nu_n}\}$ (which we don't  relabel).  We may  assume that  it has  a subsequence, denoted $\bnun$,  which satisfies  $\widehat{\F}(\bnun)\le C $, otherwise there is nothing to prove. This implies that  $\bnun$ is of the form $\sum_{i=1}^n \delta_{x_{i}}$.  We let  $\bjn$ denote the current and $\bgn$ the measures associated to $\bnun$ as in Definition~\ref{defG} and note that $\int d\bgn=W(\bjn, \indic_{\mr^2})$.
 As usual, $\bnun'= \sum_{i=1}^n \delta_{\sqrt n x_{i}}$.

 \ed{
 A first  consequence of $\widehat{\F}(\bnun)\le C $ is that, in view of  \eqref{fluctu}, we have
\begin{equation}\label{cvnu} \frac1n\bnun\to\mo,\end{equation}
in the weak sense of measures.}

\subsubsection*{Step 1: We set up the framework of Section~\ref{ergo}} We will use integers $n$ instead of $\ep$ to label sequences, and the correspondence will be $\ep = 1/\sqrt n$. We let $G = \E$ and $X= \radon_+\times \Lp \times  \radon$, where $p\in (1,2)$, 
where $\radon_+$ denotes the set of  positive Radon measures on $\mr^2$ and $\radon$ the set of those which are bounded below by the  constant $ -C (\|\bm\|_\infty +1)$  of Proposition \ref{wspread}, both equipped with the topology of weak convergence.

For $\lambda\in\mr^2$ and abusing notation we let $\theta_\lambda$ denote both  the translation $x\to x+\lambda$ and the action
$$\theta_\lambda (\nu,\j,g) = \(\theta_\lambda\#\nu, \j\circ\theta_\lambda,\theta_\lambda\# g\).$$
Accordingly the action $T^n$ is defined for $\lambda\in\mr^2$ by
$$T^n_\lambda (x,\nu,\j,g) = \(x+\frac\lambda{\sqrt n}, \theta_\lambda\#\nu, \j\circ\theta_\lambda,\theta_\lambda\# g\).$$
Then we let $\chi$ be a smooth cut-off function with integral $1$ and support in $B(0, 1)$ and
define
\begin{equation}\label{deffn} \f_n(x, \nu,\j, g) =\begin{cases} \D\frac{1}{ \pi} \int_{\mr^2}\chi(y)\, dg(y) & \text{if $(\nu,\j,g) = \theta_{\sqrt n x} (\bnun', \bjn, \bgn)$,}\\ +\infty & \text{otherwise.}\end{cases}
\end{equation}
Finally we let, in agreement with Section~\ref{ergo},
\begin{equation}\label{deF} \gF_n(\nu,\j, g) =\dashint_\E \f_n\(x,\theta_{x\sqrt n}(\nu,\j,g)\)\,dx.
\end{equation}
We have the following relation between $\gF_n$ and $\widehat{\F}$, as $n\to +\infty$:
\begin{equation}\label{FF}
\text{ $\gF_n(\nu,\j,g) $ is  }\quad \begin{cases} \le \frac1{|\E|}\widehat{\F}(\bnun)+ o(1) & \text{if $(\nu,\j,g) = (\bnun',\bjn,\bgn)$}\\  = +\infty & \text{otherwise}.\end{cases}
\end{equation}
Indeed it is obvious from \eqref{deffn}  that if  $(\nu,\j,g) \neq (\bnun',\bjn,\bgn)$ then $\gF_n(\nu,\j,g) = +\infty$. On the other hand, if  $(\nu,\j,g) = (\bnun',\bjn,\bgn)$, then from the definition of the image measure $\theta_\lambda\# \bgn$,
$$\gF_n(\nu,\j,g) =\frac1\pi \dashint_\E \int \chi(y-x\sqrt n) \,d\bgn(y)\,dx = \frac1{\pi|\E'|} \int \chi * \indic_{\E'} \,d\bgn.$$
Since $\chi * \indic_{\E'}$ is bounded above by $1$ and is equal to $1$ on $U := \{x': \dist(x',\mr^2\sm \E')\ge 1\}$ we deduce that
\begin{multline}\label{preFF}\pi\gF_n(\nu,\j,g) \le \frac{\bgn^+(\mr^2) - \bgn^-(U)}{|\E'|}  =   \frac{\bgn(\mr^2) + \bgn^-(U^c)}{n|\E|} \\ =  \frac{\pi\widehat{\F}(\bnun)} {|\E|} +  \frac{\bgn^-(U^c)}{n|\E|} .\end{multline}
Then we note that  from \eqref{lbg}--\eqref{gcj} in Proposition~\ref{wspread} the measure $\bgn^-$ is supported in the union of balls $B(x',C)$ for $x'\in\supp( \bnun')$, and bounded above by a constant. Thus $\bgn^-(U^c)$ is bounded by a constant times the number of balls intersecting $U^c$, hence by $C \bnun'\{x': \dist(x',U^c)\le C\}$. From \eqref{cvnu} this is equal to \ed{
$$C n \mo\{x: \dist(x,\partial \E)\le C/\sqrt n\} + o(n)\le Cn|\{x: \dist(x,\partial \E)\le C/\sqrt n\}|+o(n)$$ since $\bm$ is bounded.
Using standard estimates on the volumes of tubular neigborhoods, since $\partial \E$ is $C^1$ by  assumption \eqref{assumpV3}, we conclude that this is $o(n)$}.  Plugging this into \eqref{preFF} proves \eqref{FF}.

\subsubsection*{Step 2: We check the hypotheses in  Section~\ref{ergo}}
We must now check the $\Gamma$-liminf and coercivity properties of $\{\f_n\}_n$. The main point is again that $\widehat{\F}$ controls $\nu_n-n\mu_0$ by Lemma \ref{lemnew}.

\begin{lem}\label{liminf} Assume that $\{(x_n,\nu_n,\j_n,g_n)\}_n$ converges to $(x,\nu,\j,g)$. Then
$$\liminf_n\f_n\(x_n, \nu_n,\j_n, g_n\)\ge \f(x,\nu,\j,g) := \frac1\pi\int \chi\,dg.$$
\end{lem}
\begin{proof} We may assume that the left-hand side is finite, in which case $\f_n\(x_n, \nu_n,\j_n, g_n\) = \frac1\pi\int\chi\,dg_n$ for every large enough $n$, from which the result follows by passing to the limit.
\end{proof}

\begin{lem}\label{coer} Assume that for any $R>0$ we have
\begin{equation}\label{hcw} \limsup_{n\to+\infty}\int_{B_R} \f_n\(x_n+\frac\lambda{\sqrt n}, \theta_\lambda(\nu_n,\j_n, g_n)\)   \, d\lambda < +\infty.\end{equation}
Then a subsequence of $\{(x_n,\nu_n,\j_n,g_n)\}_n$ converges to some $(x,\nu,\j,g)\in \E\times X$.
\end{lem}
\begin{proof}
Assume \eqref{hcw}. Then the integrand there is bounded for a.e. $\lambda$ and from the definition \eqref{deffn} we deduce that
$$\theta_\lambda(\nu_n,\j_n, g_n) = \theta_{\sqrt n x_n+\lambda} (\bnun', \bjn, \bgn)$$
and then that $( \nu_n,\j_n, g_n) = \theta_{\sqrt n x_n} (\bnun', \bjn, \bgn).$ Thus \eqref{hcw} gives, in view of \eqref{deffn}, that for every $R>0$ there exists $C_R>0$ such that for any $n$
$$ \int_{B_R} \int \chi(y-\sqrt n x_n - \lambda)\,d\bgn(y)\,d\lambda = \int \chi * \indic_{B_R(\sqrt n x_n)}\,d\bgn < C_R.$$

This and the fact that $\bgn$ is bounded below implies that $\bgn(B_R(\sqrt n x_n))$ is bounded independently of $n$ and then, using \eqref{bgnalpha}, that the same is true of $\bnun'(B_R(\sqrt n x_n))$. In other words  $\{\nu_n = \theta_{\sqrt n x_n}\bnun'\}_n$ is a locally bounded sequence of (positive) measures hence converges weakly after taking a subsequence, and the same is true of $\{g_n = \theta_{\sqrt n x_n}\bgn\}_n$.  On the other hand $\{x_n\}_n$ is a sequence in the compact set $\E$ hence converges modulo a subsequence.

It remains to study the convergence of $\{\j_n = \bjn\circ\theta_{\sqrt n x_n+\lambda}\}_n$.  From \eqref{wg} in Proposition~\ref{wspread} and the local boundedness of $\{\nu_n\}_n$ we get that $W(\bjn,\chi*\indic_{B_R(\sqrt n x_n)} )= W(\j_n,\chi*\indic_{B_R})$ is bounded independently of $n$ for any $R>0$ and then, using \eqref{firstitem}, that $\{\j_n\}_n$ is locally bounded in $\Lp$, for any $1\le p<2$ hence a subsequence locally weakly converges in $\Lp$. \ed{Moreover, $\curl \j_n=0$ and by the above $\div \j_n$ is locally bounded in the sense of measures, hence weakly compact in $W^{-1, p}_{loc}$ for $p<2$.   By elliptic regularity, it follows that the convergence of $\j_n$   is strong in $\Lp$.  This concludes the proof of coercivity.

}
\end{proof}

\subsubsection*{Step 3: Conclusion}  From the previous steps, we may apply Theorem~\ref{gamma} in this setting (choosing $\UR=K_R$) and we deduce in view of \eqref{FF} that, temporarily denoting $Q_n$ denote the push-forward of the normalized Lebesgue measure on $\E$ by the map $x\mapsto (x,\theta_{\sqrt n x} (\bnun', \bjn,\bgn))$, and $Q = \lim_n Q_n,$
\begin{multline}\label{eF}\liminf_n \frac1{|\E|}\widehat{\F}(\bnun) \ge \liminf_n \gF_n(\bnun',\bjn,\bgn) \ge \\ \int\(\frac{1}{\pi}\int\chi\,dg\) \,dQ(x,\nu,\j,g) = \int \lim_{R\to +\infty} \dashint_{K_R} \int \frac{1}{\pi}\chi(y-\lambda)\,dg(y)\,d\lambda\,dQ(x,\nu,\j,g) =\\  \int \lim_{R\to +\infty}\(\frac1{\pi |K_R|}\int  \chi*\indic_{K_R}\,dg \)\,dQ(x,\nu,\j,g).  \end{multline}
Now we use the fact that for $Q$-almost every $(x,\nu,\j,g)$:
\begin{itemize}
\item[i)] There exists a sequence $\{x_n\}_n$ in $\E$ such that $(x,\nu,\j,g)= \lim_n (x_n,\theta_{\sqrt n x_n} (\bnun', \bjn, \bgn))$.
 \item[ii)]  As a consequence of the above $\frac1{\pi |K_R|}\int  \chi*\indic_{K_R}\,dg$ converges to a finite limit  as $R\to +\infty$.
 \end{itemize}
 The first point implies, since $\div \bjn = \bnun'-\bm'$ and $\curl  \bjn = 0$, that by passing to the limit $n\to \infty$ we have $\div \j = \nu - \bm(x)$ and $\curl \j = 0$.
 The second point implies in particular using \eqref{bgnalpha} that $\nu(B_R) \le CR^2$,  proving that $(\j,\nu)\in\mathcal A_{\bm(x)}$.

 Moreover the second point implies that  for any $C>0$ we have $g(K_{R+C}\sm K_{R-C}) =o(R^2)$ as $R\to +\infty$, and thus  from point i) above
$$ \lim_{R\to +\infty}\lim_{n\to+\infty}\frac1{R^2} \bgn(K_{R+C}(\sqrt n x_n))\sm K_{r-C}(\sqrt n x_n)) = 0. $$
Using \eqref{bgnalpha} we deduce that
$$ \lim_{R\to +\infty}\lim_{n\to+\infty}\frac1{R^2} \bnun'(K_{R+C}(\sqrt n x_n))\sm K_{r-C}(\sqrt n x_n)) = 0 $$
 and then from \eqref{wg},
 $$ \lim_{R\to +\infty}\lim_{n\to+\infty}\frac1{R^2} \left| W(\bjn, \chi*\indic_{K_R(\sqrt n x_n)})- \int \chi * \indic_{K_R(\sqrt n x_n)}\,d\bgn \right| = 0. $$
 Thus,  using  \cite[Lemma 4.8]{ss1} we may take the   limit $n\to \infty$ and   deduce
 $$ \lim_{R\to +\infty}\frac1{R^2} \left| W(\j, \chi*\indic_{K_R}) -\int \chi * \indic_{K_R}\,dg \right| = 0. $$
 Together with \eqref{eF} this yields, by definition of $W$,
\begin{equation}\label{thlow1}\liminf_n  \frac1{|\E|}\widehat{\F}(\bnun)  \ge  \frac{1}{\pi} \int W(\j) \,dQ(x,\nu,\j,g)\end{equation}
and, we recall,  $Q$-a.e. $(\j,\nu)\in \mathcal A_{\bm(x)}$.

Now we let $P_n$ (resp. $P$) be the marginal of $Q_n$ (resp. $Q$) with respect to the variables $(x,\j)$. Then  the first marginal of $P$ is the normalized Lebesgue measure on $\E$ and $P$-a.e. we have $\j\in\mathcal A_{\bm(x)}$, in particular
$$W(\j)\ge \min_{\mathcal A_{\bm(x)}} W = \bm(x) \(\min_{\mathcal A_1} W  - \frac{\pi}{2}\log \bm(x)\).$$ Integrating with respect to $P$ and noting that since only $x$ appears on the right-hand side we may replace $P$ by its first marginal there we find, in view of \eqref{defa} that the lower bound  \eqref{thlow} holds.

\subsection{Proof of Theorem \ref{th2}, completed}
As mentioned above, Part B of the theorem is a direct consequence of Proposition \ref{construct}, see Corollary~\ref{etanul}.

Part C follows from the comparison of Parts A and B: for minimizers, the chains of inequalities  \eqref{thlow} and \eqref{thup} are in fact equalities and  $\frac{1}{\pi} \int W\, dP $ must be minimized hence equal to $\alpha$. Also  we must have  $\lim_{n\to \infty} (\F(\nu_n) -\widehat{\F})(\nu_n)= \lim_{n\to \infty} \int \zeta \, d\nu_n=0$, which in view of  \eqref{lemz}, implies that $\lim \sum_i \dist (x_i, \E)^2=0$.

From the fact that $\widehat{\F}(\nu_n)=O(1)$, we deduce from Proposition \ref{5.1}, \ref{wspread} and \eqref{fg}, (arguing exactly as in the proof of Theorem \ref{th3}) that there exists $C>0$ such that  for every $x, R>1$, we have
$$D(x, R)^2 \min \(1, \frac{|D(x,R)|}{R^2}\)\le C n.$$
This completes the proof of Theorem \ref{th2}.

\subsection{Deviations: proof of Theorems \ref{th4} and Theorem \ref{valz}}
  We start with the upper bound on $\log \Q$. Let $A_n$ be a subset of $ \mc^n$. We identify points in $\mc^n$ with measures $\nu_n$ of the form $\sum_{i=1}^n \delta_{x_i}$.

From \eqref{loi2},   we have
$$\Q(A_n) =   \frac{1}{\K} \int_{A_n} e^{-\hal \beta n \F(\sum_{i=1}^n \delta_{x_i} )} \, dx_1 \dots dx_n$$ hence
\begin{equation}\label{logq}
\frac{\log \Q(A_n)}{n} =- \frac{\log \K}{n} +\frac1n \log \int_{A_n} e^{-\hal \beta n \F(\sum_{i=1}^n \delta_{x_i} )} \, dx_1 \dots dx_n.\end{equation}
We deduce, since $\widehat{F}_n(\nu_n) = \F(\nu_n) - 2\int\x\,d\nu_n$, that
\begin{equation}\label{logq2}
\frac{\log \Q(A_n)}{n} \le  - \frac{\log \K}{n}  + \frac{1}{n}    \log\( e^{-\hal \beta n \inf_{A_n}\widehat{F}_n} \int_{A_n} e^{-\beta n  \int \x\,d\nu_n} \, dx_1 \dots dx_n\)  .\end{equation}
Let $\nu_n$ such that $\widehat F_n(\nu_n)\le \inf_{A_n}\widehat{F}_n +1/n$.  Then  from \eqref{thlow} in Theorem~\ref{th2}  we have, using the notations there,
$\liminf_{n \to \infty} \widehat{F}_n(\nu_n) \ge \frac{|\E|}{\pi} \int W(\j)\, dP(x, \j)$ where  $P = \lim_n P_{\nu_n}$.  Since $P_{\nu_n}\in i_n(A_n)$ by definition we have  $P \in A_{\infty}$  since by definition  $A_{\infty} = \cap_{n>0}\overline{\cup_{m>n} i_m(A_m)}$ .  We may thus write
\begin{equation}\label{whf}
 \liminf_{n\to+\infty} \widehat{F}_n(\nu_n)\ge \frac{|\E|}{\pi}\inf_{P \in A_{\infty} }\int W(\j)\, dP(x, \j) .\end{equation}
Inserting into \eqref{logq2} we are led to
\begin{multline}\label{logq3}
\frac{\log \Q(A_n)}{n} \le   -  \frac{\beta|\E|}{2\pi} \inf_{P \in A_\infty} \int W(\j) \, dP(x, \j)   -  \frac{\log \K}{n}   \\+\frac{1}{n} \log \(      \int_{\mc^n} e^{-   \beta n \int\x\,d\nu_n  } \, dx_1 \dots  dx_n\) +o(1) \end{multline}
thus in view  of Lemma \ref{lemintxi} and \eqref{lbk},
we have  established \eqref{ldr}.
An immediate corollary of \eqref{logq3},  choosing $A_n$ to be the full space and using $\inf \frac{|\E|}{\pi}\int  W(\j)\, dP(\j)=\alpha$  and Lemma \ref{lemintxi}, is that\begin{equation}\label{logk3}
\limsup_{n\to \infty}\frac{\log \K}{n}\le - \frac{\beta\alpha}{2} +\log |\E|.\end{equation}
\smallskip

We next turn to the lower bound.
Fix $\eta>0$. Given $A$, let $P\in \mathring{A}$ be such that \begin{equation}\label{bsup2}
\int W(\j)\, dP(x,\j)\le \inf_{P\in \mathring{A}}\int W(\j)\, dP(\j)+\frac{\eta}{2}.\end{equation} Since $P\in \mathring{A}$, if $\eta$ is chosen small enough (which we assume) then  $B(P,2\eta)\subset A$, where the ball is for a distance metrizing weak convergence  as in Proposition \ref{construct}.

 We then apply Proposition \ref{construct} to $P$ and $\eta$. We find $\delta>0$ and for any $n$ large enough a set $A_n$ such that $|A_n|\ge n!(\pi\delta^2/n)^n$ and, rewriting  \eqref{bsw} with \eqref{idwn},
 \begin{equation}\label{bsup1}
 \limsup_{n\to \infty} \sup_{A_n}\F\le \frac{|\E|}{\pi}\int W(\j)\, dP(\j)+\eta.\end{equation}
Moreover, for every $(y_1,\dots,y_n)\in A_n$ and letting $\{\nu_n=  \delta_{y_1}+\dots+\delta_{y_n}\}_n$,  there exists $\{\j_n\}_n$ in $L^p_\loc(\mr^2,\mr^2)$ such that $\div \j_n = 2\pi( \nu_n' -\bm')$ and such that  the image $P_n$ of $dx_{|\E}/|\E|$ by the map $x\mapsto \(x,\j_n(\sqrt n x+\cdot)\)$ is such that
\begin{equation}\label{bsup22}\limsup_{n \to \infty} \dist(P_n,P) \le \eta.\end{equation}
In particular \eqref{fauxi} holds.
Moreover, inserting \eqref{bsup1}  and \eqref{bsup2} into \eqref{loi2}, we find that
$$\frac{\log \Q(A_n)}{n}\ge -\frac{\log \K}{n} - \frac{\beta |\E|}{2\pi }
 \inf_{P\in \overset{\circ}{A}}\int W(\j)\, dP(\j) -\hal\beta\eta +\frac{1}{n}\log \left|\frac{A_n}{\sqrt{n}}\right|+o(1) .$$
 On the other hand, using  $|A_n |\ge n! (\pi \delta^2/n)^n $ and Stirling's formula, we have $\log |A_n|\ge   2n \log \delta - Cn   $.
Combining with \eqref{logk3}, \eqref{ldlb} follows, with $C_\eta = -2\log\delta + C + \log|\E|$.

Theorem \ref{valz}  immediately follows by combining \eqref{logk3}, \eqref{lbk} and \eqref{defK}.

\ed{\subsection{Proof of Theorem \ref{additi}}
We apply the method of Theorem \ref{gamma} part B.
Let $A_{n, M}= \{ (x_1, \cdots, x_n): \widehat{F_n}(\sum_{i=1}^n \delta_{x_i}) \le M\}$.
In view of \eqref{logq2}, Corollary \ref{coro43}, and  Lemma \ref{lemintxi},  if $M$ is chosen large enough we have $\Q(A_{n,M}^c) \to 0$ as $n\to \infty$. In view of Lemma \ref{lesigne}, to prove the tightness of $i_n\# \Q$ it thus suffices to  check that if $P_n \in i_n (A_{n,M})$ then $P_n$ has a convergent subsequence. But     we have just proven this in Theorem \ref{th2}, part A, item 1.

}

\ed{\subsection{Definition of $\mathbb{W}$}\label{secnewdefw}
In this subsection we briefly examine how to define the renormalized energy as a function of the points only, via \eqref{neww}.
We prove the following:
\begin{lem} The function $\mathbb{W}$ be defined by \eqref{neww} is Borel-measurable on the set of locally finite measures.
\end{lem}
\begin{proof}
First we show that there exists a measurable map $\nu \mapsto \j_\nu$ where $\j_\nu$ satisfies \eqref{eqj}.
The set $$A=\{\j \in \mathcal{A}_m, W(\j)<\infty\}$$
is Borel measurable, since $W$ is (as proven in \cite[Theorem 1]{ss1}).
We may partition $A$ into equivalence classes for the relation $\j\sim \j'$ if $\div \j =\div \j'$.  In view of Lemma \ref{lem15}, denoting by $\j^*$ the equivalence class of $\j\in A$, we have $\j^*=\{\j+\vec{C}, \vec{C}\in \mr^2\}$.   In particular this implies that if $U$ is an open set in $A$, then $U^*=\cup_{\j\in U} \j^*$ is open too in $A/\sim$.
By Effros's theorem (cf. e.g. \cite[Theorem 5.4.3]{sri}) there thus exists a Borel section $B$ of $A$ which contains exactly one element of each equivalence class.
 The map $\j^* \mapsto \frac{1}{2\pi}\div \j + m$ is then a Borel measurable and injective map from $B$ to $\{\nu \in \mathcal{M}_+: \mathbb{W}(\nu)<\infty\}$ where $\mathcal{M}_+$ is the set of positive Radon measures on $\mr^2$.
By  \cite[Prop. 8.3.5]{cohn} its inverse is also Borel measurable and injective. This provides  a measurable selection $\psi:\nu \mapsto\j$ satisfying \eqref{eqj} on $\{\nu\in \mathcal{M}_+: \mathbb{W}(\nu)<\infty\}$.
Since
 $\j^*=\{\j+\vec{C}, \vec{C}\in \mr^2\}$.
 we may write
$$\mathbb{W}(\nu)=\inf_{\vec{C}\in \mr^2} W(\psi(\nu)+ \vec{C}) .$$
Using again the fact that  $W$ is Borel measurable and $\nu \mapsto \psi(\nu)+ \vec{C}$ too, it follows that $\mathbb{W}$ is measurable as claimed.
 \end{proof}

We may then without too much difficulty translate the results of Theorems \ref{th2}, \ref{th4} with $\int \mathbb{W}(\nu)\, dQ(x,\nu)$ instead of $\int W(\j)\, dP(x,\j)$.

}

\section{Proof of Proposition~\ref{construct}}\label{sec5}
The construction consists of the following. We are given $\ep>0$, which is the error we can afford. First  we
select a finite set of vector fields  $J_1, \dots, J_{N}$ ($N$ will depend on $\ep$) which will represent the probability $P(x,\j)$ with respect to  its $\j$ dependence, and whose renormalized energies  are well-controlled.  Since $P$ is $T_{\lambda(x)}$-invariant, we need it to be well-approximated by measures supported on  the orbits  of the  $J_i$'s under translations.
Secondly, we work in blown-up coordinates and  split the region $\E'$ (whose diameter is order $\sqrt{n}$) into many rectangles $K$ with centers $x_K$ and sidelengths of order $\bar{R}$ large enough.   Even though we choose $\bar{R}$ to be large, it will still be very small compared to the size of $\E'$, as $n\to \infty$, so that the Diracs at $x_K/\sqrt{n}$ approximate $P(x,\j)$ with respect to its $x$ dependence.
On each rectangle $K$, the weight ${m_0}'$ is temporarily replaced by its average $m_K$.
Then we split each rectangle $K$ into $q^2$ identical rectangles, with sidelengths of order $2R=\bar R / q$, where both $R$ and $q$ will be sufficiently large. We then select the proportion of the rectangles that corresponds to the weight that the orbit of each $J_i $ carries in the approximation of $P$. In these rectangles we paste a (translated) copy of $J_i$ at the scale $m_K$ and suitably modified near the boundary according to a construction of \cite{ss1} (Proposition \ref{tronque} below) so that its tangential component on the boundary is $0$ (this can be done while inducing only an error $\ep$ on $W$).
 In the few rectangles that may remain unfilled, we paste a copy of an arbitrary $J_0$ whose renormalized energy is finite. We perform the construction above provided we are far enough from $\p \E'$. The layer near the boundary must be treated separately, and there again an arbitrary (translated and rescaled) current can be pasted. Finally, we add a vector field to correct the discrepancy between $m_K$ and ${m_0}'$ in each of the rectangles.

To conclude the proof of Proposition~\ref{construct}, we collect all of the estimates on the constructed vector field to show that its energy $\w$ is bounded above in terms of $\int W\, dP$ and that the probability measures associated to the construction have remained close to $P$.

\subsection{Estimates on distances between probabilities}

First we choose distances which metrize the topologies of $\Lp$ and $\B (X)$, the set of finite Borel measures on $X=\E\times\Lp$. For $\j_1,\j_2\in\Lp$  we let
$$d_p(\j_1,\j_2) = \sum_{k=1}^\infty 2^{-k} \frac{\|\j_1 - \j_2\|_{L^p(B(0,k)) }}{1+\|\j_1 - \j_2\|_{L^p(B(0,k))}}, $$
and on $X$ we use the \ed{sum} of the Euclidean distance on $\E$ and $d_p$, which we denote $d_X$. \ed{It is a distance on $X$.} On $\B(X)$ we define a distance by choosing a sequence of bounded continuous functions $\{\vp_k\}_k$ which is dense in $C_b(X)$ and we let,  for any $\mu_1,\mu_2\in\B(X)$,
$$ d_{\B}(\mu_1,\mu_2) = \sum_{k=1}^\infty 2^{-k} \frac{|\lb \vp_k,\mu_1-\mu_2\rb|}{1+|\lb \vp_k,\mu_1-\mu_2\rb|},$$
where we have used the notation $\lb\vp,\mu\rb = \int\vp\,d\mu$.

We have the following general  facts.

\begin{lem}\label{etavar} For any $\ep>0$ there exists $\eta_0>0$ such that if $P,Q\in\B(X)$ and $\|P-Q\|<\eta_0$, then $d_{\B}(P,Q)<\ep$. Here $\|P-Q\|$ denotes the total variation of the signed measure $P-Q$, i.e. the supremum of $\lb\vp,P-Q\rb$ over measurable functions $\vp$ such that $|\vp|\le 1$.
\end{lem}
In particular, if $P = \sum_{i=1}^\infty \alpha_i\delta_{x_i}$ and  $Q = \sum_{i=1}^\infty \beta_i\delta_{x_i}$ with $\sum_i|\alpha_i - \beta_i| <\eta_0$, then  $d_{\B}(P,Q)<\ep$.

\begin{lem}\label{etadir} Let $K\subset X$ be compact. For any $\ep>0$ there exists $\eta_1>0$ such that if $x\in K, y\in X$ and $d_X(x,y)<\eta_1$ then $d_{\B}(\delta_x, \delta_y)<\ep$.
\end{lem}

\begin{lem}\label{etaconv} Let $0<\ep<1$.  If  $\mu$ is a probability measure on a set $A$ and $f,g:A\to X$ are measurable and such that $d_{\B}(\delta_{f(x)}, \delta_{g(x)})<\ep$ for every $x\in A$, then
$$d_{\B}(f^\#\mu,g^\#\mu) <C \ep(\lep+1).$$
\end{lem}
\begin{proof} Take any bounded continuous function $\vp_k$ defining the distance on  $\B(X)$. Then if $d_{\B}(\delta_{f(x)}, \delta_{g(x)})<\ep$ for any $x\in X$ we have in particular
$$\frac{|\vp_k(f(x)) - \vp_k(g(x))|}{1+ |\vp_k(f(x)) - \vp_k(g(x))|}\le 2^k \ep.$$
It follows that
$$ d_{\B}(f^\#\mu,g^\#\mu) \le \sum_k 2^{-k} \min(\ep 2^k, 1)\le \ep \({[\log_2\ep]+1}\) + \sum_{k={[\log_2\ep]+1}}^\infty 2^{-k} \le C\ep(\lep+1). $$
\end{proof}

\subsection{Preliminary results}
In what follows $\E' = \sqrt n \E$, $\bm'(x) = \bm(x/\sqrt n)$: we work in blown-up coordinates.  We consider a probability measure $P$ on $\E\times\Lp$ which  is as in the proposition.  We let $\tP$ be the probability measure on $\E\times \mathcal A_1$ which is the image of $P$ under $(x,\j)\mapsto (x,\sigma_{1/\bm(x)} \j)$, so that
\begin{equation}\label{proba} \tP = \int \delta_x\otimes\delta_{\sigma_{1/ \bm(x)} \j}\,dP(x,\j),\quad P = \int \delta_x\otimes\delta_{\sigma_{ \bm(x)} \j}\,d\tP(x,\j)\end{equation}
It is easy to check that since $P$ is $T_{\lambda(x)}$-invariant, $\tP$ is as well, and in particular it is translation-invariant.

The construction is based on  the following statement which is a rewriting of Proposition~4.2 in \cite{ss1} and the remark following it:
%
\begin{pro}[Screening of an arbitrary vector field]\label{tronque} Let $K_R = [-R,R]^2$, let $\{\chi_R\}_R$ satisfy \eqref{defchi}.

Let  $G\subset\mathcal A_1$ be  such that there exists $C>0$ such that for any $ \j\in G$ we have
\begin{equation} \label{hypunif}  \forall R>1,\  \frac{\nu(K_R)}{|K_R|} < C,\end{equation}
for the associated $\nu$'s, and
\begin{equation}\label{cvwunif}
 \lim_{R\to +\infty} \frac{W(\j,\chi_R)}{|K_R|}  = W(\j),\end{equation}
where the convergence is uniform  w.r.t. $\j\in G$.

Then for any $\ep>0$ there exists $R_0>0$, $\eta_2>0$  such that  if $R>R_0$ and $L$  is a rectangle centered at $0$ whose sidelengths belong to $[2R,2R(1+\eta_2)]$ and such that  $|L| \in\mn$,  then  for every $\j\in G$ there exists  a $\j_L\in\Lp$  such that   the following hold
\begin{itemize}
\item[i)]  $\j_L = 0$ in $L^c$,
\item[ii)] There is a discrete subset $\Lambda \subset L$ such that $$\div \j_L = 2\pi \(\sum_{p\in\Lambda}\delta_p - \indic_L\).$$ In particular $\j
\cdot \vec{\nu} =0$ on $\partial L$,  there is no singular part of the  divergence on $\p L$ and thus $\# \Lambda =|L|$.
\item[iii)] If $d(x, L^c) >R^{\tq}$ then $\j_L(x) = \j(x)$
\item[iv)]
\begin{equation}\label{restronque} \frac{W(\j_L,\indic_L)}{|L|} \le W(\j)+\ep.\end{equation}\end{itemize}
\end{pro}
\ed{We note that if $\j$ is such that $\div \j= 2\pi \sum_p\delta_{p}-1$ and we have $\curl \j=0$ in a neighborhood of each $p\in \Lambda$, then the definition \eqref{WR} still makes sense, in particular the limit exists. This is what is meant by $W$ in \eqref{restronque}, as well as in the rest of the section.
}


The next lemma explains how to partition $\E$ into rectangles. \ed{The main point is to cut $\E'$ into stripes and then each stripe into rectangles in such a way that $\int m_0'$ over each rectangle is  a large integer.
}

\begin{lem}\label{rect} There exists a constant $C_0>0$ such that, given any $R>1$ and $q\in\mn^*$,  there exists for any $n\in\mn^*$ a collection $\cK_n$ of closed rectangles in $\E'$ with disjoint interiors,  whose sidelengths are between  $\Rb = 2qR$ and $\Rb+C_0\Rb/R^2$,  and which are such that
$$ \{x\in \E': d(x,\p \E')\le \Rb \}\subset  \E'\sm\bigcup_{K\in\cK_n} K \subset \{x\in \E': d(x,\p \E')\le C_0 \Rb \}$$
and, for all $K\in\cK_n$,
\begin{equation}\label{entier}\int_K \bm' \in q^2\mn.\end{equation}
\end{lem}

\begin{proof}
For each $j\in\mz$ we let
$$m_j(t) = \int_{x=-\infty}^t\int_{y=j\Rb}^{(j+1)\Rb} \bm'(x,y) \,dy\,dx.$$
Then each strip $\{j\Rb\le y < (j+1) \Rb\}$ is cut  into rectangles $[t_{ij},t_{(i+1),j}]\times[ j\Rb,(j+1) \Rb]$ where $t_{0j} = -\infty$ and
$$t_{i+1,j} = \min\{t\ge t_{ij}+\Rb : m_j(t_{ij}) \in q^2 \mn\}.$$
Since  by assumption \eqref{assumpV2} we have $\bm'(x)\in [\bw,\wb]$ for any $x\in \E'$,  it is not difficult to check that  if such a rectangle is included in $\E'$ then
$$ t_{ij}+\Rb\le t_{i+1,j} \le t_{ij} +\Rb + \frac{q^2}{\bw \Rb},$$
 and thus  its sidelengths are between  $\Rb$ and $\Rb+C\Rb/R^2$ since $\Rb/R^2 = 4 q^2/\Rb$.
We let $\cK_n$ be the set of rectangles of the form $[t_{ij},t_{(i+1),j}]\times[ j\Rb,(j+1) \Rb]$ which are included in  $\{x: d(x,\p \E')>\Rb\}$.  From the above,
it follows that these rectangles in fact cover the set $\{x: d(x,\p \E')>C\Rb\}$ for some $C>0$ independent of $R>1$, $q$. By construction each $K\in\cK_n$ is such that
$$\int_K \bm' = m_j(t_{(i+1),j}) - m_j(t_{ij})\in q^2\mn.$$
\end{proof}

The next lemma explains how to select a good subset of  $\Lp$.
\begin{lem}\label{GH}
Let $\tilde{P}$ be a translation invariant measure on $X$ such that $\tilde{P}$-a.e. $\j\in\mathcal A_1$ and $W(\j)<\infty$. Then for any  $\ep>0$, for any $R_\ep>0$, there exist subsets $H_\ep \subset G_\ep $ in $\Lp$ which are compact and such that
\begin{itemize}
\item[i)]  $\eta_0$ being given by Lemma~\ref{etavar} we have
 \begin{equation}\label{pgep} \tP(\E\times {G_\ep}^c) <\min({\eta_0}^2,\eta_0 \ep),\quad  \tP(\E\times  H_\ep^c) < \min(\eta_0,\ep).\end{equation}
\item[ii)] For every $\j\in H_\ep$, there exists $\Lj\subset K_{\wb R_\ep}$ such that
\begin{equation}\label{gamj}\text{$|\Lj|< C {R_\ep}^2 \eta_0$ and
 $ \lambda\notin \Lj\implies \theta_\lambda \j\in G_\ep.$}\end{equation}
\item[iii)] The convergence in the definition of $W(\j)$  is uniform w.r.t. $(x,\j)\in G_\ep$ and, writing $\div \j = 2\pi(\nu-1)$,
\begin{equation}\label{wunif}\text{$W(\j)$ and $\frac{\nu(K_R)}{R^2}$ are bounded uniformly w.r.t.  $(x,\j)\in G_\ep$ and  $R>1$.}\end{equation}
\item[iv)] We have
 \begin{multline}\label{ppseconde}  d_{\B}(P,P'') < C \ep(\lep+1),\quad\text{where}  \\
P'' = \int_{\E\times H_\ep}\frac1{\bm(x)|K_{R_\ep}|}  \int_{\sqrt{\bm(x)}K_{R_\ep}\sm\Lj} \delta_x\otimes\delta_{\sigma_{\bm(x)}\theta_\mu \j}\,d\mu\,d\tP(x,\j).\end{multline}
\end{itemize}
Moreover, there exists a partition of $H_\ep$ into $\cup_{i=1}^{N_\ep} H_\ep^i$  satisfying $\diam (H_\ep^i) <\eta_3$, where $\eta_3$ is such that
\begin{equation}\label{eta5}
\j \in H_\ep, \ d_p(\j,\j') <\eta_3 , \
m \in [\bw, \wb], \ \mu \in  \sqrt{\wb} K_{R_\ep}\backslash \Gamma(\j)
    \implies
d_{\B} (\delta_{\sigma_m  \theta_\mu \j}, \delta_{\sigma_m \theta_\mu \j'}) <\ep;\end{equation}
 and there exists  for all $i$,  $J_i \in H_\ep^i$ such that
 \begin{equation}\label{wjk} W(J_i) < \inf_{H_\ep^i} W +\ep.\end{equation}

\end{lem}
\ed{At this point, denoting $\tilde{Q}$ the projection of $\tilde{P
}$ under $\j\mapsto \frac{1}{2\pi} \div \j + 1$, we may always choose $J_i$ such that $W(J_i) <\inf_{H^i_\ep} \mathbb{W}+\ep$. }

\begin{proof}\mbox{}\medskip

\noindent
{\it Step 1:  Choice of $G_\ep$.}   Since $\Lp$ is Polish we can always find a compact set  $G_\ep$ satisfying \eqref{pgep} and  $P({G_\ep}^c)<\eta_0$. Then from Lemma~\ref{etavar},  $P\llcorner G_\ep$ (the restriction  of $P$ to $G_\ep$)  satisfies $d_{\B}(P,P\llcorner G_\ep)<\ep.$

From the translation invariance of $\tP$ and  for any $\lambda$, we have $\tP(\E\times \theta_\lambda G_\ep) >1-\eta_0$ and therefore  $d_{\B}(\tP, \tP\llcorner \theta_\lambda G_\ep)< \ep$. In view of \eqref{proba}, it follows that for any $\lambda\in\mr^2$ we have $\|P-P_\lambda\|<\eta_0$ and then $d_{\B}(P,P_\lambda)<\ep$, where
$$P_\lambda = \int_{\E\times \theta_\lambda G_\ep} \delta_x\otimes\delta_{\sigma_{\bm(x)} \j}\,d\tP(x,j) = \int_{\E\times G_\ep} \delta_x\otimes\delta_{\theta_\lambda\sigma_{\bm(x)} \j}\,d\tP(x,\j).$$
Then using Lemma~\ref{etaconv} we deduce that
if $A\subset \mr^2$ is any measurable set of positive measure, then
\begin{equation}\label{pprime} d_{\B}(P,P') < C \ep(\lep+1),\quad\text{where}  \quad P' = \int_{\E\times G_\ep}\dashint_A \delta_x\otimes\delta_{\theta_\lambda\sigma_{\bm(x)} \j}\,d\lambda\,d\tP(x,\j).\end{equation}

Moreover, since $P$ is $T_{\lambda(x)}$-invariant, choosing $\chi$ to be a smooth positive function with integral $1$ supported in $B(0,1)$, the ergodic theorem  (as in Section \ref{ergo} or see again \cite{becker}) ensures that  for $P$-almost every $(x,\j)$ the limit
$$\lim_{R\to +\infty} \frac1{|K_R|} \int_{K_R} W(\j(\lambda+\cdot),\chi(\lambda+\cdot))\,d\lambda$$
exists. Then   $\indic_{K_R} * \chi$ is  a family of functions  which satisfies \eqref{defchi} with respect to the family of squares $\{K_R\}_R$, and from the definition of the renormalized energy relative to $\{K_R\}_R$  we may rewrite the limit above as
\begin{equation}\label{cvw} W(\j) = \lim_{R\to +\infty} \frac1{|K_R|} W(\j,\indic_{K_R} * \chi). \end{equation}
By Egoroff's theorem we may choose the compact set $G_\ep$ above to be such that, in addition to \eqref{pprime}, the convergence in \eqref{cvw}  is uniform on $G_\ep$. In fact, since  $W(\j)<+\infty$ and $\limsup_R\nu(K_R)/R^2<+\infty$ for $P$-a.e. $(x,\j)$,  where $\div  \j = 2\pi(\nu-1)$, we may choose  $G_\ep$ such that \eqref{wunif} holds.

The arguments above show that the properties \eqref{pprime}, \eqref{wunif} can be satisfied for a compact set $G_\ep$ of measure arbitrarily close to $1$. We choose $G_\ep$ such that  \eqref{pgep} holds.

 The next difficulty we have to face is that $\theta_\lambda \j$ need not belong to $G_\ep$ if $\j$ does. \medskip

\noindent{\it Step 2:  Choice of $H_\ep$.}  For $\j\in G_\ep$, let $\Lj$ be the set of $\lambda$'s in $\sqrt\wb K_{R_\ep}$ such that $\theta_\lambda \j\not\in G_\ep$. Since, from \eqref{pgep} and the translation-invariance of $\tP$, for any $\lambda\in\mr^2$ we have $\tP(\E\times \theta_\lambda (G_\ep)^c)<{\eta_0}^2$, it follows from Fubini's theorem that
$$\int_{G_\ep} |\Lj|\,d\tP(x,\j) = \int_{\sqrt\wb K_{R_\ep}} \tP(\E\times(\theta_\lambda G_\ep)^c)\,d\lambda < 4 \wb{R_\ep}^2\min({\eta_0}^2,\eta_0 \ep).$$
Therefore, letting
\begin{equation}\label{defhep} H_\ep = \{\j\in G_\ep: |\Lj| <  4\wb {R_\ep}^2  \eta_0\}, \end{equation}  we have that
\eqref{pgep} holds.

Combining \eqref{pgep} and \eqref{defhep} with Lemma~\ref{etavar}, we deduce from \eqref{pprime} that  \eqref{ppseconde} holds,
where we have used the fact that $\theta_\lambda \sigma_m \j = \sigma_m\theta_{\sqrt m\lambda} \j$ to change the integration variable to $\mu = \sqrt{\bm(x)}\lambda$ in \eqref{pprime}. \medskip

\noindent{\it Step 3:  Choice of $J_1,\dots,J_{N_\ep}$.}  We use the fact that $G_\ep$ is relatively compact in $\Lp$, Lemma \ref{etadir},  and the fact that
$(m,\j)\mapsto \sigma_{m}  \j$ is continuous to find that  there exists $\eta_4>0$ such that for any $m\in [\bw,\wb]$ and any $\j \in G_\ep$ it holds that
\begin{equation}\label{eta4} d_p(\j,\j')<\eta_4\quad \implies\quad d_{\B}\(\delta_{\sigma_m \j},\delta_{\sigma_m \j'}\)<\ep. \end{equation}   Moreover, from the continuity of $(\mu,\j)\mapsto \theta_{\mu}  \j$, there exists $\eta_3>0$ such that
\begin{equation}\label{eta52} \j\in G_\ep, d_p(\j,\j')<\eta_3,\ \mu\in \sqrt\wb K_{R_\ep} \quad \implies\quad d_p\(\theta_\mu \j,\theta_\mu \j'\)<\eta_4. \end{equation}
If $\j \in H_\ep$ and  $\mu \in K \backslash \Lj $ then $\theta_\mu \j \in G_\ep $ hence applying  \eqref{eta4},  we get \eqref{eta5}.

Now we cover the relatively compact set $H_\ep$ by a finite number of balls $B_1,\dots, B_{N_\ep}$ of radius $\eta_3/2$, derive from it a partition of $H_\ep$ by sets with diameter less than $\eta_3$, by letting $H^1_\ep = B_1\cap H_\ep$ and
$$H^{i+1}_\ep = B_{i+1}\cap H_\ep\sm \(B_1\cup\dots\cup B_i\).$$
we then have
\begin{equation}\label{partition} H_\ep = \bigcup_{i=1}^{N_\ep} H^i_\ep,\quad \diam(H^i_\ep)<\eta_3,\end{equation}
where the union is disjoint.
Then we may  choose $J_i\in H_\ep^i$ such that  \eqref{wjk} holds.
\end{proof}
\subsection{Completing the construction}
{\it Step 1: Choice of $R_\ep$.}
We apply Proposition~\ref{tronque} with $G = G_\ep$ and   $\sqrt m R$, where $m\in[\bw,\wb]$. The proposition yields $\eta_2>0$,  $R_\ep>1$ such that for any  $\j\in G_\ep$ and  any $m\in[\bw,\wb]$ and  any rectangle $L$ centered at $0$ with sidelengths in $[2\sqrt m R_\ep, 2\sqrt m R_\ep(1+\eta_2)]$, \eqref{restronque} is satisfied for some $\j_L$, with $R$ replaced by $\sqrt m R_\ep$.  The reason for including  $\sqrt m$ is that we will need to scale the construction to account for the varying weight $\bm(x)$.

Since our  rectangles will be obtained from Lemma~\ref{rect} and we wish to use the approximation by $\j_L$ in them,  we choose $R_\ep$ large enough so that
if  $m\in[\bw,\wb]$ and $L$ is a rectangle centered at zero with sidelengths in $[2\sqrt m R_\ep, 2\sqrt m R_\ep(1+C_0/{R_\ep}^2)]$ then
\begin{equation}\label{aspratio} \frac{C_0}{R_\ep^2} < \eta_2,\quad \frac {C_1}{{R_\ep}^2} <\eta_0,   \quad K_{\sqrt m R_\ep(1-\eta_0)}\subset \{x: d(x,L^c)> {\sqrt m R_\ep}^{\tq}\}\subset K_{\sqrt m R_\ep(1+\eta_0)}, \end{equation}
where $C_0$ is the constant in Lemma~\ref{rect},  $C_1\ge 1$ is to be determined later,  and $\eta_0$ is the constant in Lemma~\ref{etavar}.

 If $\lambda\in K_{\sqrt m R_\ep(1-\eta_0)}$ and  since $\j=\j_L$ if $d(x,L^c)> {\sqrt m R_\ep}^{\tq}$, we deduce from \eqref{aspratio} that  $\theta_\lambda \j_L = \theta_\lambda \j$ in $B(0, {\sqrt m R_\ep}^{\tq})$,  so that  from the definition of $d_p$,  taking  $R_\ep$ larger if necessary,
\begin{equation}\label{proxitronque}  \forall \j\in G_\ep,m\in [\bw,\wb],\lambda\in K_{\sqrt m R_\ep(1-\eta_0)},\quad d_p(\theta_\lambda\sigma_m \j, \theta_\lambda \sigma_m \j_L) < \frac{\eta_1}{2},\end{equation}
where $\eta_1$ comes from Lemma~\ref{etadir} applied on $\{\sigma_m \j: m\in [\bw,\wb], \j\in G_\ep\}$, i.e.  is such that
\begin{equation}\label{eta1}
\text{$m\in  [\bw,\wb]$,  $\j\in G_\ep,\j'\in\Lp$ and $d_p(\j,\j')<\eta_1$} \implies d_{\B}(\delta_{\sigma_m \j},\delta_{\sigma_m \j'})<\ep.\end{equation}\medskip

\noindent{\it Step 2:  Choice of $q_\ep$ and the  rectangles.} We choose an integer $q_\ep$ large enough so that
\begin{equation}\label{qep} \frac {N_\ep}{C_1 {q_\ep}^2} < \eta_0,\qquad \frac{N_\ep}{{q_\ep}^2}\times \max_{\substack{0\le i\le N_\ep\\ \bw\le m\le \wb}} W_K(\sigma_m J_i)< \ep\end{equation}
where $C_1>1$ is to be determined later.
We apply Lemma~\ref{rect} with $R_\ep$, $q_\ep$ and $N_\ep$ to obtain for any $n$ a collection $\cK_n$ of rectangles (we omit to mention the $\ep$ dependence) which cover most of $\E'$, and we also apply Lemma \ref{GH}.
 We rewrite $P''$ given by \eqref{ppseconde} as
\begin{equation}\label{psecondebis}P'' = \sum_{K\in\cK_n} \int_{\frac{K}{\sqrt n}\times H_\ep} \frac1{\bm(x)|K_{R_\ep}|}\int_{\sqrt{\bm(x)}K_{R_\ep}\sm\Lj} \delta_x\otimes\delta_{\sigma_{\bm(x)}\theta_\mu \j}\,d\mu\,d\tP(x,\j).\end{equation}
Now we  claim that if $n$ is large enough and $x\in K/\sqrt n$, $\j\in H_\ep^i$, $\mu\in \sqrt{\bm(x)}K_{R_\ep}\sm\Lj$, then
\begin{equation}\label{crucial} d_{\B}\(\delta_x\otimes\delta_{\sigma_{\bm(x)}\theta_\mu \j}, \delta_{x_K} \otimes\delta_{\sigma_{m_K}\theta_\mu J_i}\)<2\ep,\end{equation}
where $x_K$ is the center of $K/\sqrt n$ and $m_K$ is the average of $\bm$ over $K/\sqrt m$. Indeed, since $\bm$ is $C^1$ we have $|x-x_K|<C/\sqrt n$, $|\bm(x) - m_K|<C/\sqrt n$ thus if $n$ is large enough, since $\theta_\mu \j\in G_\ep$ we find
$$ d_{\B}\(\delta_x\otimes\delta_{\sigma_{\bm(x)}\theta_\mu \j}, \delta_{x_K}\otimes\delta_{\sigma_{m_K}\theta_\mu \j}\)<\ep.$$
 Moreover, since $d_p(\j,J_i)<\eta_3$, we deduce  from \eqref{eta5}  that
$$ d_{\B}\(\delta_{x_K}\otimes\delta_{\sigma_{m_K}\theta_\mu \j}, \delta_{x_K}\otimes\delta_{\sigma_{m_K}\theta_\mu J_i}\)<\ep,$$
which together with the previous estimate proves \eqref{crucial}.

Using \eqref{crucial} together with  Lemmas~\ref{etadir}, \ref{etavar},  and \eqref{gamj}, we deduce from \eqref{ppseconde} and  \eqref{psecondebis} that $d_{\B}(P,P''')< C\ep(\lep +1)$, where
\begin{equation}\label{p3}\begin{split} P''' &= \sum_{\substack{K\in\cK_n\\1\le i\le N_\ep}} \int_{\frac{K}{\sqrt n}\times H_\ep^i} \dashint_{\sqrt{m_K}K_{R_\ep}} \delta_{x_K}\otimes\delta_{\sigma_{m_K}\theta_\mu J_i}\,d\mu\,d\tP(x,\j)\\ & =
\sum_{\substack{K\in\cK_n\\1\le i\le N_\ep}} p_{i,K}\dashint_{\sqrt{m_K}K_{R_\ep}} \delta_{x_K}\otimes\delta_{\sigma_{m_K}\theta_\mu J_i}\,d\mu,\end{split}\end{equation}
where
\begin{equation}\label{pkk} p_{i,K} = \tP\(\frac{K}{\sqrt n}\times H_\ep^i\).\end{equation}\medskip

\noindent{\it Step 3: Choice of subrectangles and vector field $\j_n$.}
We now replace $p_{i,K}$ in the definition \eqref{p3}  by
\begin{equation}\label{nkk} \frac{|K|}{{q_\ep}^2|\E'|} n_{i,K},\quad\text{where}\quad n_{i,K} = \left[\frac{{q_\ep}^2|\E'|}{|K|}p_{i,K}\right].\end{equation}
We have, since  $\tP(\frac K{\sqrt n}\times\Lp) = |K|/|\E'|$,
\begin{equation}\label{pnkk}\sum_{k=1}^{N_\ep} n_{i,K}\le \frac{{q_\ep}^2|\E'|}{|K|}\tP\(\frac{K}{\sqrt n}\times\Lp\) = {q_\ep}^2\end{equation}
and
\begin{equation*}
\left|\frac{|K_{R_\ep|}}{|\E'|} n_{i,K} - p_{i,K}\right| < C
 \(\frac{|K|}{{q_\ep}^2|E'|}  + \frac{n_{i,K}   }{R_\ep^2 |\E'|}\).
  \end{equation*}  Summing with respect to $i $ and $K$, using the facts that
  $\sum_{K\in \cK_n}       |K| < |\E'|$, \eqref{pnkk}, and the fact that the cardinal of $\cK_n$ is $\frac{|\E'|}{4q_\ep^2 R_\ep^2} $,  we find
  \begin{equation*}
\sum_{1\le i \le N_\ep, K\in \cK_n}
\left|\frac{|K_{R_\ep|}}{|\E'|} n_{i,K} - p_{i,K}\right|<  C\( \frac{N_\ep}{{q_\ep}^2} + \frac{1}{R_\ep^4}\).\end{equation*}
    We may always choose $C_1$ large enough in \eqref{aspratio} and \eqref{qep} so that the right-hand side is $<\eta_0$.
     Then Lemma~\ref{etavar} implies that $d_{\B}(P,P^{(4)})< C\ep(\lep +1)$ is still true after replacing $p_{i,K}$ by $\frac{|K_{R_\ep}|}{|\E'|} n_{i,K}$ in \eqref{p3}, i.e. where
\begin{equation}\label{p3bis} P^{(4)} =
\frac 1{|\E'|}\sum_{\substack{K\in\cK_n\\1\le i\le N_\ep}} \frac{n_{i,K}}{m_K}\int_{\sqrt{m_K}K_{R_\ep}} \delta_{x_K}\otimes\delta_{\sigma_{m_K}\theta_\mu J_i}\,d\mu.\end{equation}

Next, we divide each $K\in\cK_n$ into  a collection $\LK$ of ${q_\ep}^2$ identical subrectangles in the obvious way and  we partition $\LK$ into collections $\LKk$, $0\le i\le N_\ep$ such that if $k\ge 1$ then $\LKk$ contains $n_{i, K}$ subrectangles. This is clearly possible from \eqref{pnkk}. If the inequality is strict we put the extra subrectangles in $\LKo$, there will be $n_{0,K}$ of them and then
\begin{equation}\label{sumnkk}\sum_{k=0}^{N_\ep} n_{k,K} = {q_\ep}^2.\end{equation}
We rewrite \eqref{p3bis} as
\begin{equation}\label{p3ter} P^{(4)}=
\frac 1{|\E'|}\sum_{\substack{K\in\cK_n\\1\le i\le N_\ep\\ \tilde L\in\LKk}} \frac{1}{m_K}\int_{\sqrt{m_K}K_{R_\ep}} \delta_{x_K}\otimes\delta_{\sigma_{m_K}\theta_\mu J_i}\,d\mu.\end{equation}

Now, for $\tilde L\in \LKk$, let $L = \sqrt{m_K}(\tilde L - x_{\tilde L})$, where $x_{\tilde L}$ denotes the center of $\tilde L$. From Lemma~\ref{rect}, a rectangle $K\in\cK_n$ has sidelengths between $2q_\ep R_\ep$ and $2q_\ep R_\ep(1+C_0/{R_\ep}^2)$. Therefore $L$ is a rectangle centered at zero with sidelengths between $2\sqrt{m_K} R_\ep$ and $2\sqrt{m_K} R_\ep(1+C_0/{R_\ep}^2)$,  and  \eqref{proxitronque} holds.

This, and the results of Lemma \ref{GH}, allow us to apply Proposition~\ref{tronque} on $L$ to any $J_i$, $1\le i\le N_\ep$. Note that $|L|\in\mn$ follows from the fact that
$$|L| = m_K |\tilde L| = \dashint_K \bm' \frac{|K|}{{q_\ep}^2} = \frac{1}{{q_\ep}^2} \int_K \bm'$$
and \eqref{entier}. In this way, we define currents $J_{i,L}$ which satisfy \eqref{restronque} and \eqref{proxitronque}. We claim that,  as a  consequence of the latter, we have
\begin{equation}\label{claimtronque}\text{$\j' = J_{i,L}$ on $L$} \implies d_{\B}\(\dashint_{\sqrt{m_K}K_{R_\ep}} \delta_{x_K}\otimes\delta_{\sigma_{m_K}\theta_\mu J_i}\,d\mu,\frac1{m_K|K_{R_\ep}|}\int_L \delta_{x_K}\otimes\delta_{\sigma_{m_K}\theta_\mu \j'}\,d\mu\) < C\ep.\end{equation}
This goes as follows:
(i)  Using Lemma~\ref{etavar} and \eqref{gamj},\eqref{aspratio},  we find that integrating on $\sqrt{m_K}K_{(1-\eta_0)R_\ep}\sm \Lji$  instead of $\sqrt{m_K}K_{R_\ep}$ and $L$ induces an error of $C\ep$.
(ii) From \eqref{proxitronque}, and \eqref{eta1} applied to $\theta_\mu J_i$ and  $\theta_\mu \j'$ we have $d_{\B}(\delta_{\theta_\mu J_i}, \delta_{\theta_\mu \j'}) < \ep$ and thus in view of Lemma~\ref{etaconv} we may replace $\theta_\mu J_i$ by $\theta_\mu \j'$ in the integral with  an error of $C\ep\lep$ at most. (iii) Using \eqref{aspratio}, \eqref{gamj} and Lemma~\ref{etavar} again, we may integrate back on $\sqrt{m_K}K_{R_\ep}$ and $L$ rather than on $K_{(1-\eta_0)R_\ep}\sm\Lji$, with an additional error of $C\ep$. this proves \eqref{claimtronque}.

Combining  \eqref{claimtronque} with \eqref{p3ter} and $d_{\B}(P,P^{(4)})< C\ep(\lep +1)$, using Lemma~\ref{etaconv} we find $d_{\B}(P,P^{(5)})< C\ep(\lep +1)$, where
\begin{equation}\label{PQ} P^{(5)} = \frac 1{|\E'|}\sum_{\substack{K\in\cK_n\\1\le i\le N_\ep\\ \tilde L\in\LKk}}\frac{1}{m_K}\int_{L} \delta_{x_K}\otimes\delta_{\sigma_{m_K}\theta_\mu \tilde J_{i,L}}\,d\mu= \frac 1{|\E'|}\sum_{\substack{K\in\cK_n\\1\le i\le N_\ep\\ \tilde L\in\LKk}}\int_{L/{\sqrt {m_K}}} \delta_{x_K}\otimes\delta_{\theta_\lambda \sigma_{m_K} \tilde J_{i,L}}\,d\lambda, \end{equation}
where the last equality follows by changing variables to $\lambda = \mu/\sqrt m_K$, and where $\tilde J_{i,L}$ denotes an arbitrarily chosen element of $\Lp$ such that  $\tilde J_{i,L} = J_{i,L}$ on $L$, the constant $C$ being independent of this choice.

If we choose an arbitrary $J_0$ in $\mathcal A_1$ and let the sum in \eqref{PQ} range over $0\le i\le N_\ep$ instead of $1\le i\le N_\ep$
 we obtain a measure $P^{(6)}$ such that, by \eqref{qep},
 $$\|P^{(5)} - P^{(6)}\|\le \frac 1{|\E'|}\sum_{K\in\cK_n} \frac{N_\ep |K|}{{q_\ep}^2}\le \eta_0,$$
 hence using Lemma~\ref{etavar} we have $d_{\B}(P^{(5)},P^{(6)})<\ep$ and then $d_{\B}(P,P^{(6)})< C\ep(\lep +1)$.

We now define the vector field $\jnint:\mr^2\to\mr^2$  by letting  $\jnint(x) = \sigma_{m_K} J_{i,L} (x-x_{\tilde L})$ on $\tilde L = x_{\tilde L} + L/\sqrt{ m_K}$, for every $K\in \cK_n$, $0\le i\le N_\ep$ and $\tilde L\in \LKk$. Then, for every $L\in\LKk$ we have $\jnint(x_{\tilde L}+\cdot) = \sigma_{m_K}J_{i,L}$ on $\tilde L$, therefore we may choose $\tilde J_{i,L} = \sigma_{1/m_K}\jnint(x_{\tilde L}+\cdot)$ in \eqref{PQ} and then then  we may summarize the above by writing
\begin{equation}\label{resume} d_{\B}(P,P^{(6)})< C\ep(\lep +1), \quad P^{(6)} = \frac 1{|\E'|}\sum_{K\in\cK_n}\int_{K} \delta_{x_K}\otimes\delta_{\theta_\lambda \jnint}\,d\lambda.\end{equation}
Note that since $J_{i,L} = 0$ outside $L$, we also have
\begin{equation}\label{jnint} \jnint=\sum_{\substack{K\in\cK_n\\1\le i\le N_\ep\\ \tilde L\in\LKk}}\sigma_{m_K} J_{i,L} (\cdot-x_{\tilde L}),\quad \div \jnint = 2\pi \sum_{\substack{K\in\cK_n\\p\in\Lambda_K}}(\delta_p - m_K),\end{equation}
where $\Lambda_K$ is a finite subset of the interior of $K$. The second equation is satisfied in the sense of distributions on $\mr^2$.\medskip

\noindent{\it Step 4:  Treating  the boundary.} Let $\hat \E' := \E'\sm\cup_{K\in\cK_n} K$. We let $t\in[0,\ell\sqrt n]$ denote arclength on $\p \E'$ --- where $\ell$ is the length of $\p \E$ --- and $s$ denote the distance to $\p \E'$, so that $(t,s)$ is a $C^1$ coordinate system on $\{x\in \E': d(x,(\E')^c)<c\sqrt n\}$,  if $c>0$ is small enough, since the boundary of $\E$  is $C^1$ by \eqref{assumpV3}. We let $C_t$ denote the curvilinear rectangle of points with coordinates in $[0,t]\times [0, C\Rb_\ep]$, where $\Rb_\ep = q_\ep R_\ep$ and $C$ is large enough so that $\hat \E'\subset \{x\in \E': d(x,\p\E')<C\Rb_\ep\}$, and define $m(t) = \int_{C_t\cap\hat \E'} \bm'$.
Since the distance of $\cup_{K\in\cK_n} K$ to a given $x\in\p \E'$ is between $\Rb_\ep$ and $C_0\Rb_\ep$ from Lemma~\ref{rect} and since  $\bm' $ is bounded above and below by \eqref{bornmo}, the derivative of $t\mapsto m(t)$ is between $\Rb_\ep/C$ and $C\Rb_\ep$ for some $C>0$ large enough.

We let \begin{equation}\label{keps} k_\ep = \left[\frac{\ell\sqrt n}{\Rb_\ep}\right]\end{equation} and  choose $0 = t_0,\dots,t_{k_\ep} =\ell\sqrt n$ to be such that
$$m(t_l) = \left[\frac l {k_\ep} m(\ell\sqrt n)\right] .$$
We note that indeed $t_{k_\ep} = \ell\sqrt n$~: Since the integral of $\bm'$ on each square $K\in\cK_n$ is an integer as well as the integral on $\E'$,  we have  $\int_{\hat \E'}\bm'\in\mn$ and therefore $m(\ell\sqrt n)\in\mn$.

From the above remark about the derivative of $t\to m(t)$, we deduce that  $\frac{m(\ell\sqrt n)}{\ell\sqrt n}$ belongs to the interval $[\Rb_\ep/C,C\Rb_\ep]$ for some $C>0$ and then it is easy to deduce that if $\sqrt n$ is large enough compared to $\Rb_\ep$ then
$$ n_l := m(t_{l+1}) - m(t_l)\in \left[{\Rb_\ep}^2/C, C{\Rb_\ep}^2\right], \quad  t_{l+1}- t_l  \in [\Rb_\ep/C,C\Rb_\ep]. $$
This means that the sidelengths of the curvilinear rectangle $C_{t_{l+1}}\sm C_{t_l}$ are comparable to $\Rb_\ep$, and that the number of points $n_l$ to put there  in is of order ${\Rb_\ep}^2$.

We may then include each of the sets $K_l :=\hat \E'\cap( C_{t_{l+1}}\sm C_{t_l})$ in a ball $B_l$ with radius in $[\Rb_\ep/C,C\Rb_\ep]$ and we may also choose a set of  $n_l$  points $\Lambda_l$  which are at distance at least $1/C$ from each other and the complement of $K_l$. Let $\j_l = -\nab H$, where $H$ solves $-\Delta H = 2\pi(\sum_{p\in\Lambda_l}\delta_p - m_l)$ in $B_l$ and $\nab H \cdot \vec{\nu} = 0$ on $\p B_l$, where
$$m_l= \frac{n_l}{|K_l|}\indic_{K_l}.$$
Then we have  $\div \j_l =2\pi(\sum_{p\in\Lambda_l}\delta_p - m_l)$ in $ B_l$ and $ \j_l \cdot \vec{\nu} = 0$ on $\p B_l$ and we claim that for any  $q\ge 1$,
\begin{equation}\label{wi} W(\j_l,\indic_{B_l}) \le C_{\ep},\quad \|\j_l\|_{L^q(B_l\sm K_l)} \le C_{\ep,q},\end{equation}
where the constants do not depend on $n$, but do depend on $\ep$ through $\Rb_\ep$. This is proved by noting that these quantities are finite, and that a compactness argument shows that the bound is uniform  for any choice of points which are at distance at least $1/C$ from each other and the complement of some $K_l\subset B_l$, using for instance the explicit formulas for $W$ in \cite{lr}.
 Note that because the sets $\{K_l\}$ and the rectangles $\{K\}$ are disjoint, have measure between ${\Rb_\ep}^2/C$ and $C{\Rb_\ep}^2$ and diameter between ${\Rb_\ep}/C$ and $C{\Rb_\ep}$, we know that their overlap is bounded by a constant $C$ independent of $\ep$, $n$. \medskip

\noindent{\it Step 5: Rectification of the weight.} We rectify the weights $m_K$, $m_l$:  For $K\in\cK_n$ we let  $H_K$ solve $-\Delta H_K = 2\pi (\bm' - m_K)$ on $K$ and $\nab  H_K\cdot \vec{\nu}  = 0$ on $\p K$. Similarly we let $H_l$ solve
$-\Delta H_l= 2\pi (\bm'\indic_{K_l}- m_l) $, $\nab H_l\cdot \vec{\nu} =0$. By elliptic regularity, we deduce for any $q>1$ that $\|\nab H_K\|_{L^q(K)}$  (resp. $\|\nab H_l\|_{L^q(B_l)}$) is bounded by $C_{q,\ep}  \|\bm' - m_K\|_{L^\infty(K)}$ (resp. $C_{q,\ep} \|\bm' - m_l\|_{L^\infty(B_l)}$.) Since $\bm$ is $C^1$ we have $|\nab \bm'|\le C/\sqrt n$, therefore $\|\bm' - m_K\|_{L^\infty(K)}\le C\Rb_\ep/\sqrt n$, while   $\|\bm' - m_l\|_{L^\infty(B_l)}\le C$. We deduce that
\begin{equation}\label{hik} \|\nab H_K\|_{L^q}\le \frac{C_{q,\ep}}{\sqrt n},\quad \|\nab H_l\|_{L^q}\le C_{q,\ep}.\end{equation}

We let
\begin{equation}\label{jk} \j_K = \jnint_{|K},\end{equation}
 and \begin{multline}\label{defjn} \j_n = \jnint + \sum_{K\in\cK_n} -\nab H_K  + \sum_{i=1}^{k_\ep}  - \nab H_l = \sum_{K \in \cK_n} \j_K - \nab H_K +\sum_{l=1}^{k_\ep} \j_l - \nab H_l,
 \\ \Lambda_n = \cup_{K\in\cK} \Lambda_K\cup_{l=1}^{k_\ep} \Lambda_l,\end{multline}
where $\j_K$ and $\nab H_K$ are set to $0$ outside $K$ and similarly for $\j_l$, $\nab H_l$ outside $B_l$. Then $\div \j_n = 2\pi(\sum_{p\in\Lambda_n}\delta_p - \bm')$ in $\mr^2$.
This completes the construction of $\j_n$.

\subsection{Estimating the energy}
{\it Step 1: Energy estimate.}
We have
$$W(\j_K,\indic_K) = \sum_{\substack{0\le i\le N_\ep\\ \tilde L\in\LKk}} W(\sigma_{m_K} J_{i,L} (\cdot-x_{\tilde L}),\tilde L).$$
From \eqref{restronque} we find, letting $L = \sqrt {m_K} (\tilde L - x_{\tilde L})$, using \eqref{sumnkk} and $|L| = |K|/{q_\ep}^2$, that
\begin{equation}\label{wkk}W(\j_K,\indic_K)= \sum_{\substack{0\le i\le N_\ep\\ \tilde L\in\LKk}} W(\sigma_{m_K} J_{i,L},\indic_{\tilde L - x_{\tilde L}}) \le |K|\( \sum_{i=0}^{N_\ep} \frac{n_{i,K}}{{q_\ep}^2} W(\sigma_{m_K} J_i) +C \ep\).\end{equation}

We estimate the integral of $|\j_n|^2$ on $\mr^2\sm \cup_{p\in\Lambda_n} B(p,\eta)$. From \eqref{defjn}, this integral involves on  the one hand the square terms
\begin{equation}\label{sqterms} \sum_{l=1}^{k_\ep} \int_{(B_i)_\eta}|\j_l - \nab H_l|^2+ \sum_{K\in\cK_n}\int_{K_\eta}|\j_K -\nab H_K|^2 ,\end{equation}
where $K_\eta = K\sm \cup_{p\in \Lambda_n} B(p,\eta)$ and similarly for $(B_l)_\eta$, and on the other hand the rectangle terms
$$ \sum_{\substack{K,K'\in\cK_n\\ K\neq K'}}\int_{K_\eta\cap K'_\eta} (\j_K -\nab H_K)\cdot(\j_{K'} - \nab  H_{K'}) +\sum_{1\le l\neq i\le k_\ep}\dots + \sum_{\substack{K\in\cK_n\\ 1\le l\le k_\ep}}\dots $$
We estimate the latter as follows: Since the rectangles in $\cK_n$ do not overlap, the first sum is equal to zero. A nonzero rectangle term must involve some $B_l$, and moreover a given $B_l$ can only be present in a  number of terms bounded independently of $n,\ep$ because the overlap of the balls $B_l$ and the rectangles $K$ is bounded. Thus from \eqref{keps} we have at most $C \sqrt n / \Rb_\ep$ nonzero rectangle terms. Moreover, since the $K_l$'s are disjoint, and disjoint from the $K$'s, in a rectangle term involving $B_l\cap K$ the integral can be taken over $K\sm K_l$, and in a term involving $B_l\cap B_i$ it can be taken over $(B_l\cap B_i\sm K_i) \cup (B_i\cap B_l\sm K_l)$.

In any case we use H\"older's inequality and the bound $\|\j_l-\nab H_l\|_{L^q(B_l\sm K_l)} \le C_{\ep,q}$ for some $q>2$, which follows from \eqref{wi}, \eqref{hik},  together with the bound
$$\|\j_l-\nab H_l\|_{L^{q'}(B_l)}, \ \|\j_K-\nab H_K\|_{L^{q'}(K)} \le C_{\ep,q},$$
which follows from \eqref{wkk}, \eqref{wi} using \ed{Lemma~\ref{lemnew}},  to conclude that each rectangle term is bounded by $C_{\ep}$  and then that their sum is $O(\sqrt n)$, meaning a quantity bounded by a constant depending on $\ep$  times $\sqrt n$.

The  limit as $\eta\to 0$ of the terms in \eqref{sqterms} is estimated as above  by expanding the squares and using H\"older's inequality with  \eqref{hik}, \eqref{wi}, \eqref{wkk}, together with the bound \eqref{keps} to show that
\begin{multline*}
\lim_{\eta\to 0} \frac12 \( \sum_{l=1}^{k_\ep} \int_{(B_l)_\eta}|\j_l - \nab H_l|^2+ \sum_{K\in\cK_n}\int_{K_\eta}|\j_K -\nab H_K|^2+\pi \#\Lambda_n \log \eta\) \\ \le  \sum_{K\in\cK_n} W(\j_K,\indic_K) + O(\sqrt n). \end{multline*}
In view of the bound $O(\sqrt n)$ for the rectangle terms and \eqref{wkk} we find using \eqref{sumnkk} that
\begin{equation}\label{caco} W(\j_n,\indic_{\mr^2}) \le \sum_{\substack{K\in\cK_n\\ 0\le i\le N_\ep}} |K|\frac{n_{i,K}}{{q_\ep}^2} W(\sigma_{m_K} J_i)  +  C n \ep + O(\sqrt n).\end{equation}

\noindent
{\it Step 2: We proceed to estimating $W(\j_n,\indic_{\mr^2})$.} We have, using \eqref{nkk}, \eqref{pkk}, \eqref{qep}, then  the fact that $\bm' - m_K\le C\Rb_\ep/\sqrt n$ on $K$,  then \eqref{wjk} with \eqref{minalai1}, then \eqref{wjk} and finally \eqref{proba}, that
\begin{align}
\sum_{i=1}^{N_\ep} \frac{|K| n_{i,K}}{{q_\ep}^2} & W(\sigma_{m_K} J_i)   \nonumber
\le  |\E'| \sum_{i=1}^{N_\ep} \tP\(\frac{K}{\sqrt n}\times H_\ep^i\) W(\sigma_{m_K} J_i)+ |K|\ep \nonumber\\
&\le   |\E'|  \sum_{i=1}^{N_\ep} \int_{\frac{K}{\sqrt n}\times H_\ep^i} W(\sigma_{\bm'(x)}J_i)\,d\tP(x,\j) + |K|\(\frac{C}{\sqrt n}+\ep \) \nonumber\\
&\le |\E'| \sum_{i=1}^{N_\ep} \int_{\frac{K}{\sqrt n}\times H_\ep^i} W(\sigma_{\bm'(x)}\j)\,d\tP(x,\j) + |K|\(\frac{C}{\sqrt n}+C\ep\) \nonumber\\
& \label{ee} =|\E'| \int_{\frac K{\sqrt n}\times H_\ep} W(\j)\,d P(x,\j) + |K|\( \frac{C}{\sqrt n}+C\ep\).
\end{align}
Here we have used the fact that $W$ is bounded below by some (negative) constant, a fact proved in \cite{ss1} that we use below several times.

We proceed by estimating $n_{0,K}$. From \eqref{nkk} we deduce that
$$ \sum_{i=1}^{N_\ep} (n_{i,K} +1) \ge \frac{{q_\ep}^2 |\E'|}{|K|} \tP\(\frac K{\sqrt n}\times H_\ep\) \ge \frac{{q_\ep}^2 |\E'|}{|K|} \(\frac{|K|}{|\E'|}-\tP\(\frac K{\sqrt n}\times {H_\ep}^c\)\),$$
and then it follows from \eqref{sumnkk} that
$$ n_{0,K}= {q_\ep}^2- \sum_{i=1}^{N_\ep} n_{i, K}   \le N_\ep +  \frac{{q_\ep}^2 |\E'|}{|K|} \tP\(\frac K{\sqrt n}\times {H_\ep}^c\).$$
Summing over $K\in\cK_n$, using the fact that
\begin{equation}\label{bord}|\E'\sm \cup_{\cK_n} K|< C_\ep \sqrt n\end{equation}
and then \eqref{qep}, \eqref{pgep}, we find that
$$ \sum_{K\in\cK}\frac{|K|}{{q_\ep}^2} n_{0,K}    W(\sigma_{m_K} J_0)
\le C |\E'|\(\tilde P(\E\times {H_\ep}^c) +\frac1{\sqrt n}+ \ep\) \le C n \(\frac{C_\ep}{\sqrt n}+\ep\).$$

Summing \eqref{ee} with respect to $K\in\cK_n$ and adding the above estimate we find, in view of \eqref{bord}, \eqref{caco}  and \eqref{pgep}, that
\begin{equation}\label{finup} W(\j_n,\indic_{\mr^2}) \le n |\E| \int_{\E\times L^p_\loc} W (\j)\,dP(x,\j) + Cn \(\ep + \frac{C_\ep}{\sqrt n}\).\end{equation}

\ed{Note that at this point if we had chosen $J_i$ such that $W(J_i)<\inf_{H^i_\ep} \mathbb{W} + \ep$, we obtain
$$W(\j_n, \indic_{\mr^2}) \le n |\E| \int_{\E \times \mathcal{M}_+
} \mathbb{W}(\nu) \, dQ(x, \nu) + C n \(\frac{C_\ep}{\sqrt n}+\ep\).$$
}
\noindent
{\it Step 3: Energy bound for $(x_1,\dots,x_n)$.} From \eqref{finup},  the constructed fields $\{\j_n\}$ and points $\{\Lambda_n\}_n$ satisfy $\div \j_n = 2\pi\(\sum_{p\in\Lambda_n}\delta_p - \bm'\)$ in $\mr^2$ with $\#\Lambda_n = n$ (cf. item (iii) in Proposition \ref{tronque}) and
\begin{equation}\label{wje}\limsup_n \frac{W(\j_n,\indic_{\mr^2})}{n} \le |\E|\int W(\j)\,dP(x,\j) +C\ep.\end{equation}

Now let $\{x_i\}_i = \{p/\sqrt n\}_{p\in\Lambda_n}$ be the points in $\Lambda_n$ in the initial scale,  and let still  $\nu_n = \sum_i\delta_{x_i}$.
The next step is to show that modifying $\j_n$ to make it curl-free can only decrease its energy. To see that, defining $H_n'$ by  \eqref{Hp}, we have that  $-\Delta H_n' = \div \j_n$ and we  may thus write  $\j_n = - \nab H_n'+ \np f_n$ for some function $f_n$.
 But $\j_n=0$ outside of $\E$, by construction, while $H_n'$ decays fast at infinity by its definition \eqref{Hp} and the fact that the right-hand side of \eqref{Hp} has integral $0$.
Letting $U_\eta=  \cup_{p \in \Lambda_n}  B(p, \eta) $,  we first have
\begin{multline*}\int_{B_R \backslash U_\eta } |\np f_n- \nab H_n'|^2 - \int_{B_R \backslash U_\eta}
 |\nab H_n'|^2  = - 2\int_{B_R \backslash U_\eta }
 \np f_n \cdot \nab H_n' + \int_{B_R \backslash U_\eta}
 |\nab f_n|^2\\ \ge -2\int_{B_R \backslash U_\eta }
 \np f_n \cdot \nab H_n' .\end{multline*}
Since $\j_n\in L^q_{loc}$ for any $q<2$ and since $f_n \in W^{1,q}_{loc} (\mr^2)$ for all $q$, the last  term on the right-hand side converges as $\eta\to 0$ to the integral over $B_R$. Also integrating by parts, using the Jacobian structure and the decay of $f_n$ and $H_n'$, we have  $\int_{B_R} \np f_n \cdot \nab H_n' \to 0$ as $R \to +\infty$.  Therefore, letting  $\eta\to 0$  then $R\to + \infty$ in the above yields
   $$ W(\j_n,\indic_{\mr^2}) - W( \nab H_n',\indic_{\mr^2})   \ge 0.$$

Since $\Lambda_n\subset \E'$ by construction, we have $\supp(\nu_n)\subset \E$ and thus $\int \zeta\,d\nu_n = 0$.
Together with \eqref{wje},  we deduce in view of \eqref{idwn}  that
\begin{equation}\label{wjef}\limsup_{n \to \infty}\frac{1}{n} \(\w(x_1, \dots, x_n) - n^2 \I (\mo)   + \frac{n}{2}\log n \)  \le \frac{|\E|}{\pi}\int W(\j) \, dP(x,\j)+C\ep.\end{equation}\medskip
(\ed{Respectively $\le  \frac{|\E|}{\pi}\int \mathbb{W}(\nu) \, dQ(x,\nu)+C\ep$.)}

\noindent{\it Step 4: Existence of  $A_n$.}
We claim  that if $n$ is large enough and if $\j\in\Lp$ is such that
\begin{equation}\label{djnint}d_p(\j(\sqrt n x+\cdot), \jnint(\sqrt n x+\cdot))<\eta_1/2\end{equation}
for any $x\in \E\sm \Xi$ for some set $\Xi$  satisfying $|\Xi|<\eta_0|\E|$, then
\begin{equation}\label{distmes} d_{\B}\(\dashint_{\E}\delta_{x}\otimes\delta_{\theta_{\sqrt n x} \j}\,dx, \dashint_{\E}\delta_{x}\otimes\delta_{\theta_{\sqrt n x} \jnint}\,dx\) < C\ep(\lep +1).\end{equation}
This would follow immediately from Lemmas~\ref{etavar},~\ref{etadir} and~\ref{etaconv} if $\theta_{\sqrt n x}\jnint$ belonged to some compact set independent of $x\notin \Xi$ and $n$. In our case we note that if $x$ belongs to some $\tilde L\in \LKk$, where $K\in\cK_n$ and $0\le i\le N_\ep$, then
$$\jnint(x+\cdot) = \sigma_{m_K} J_{i,L}(\cdot + x - x_{\tilde L}). $$
Moreover,  since $J_i\in H_\ep$, from \eqref{defhep} it follows that
if $x-x_{\tilde L}\notin\Gamma(J_i)/\sqrt{m_K}$ then $\j':=\sigma_{m_K} J_{i}(\cdot + x - x_{\tilde L})\in G_\ep$. If in addition, $\dist(x,\p \tilde L)> \eta_0 R_\ep$, then  we deduce from \eqref{proxitronque} that
$d_p\(\jnint(x+\cdot),\j'   \)<\eta_1/2$ and
$d_p\(\j(x+\cdot), \j'  \)<\eta_1$.
   Lemma~\ref{etadir}  then yields $d_{\B}(\delta_{\jnint(x+\cdot)}, \delta_{\j'}  ) <\ep$
  and $d_{\B}(\delta_{\j(x+\cdot)},\delta_{\j'}  )     <\ep$
  thus
  $$d_{\B}(\delta_{\jnint(x+\cdot)},
  \delta_{\j(x+\cdot)})<2\ep.$$
  In view of Lemma~\ref{etaconv} we find
$$ d_{\B}\(\frac1{|\E|}\int_{\E\sm \tilde \Xi}\delta_{x}\otimes\delta_{\theta_{\sqrt n x} \j}\,dx, \frac1{|\E|}\int_{\E\sm \tilde \Xi}\delta_{x}\otimes\delta_{\theta_{\sqrt n x} \jnint}\,dx\) < C\ep(\lep +1),$$
where $\tilde \Xi$ is the union of $\Xi$ and of the union with respect to $0\le i\le N_\ep$, $K\in\cK_n$ and $\tilde L\in \LKk$  of $\frac{1}{\sqrt{n}}\(x_{\tilde L}+\Gamma(J_i)/\sqrt{m_K}\)$,  of $\frac{1}{\sqrt{n}}\{x\in\tilde L: \dist(x,\p \tilde L)\le  \eta_0 R_\ep\}$, and of $\E\sm\cup_{\cK_n }\frac{K}{\sqrt{n}}$. It turns out that $|\tilde \Xi|< C \eta_0$ if $n$ is large enough, $C$ being of course independent of $\ep$, and thus using Lemma~\ref{etavar} we deduce \eqref{distmes}. The claim is proved.

To prove the existence of the set $A_n$, we note that the currents $J_{i}$ used in constructing $\jnint$ depend on $\ep$ but are independent on $n$. Then they are truncated to obtain $J_{i,K}$ where the sidelengths of $L$ are in $[R_\ep/C,C R_\ep]$, i.e. in an interval independent of $n$. It follows at once that there exists $\delta>0$ such that  the points in $L$ may be perturbed by an amount $\delta$ so that for every $i$, $K$ and $\tilde L\in\LKk$ the perturbed $J_{i,L}^{pert}$ is at a distance  at most $\eta_1/4$ of $J_{i,K}$, for every $n$. Then in view of \eqref{bord} and \eqref{hik} it follows that for $n$ large enough the resulting $\j_n^{pert}$ will satisfy \eqref{djnint} for $x$ far enough from $\p \E'$, i.e. outside a set of proportion  relative to $|\E'|$ tending to $0$ as $n\to \infty$. We deduce that $\j_n^{pert}$ satisfies \eqref{distmes}, hence if $n$ is large enough
$$d_{\B}(P_{\j_n^{pert}}, P)< C\ep(\lep +1).$$
The same reasoning implies that if we let $\{x_1,\dots,x_n\}$ be the points in $\Lambda_n$ in original coordinates, then perturbing the points in $\Lambda_n$  by an amount $\delta>0$ small enough, i.e. perturbing the $x_i$'s by an amount $\delta/\sqrt n$ at most we obtain points $y_i$ such that $w_n(y_i)\le w_n(x_i) + \ep$. Since the ordering of the points is irrelevant, we let $S_n$ denote the set of permutations of $1\dots n$ and define
$$A_n = \{(y_1,\dots,y_n): \exists\, \sigma\in S_n,\,  |x_i - y_{\sigma(i)}|<\delta.\}$$
Then, given $\eta>0$,  from the previous discussion and choosing $\ep>0$ small enough we have for any $n$ and any $(y_1,\dots,y_n)\in A_n$ that \eqref{bsw} is satisfied and  the existence of $\j_n$ such that $\div \j_n = 2\pi \(\sum_i \delta_{{y_i}'} - \bm'\)$ and such that $\{P_{\j_n}\}_n$ satisfies \eqref{convpn}.

This concludes the proof of Proposition~\ref{construct}.

\noindent
{\sc Etienne Sandier}\\
Universit\'e Paris-Est,\\
LAMA -- CNRS UMR 8050,\\
61, Avenue du Général de Gaulle, 94010 Créteil. France\\
\& Institut Universitaire de France\\
{\tt sandier@u-pec.fr}\\ \\
{\sc Sylvia Serfaty}\\
UPMC Univ  Paris 06, UMR 7598 Laboratoire Jacques-Louis Lions,\\
 Paris, F-75005 France ;\\
 CNRS, UMR 7598 LJLL, Paris, F-75005 France \\
 \&  Courant Institute, New York University\\
251 Mercer st, NY NY 10012, USA\\
{\tt serfaty@ann.jussieu.fr}

\end{document}